\documentclass{sig-alternate}

\usepackage[update,prepend]{epstopdf} 
\usepackage{verbatim}
\usepackage{graphics}
\usepackage{graphicx}
\usepackage{color}
\usepackage[rgb,usenames,dvipsnames,svgnames,table]{xcolor}
\usepackage[ruled,vlined]{algorithm2e}
\usepackage{complexity}
\usepackage{multirow}
\usepackage{tabularx}
\usepackage{datetime}

\usepackage{amsthm}
\usepackage{amsmath}
\usepackage{tikz}
\usepackage{tikz-qtree}
\usetikzlibrary{arrows,matrix,positioning}
\usetikzlibrary{shapes,snakes}
\usepackage{caption}
\usepackage{subcaption}
\usepackage[normalem]{ulem}
\usepackage{accents}

\makeatletter  
 \let\@copyrightspace\relax  
\makeatother
\begin{document}

\newcommand\eat[1]{}
\newcommand\scream[1]{*** #1 ***}
\newcommand\reminder[1]{*** #1 ***}

\newcommand{\sue}[1]{[\textcolor{red} {{\bf Susan: }#1}]}
\newcommand{\xiaocheng}[1]{[\textcolor{Mulberry}{{\bf Xiaocheng: }#1}]}
\newcommand{\julia}[1]{[\textcolor{blue}{{\bf Julia: }#1}]}

\newcommand\kw[1]{kw(#1)}
\newtheorem{theorem}{Theorem}
\newtheorem{example}{Example}[section]
\theoremstyle{definition}
\newtheorem{definition}{Definition}[section]

\newtheorem{lemma}{Lemma}[section]
\newtheorem{proposition}{Proposition}[section]
\newtheorem{corollary}[theorem]{Corollary}
\newtheorem{condition}{Conditions}
\newtheorem{remark}{Remark}[section]

\def\myExp{{\em myExperiment.org}\xspace}
\def\rpt{{\em rpTree}\xspace}
\def\rpts{{\em rpTree}s\xspace}
\newcommand*{\dt}[1]{%
  \accentset{\mbox{\large\bfseries .}}{#1}}
\newcommand\sprec{\dt{\prec}}
\newcommand\nsprec{\dt{\nprec}}

\def\G{{$G=(\Sigma,\Delta, S,R)$}\xspace}

\title{ Search and Result Presentation in Scientific Workflow Repositories}
\eat{\subtitle{[Extended Abstract]
\titlenote{A full version of this paper is available as
\textit{Author's Guide to Preparing ACM SIG Proceedings Using
\LaTeX$2_\epsilon$\ and BibTeX} at
\texttt{www.acm.org/eaddress.htm}}}}
%
%
%
%
%

\numberofauthors{1}
\newcommand{\superscript}[1]{\ensuremath{^{\textrm{#1}}}}
\def\wa{\superscript{1}}
\def\wb{\superscript{2}}
\def\wc{\superscript{3}}

\def\sharedaffiliation{\end{tabular}\newline\begin{tabular}{c}}
\author{
  \alignauthor{  Susan B. Davidson\wa~ Xiaocheng Huang\wb~ Julia Stoyanovich\wc~ Xiaojie Yuan\wb}
    \sharedaffiliation
\begin{small}
  \begin{tabular}{ccc}
    \affaddr{{\wa}University of Pennsylvania} &\affaddr{{\wb}Nankai University}&\affaddr{{\wc}Drexel University}\\
    \affaddr{Philadelphia, USA} &\affaddr{Tianjin, China}&\affaddr{Philadelphia, USA}\\
    \affaddr{susan@cis.upenn.edu} &\affaddr{huangx@seas.upenn.edu}&  \affaddr{stoyanovich@drexel.edu} \\
       &\affaddr{yuanxj@nankai.edu.cn}&
  \end{tabular}
\end{small}
}

\maketitle \thispagestyle{empty}\pagestyle{empty}
\pagenumbering{arabic}
\begin{abstract}

  We study the problem of searching a repository of complex
  hierarchical workflows whose component modules, both composite and
  atomic, have been annotated with keywords. Since keyword search does
  not use the graph structure of a workflow, we develop a model of
  workflows using {\em context-free bag grammars}. We then give
  efficient polynomial-time algorithms that, given a workflow and a
  keyword query, determine whether some execution of the workflow {\em
    matches} the query.  Based on these algorithms we develop a {\em
    search and ranking} solution that efficiently retrieves the
  top-$k$ grammars from a repository.  \eat{ We then give an
    efficient, polynomial time algorithm that, given a workflow and a
    keyword query, determines whether or not some execution of the
    workflow {\em matches} the query, and use it for {\em searching} a
    workflow repository.  Further, building on probabilistic
    context-free grammars, we develop efficient algorithms for
    calculating the {\em relevance score} of a grammar for a given
    query, and use it in scope of a top-$k$ solution to find grammars
    with highest relevance to a query. }Finally, we propose a novel
  result presentation method for grammars matching a keyword query,
  based on {\em representative parse-trees}.  The effectiveness of our
  approach is validated through an extensive experimental evaluation.

\end{abstract}



\eat{
}

\section{Introduction}
\label{sec:intro}
Data-intensive workflows are gaining popularity in the scientific
community.  Workflow repositories are emerging in support of sharing
and reuse, either as part of a particular workflow system (e.g.,
VisTrails~\cite{vistrails} or Taverna~\cite{taverna}) or independently
within a particular community (e.g., \myExp~\cite{myexperiment}).  As
workflows become more widely used, workflow repositories grow in size,
making information discovery an interesting challenge. 

Current workflow repositories, e.g., \myExp, support tagging of
workflows with keywords.  Notably, because workflows are modular,
users may wish to share and reuse {\em components} of a
workflow \cite{DBLP:conf/ssdbm/StarlingerBL12}.  It is thus important
to support tagging, and to enable search, not just at the level of a
workflow, but also at the level of modules and subworkflows.

Recent work considered search in workflow
repositories~\cite{DBLP:conf/amw/DavidsonLS11,Liu:2010:SWH:1920841.1920958,4812555},
and also argued that, because workflows can be large and complex, it
is important to provide usable result presentation mechanisms.  In
this paper we propose a novel search and result presentation approach
for complex hierarchical workflows.  We now illustrate our approach
with an example.


\begin{figure}[ht!]
  \centering
\includegraphics[width=0.4\textwidth]{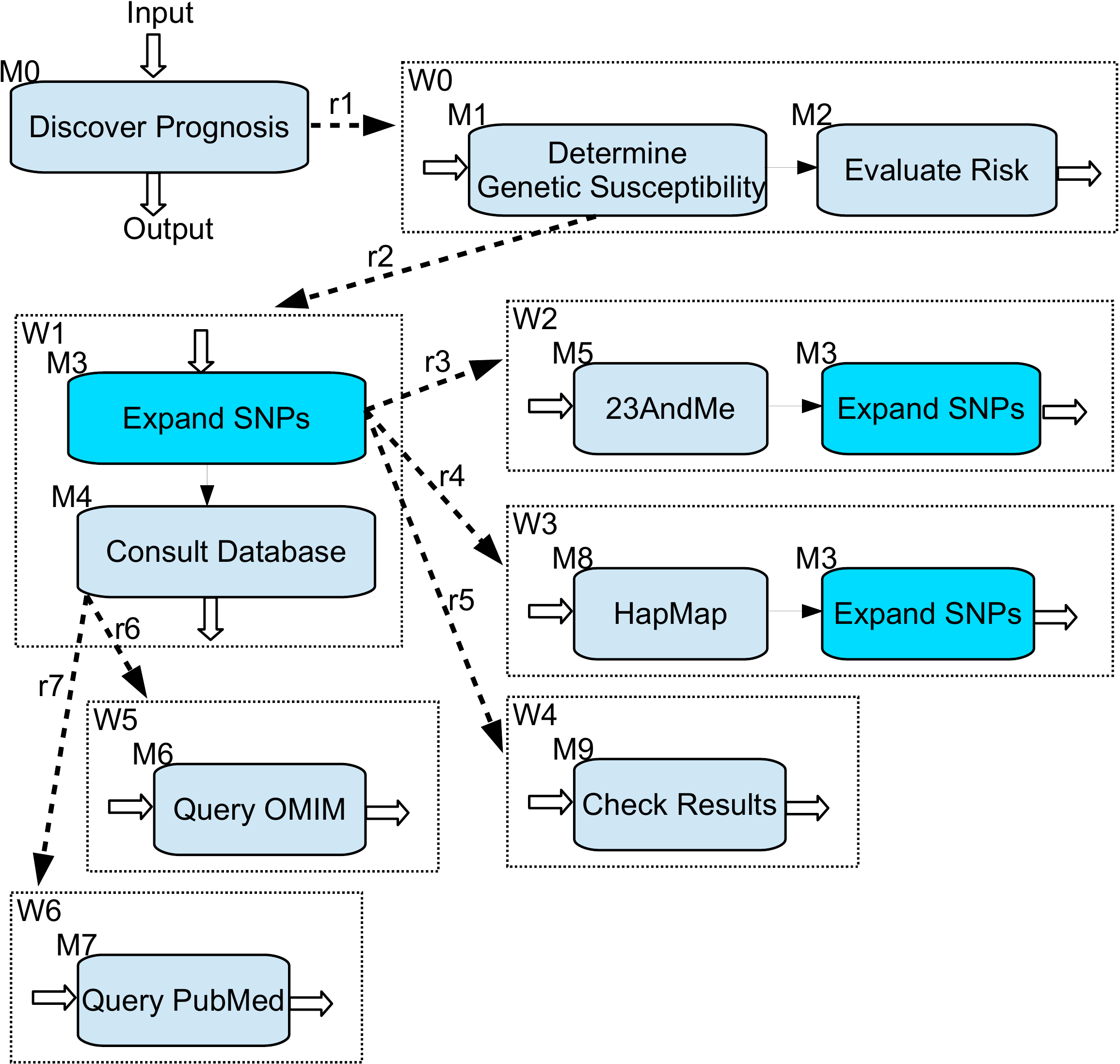}
   \caption{Disease succeptibility workflow.}
   \label{fig:graphWorkflow}
\end{figure}

\begin{figure}[h]
\begin{center}
 \begin{tabular}{l}
$r_1:M_0 \Rightarrow \{ M_1, M_2 \}$~~~$r_2:  M_1 \Rightarrow \{ M_3, M4 \}$\\
$r_3: M_3 \Rightarrow  \{M_5, M_3 \}$ $r_4: M_3\Rightarrow \{M_8, M_3\}$ $r_5: M_3  \Rightarrow \{M_9\}$\\
$r_6: M_4\Rightarrow \{ \text{lookup}, M_6\}$~~~ $r_7:M_4\Rightarrow \{\text{lookup},M_7\}$\\
$r_8:M_2\Rightarrow \{ \text{evaluate} \}$~~~~~~~ $r_9:M_5\Rightarrow \{ \text{23andMe} \}$\\
$r_{10}: M_9\Rightarrow \{\text{check}\}$ ~~~~~~~~ $r_{11}: M_6\Rightarrow \{\text{OMIM}\}$\\
$r_{12}: M_7\Rightarrow\{\text{PubMed}\}$~~~~~ $r_{13}: M_{8}\Rightarrow\{\text{HapMap}\}$
\end{tabular}
\end{center}
\caption{Disease succeptibility workflow as a bag grammar.}
\vspace{-0.2cm}
\label{fig:grammar}
\end{figure}

Consider a workflow in Figure~\ref{fig:graphWorkflow} that computes
succeptibility of an individual to genetic disorders, and is based
closely on~\cite{DBLP:journals/bioinformatics/StoyanovichP08}.  This
workflow takes a person's genetic information in the form of single
nucleotide polymorphisms (SNPs) as input, and produces an assessment
of genetic disorder risk.  The workflow is implemented by module
$M_0$, which is {\em composite}, and, when invoked, executes modules
$M_1$ and $M_2$ in sequence.  Module $M_1$ expands the set of SNPs by
considering known associations between SNP pairs and triplets.  This
module is composite, and has three {\em alternative} executions.  In
the first, the SNP set is expanded using a proprietary database of
associations, e.g., 23andMe (module $M_5$), followed by a {\em
recursive} call to $M_3$.  In the second alternative, a public
association database, e.g., HapMap (module $M_8$) is used, followed by
a recursive call to $M_3$.  The final alternative involves checking
the retrieved results and terminating the recursion (module
$M_9$). Having computed an expanded SNP set, the workflow goes on to
look up any genetic disorders associated with the SNPs.  This is
implemented by module $M_4$, which has two alternative executions: it
may issue a query to OMIM (module $M_6$) or to PubMed (module
$M_7$). Having retrieved results from OMIM or from PubMed, the
workflow terminates.

Suppose that this workflow exists in a repository, and that some of
its modules are tagged.  Let us assume the following assignment of
keywords to workflow modules: $M_2$ (evaluate), $M_4$(lookup), $M_5$
(23andMe), $M_6$ (OMIM), $M_7$ (PubMed), $M_8$ (HapMap), and $M_9$
(check).  It is not required that all modules be tagged, e.g., there
is no keyword assigned to $M_1$ in our example. It is also possible,
and even likely, that multiple keywords are assigned per module, and
that keywords are reused across modules, and across
workflows~\cite{DBLP:conf/ssdbm/StarlingerBL12}.  However, we do not
use multiple or repeating keywords here, to simplify our example.

Suppose now that a user wants to know whether the workflow in
Figure~\ref{fig:graphWorkflow} matches a particular keyword query.
Assuming ``and'' query semantics, answering this question amounts to
determining if {\em there exists some execution of the workflow in
  which all query keywords are present}.  For example, query
$\{23andMe, HapMap\}$ matches an execution in which module $M_3$ is
run twice, evaluating $M_5 \rightarrow M_3$ (23andMe) on the first
invocation, and $M_8 \rightarrow M_3$ (HapMap) on the second
invocation.  Intuitively, this execution exists because of the
combination of alternation and recursion at $M_3$.

On the other hand, there is no execution that matches
$\{OMIM,PubMed\}$ because, once an alternative for expanding $M_4$ is
chosen, then $M_6$ (OMIM) or $M_7$ (PubMed) is executed, and there is
no recursion that allows $M_4$ to repeat, possibly choosing another
branch.

It has recently been shown that complex hierarchical workflows can be
naturally represented as context-free graph grammars involving
recursion and alternation~\cite{DBLP:conf/sigmod/BaoDM11,
  DBLP:conf/vldb/BeeriEKM06}.  We build on this work and adapt it to
keyword search in workflows with tagged modules.  Because of our
proposed query semantics, observe that, while hierarchical workflow
structure, alternation and recursion are important for determining
whether a workflow matches a query, the graph structure within each
composite module is unimportant for our purposes.  This observation
leads us to model scientific workflows as {\em context-free bag
  grammars} (also called {\em commutative
  grammars}~\cite{Esparza97petrinets}).

Figure~\ref{fig:grammar} represents a bag grammar corresponding to the
workflow in Figure~\ref{fig:graphWorkflow}.  The bag grammar captures
the hierarchical structure of the workflow (expansion of composite
modules), alternation and recursion. Importantly, the grammar makes
assignment of keywords to modules explicit, by including keywords as
terminals.  Note that keywords may annotate both atomic and {\em
  composite modules}, appearing in the corresponding grammar
productions.  So, $M_4$ is tagged with {\em lookup}, which is
captured in productions $r_6$ and $r_7$.

Query $\{23andMe, HapMap\}$ matches the workflow in
Figure~\ref{fig:graphWorkflow}, and we would now like to explain to
the user how the match occurs.  Let us now return to our example, and
focus on {\em result presentation}.  Providing a usable result
presentation mechanism is important, because workflow specifications
can be large, and each workflow can match a query in multiple ways,
due in large part to recursion and alternation.  We propose here a
result presentation mechanism based on a novel notion of {\em
  representative parse trees} (\rpt for short).
Figure~\ref{fig:rpTree_Q1} shows a particular \rpt for query
$\{23andMe, HapMap\}$, with nodes representing bag grammar productions
and terminals (see Figure~\ref{fig:grammar}), and edges corresponding
to a firing of a production.  Keyword matches occur at the leaves.

Intuitively, an \rpt represents a class of parse trees of a bag
grammar that derive a particular set of terminals.  An \rpt is an {\em
irredundant} representative of its class, in the sense that it does
not fire recursive productions that do not derive additional
query-relevant terminals.  For example, the \rpt in
Figure~\ref{fig:rpTree_Q1} represents also a tree in which production
$r_3$ is fired recursively twice, both times followed by $r_9$, and
thus generating the terminal {\em 23andMe} twice.  We will make the
sense in which an \rpt is an irredundant representative of its class
more precise in Section~\ref{sec:resultPresentation}, and will show
how \rpts can be derived efficiently.

\begin{figure}[t]
\begin{tikzpicture}[node distance=20]
\node(r1) {$r_1(M_0)$};
\node(under r1) [below of=r1] {};
\node(r2) [left of=under r1] {$r_2(M_1)$}; \node(r8) [right of=under r1, node distance=85] {$r_8(M_2)$}; \node(eval) [below of=r8] {evaluate};
\node(under r2)[below of=r2] {};
\node(r3)[ left of=under r2] {$r_3(M_3)$}; \node(r6)[right of=under r2, node distance=70] {$r_6(M_4)$}; 
\node(under r6)[below of=r6]{}; \node(lookup)[left of=under r6,node distance=15]{lookup}; \node(r11)[right of=under r6]{$r_{11}(M_6)$}; \node(OMIM)[below of=r11]{OMIM};
\node(under r3) [below of=r3] {};
\node(r4)[left of=under r3] {$r_4(M_3)$}; \node(r9)[right of=under r3, node distance=40]{$r_{9}(M_5)$};\node(23andMe)[below of=r9]{\bf 23andMe};
\node(under r4)[below of=r4]{};
\node(r5)[left of=under r4]{$r_5(M_3)$}; \node(r13)[right of=under r4,node distance=20]{$r_{13}(M_8)$}; \node(HapMap)[below of=r13]{\bf HapMap};
\node(r10)[below of=r5] {$r_{10}(M_{9})$};
\node(check)[below of=r10]{check};

\draw[-] (r1) -- (r2);\draw[-] (r1) -- (r8);
\draw[-] (r2) -- (r3); \draw[-] (r2) -- (r6); \draw[-] (r8) -- (eval);
\draw[-] (r3) -- (r4); \draw[-] (r3) -- (r9); \draw[-] (r6) -- (r11); \draw[-](r6)--(lookup);\draw[-](r11)--(OMIM);
\draw[-] (r4) -- (r5); \draw[-] (r4) -- (r13); \draw[-] (r9) -- (23andMe);
\draw[-] (r5) -- (r10); \draw[-] (r13) -- (HapMap);
\draw[-] (r10) -- (check);
\end{tikzpicture}
\caption{An \rpt for query $\{23andMe,HapMap\}$.}
\label{fig:rpTree_Q1}
\end{figure}
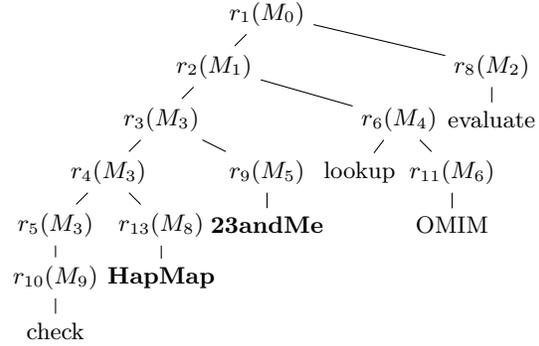

\eat{

\begin{figure}[t]
\begin{tikzpicture}[node distance=20]
\node(r1) {$r_1(M_0)$};
\node(under r1) [below of=r1] {};
\node(r2) [left of=under r1, node distance=55] {$r_2(M_1)$}; \node(r8) [right of=under r1] {$r_8(M_2)$}; \node(eval) [below of=r8] {evaluate};
\node(under r2)[below of=r2] {};
\node(r4)[ left of=under r2] {$r_4(M_3)$}; \node(r6)[right of=under r2, node distance=40] {$r_6(M_4)$}; 
\node(under r4) [below of=r4] {};
\node(r5)[left of=under r4] {$r_5(M_3)$}; \node(r13)[right of=under r4]{$r_{13}(M_8)$};\node(HapMap)[below of=r13]{\bf HapMap};
\node(r11)[below of=r6] {$r_{11}(M_6)$}; \node(OMIM)[below of=r11]{\bf OMIM};
\node(r10)[below of=r5]{$r_{10}(M_9)$};\node(threshold)[below of=r10]{threshold};

\draw[-] (r1) -- (r2);\draw[-] (r1) -- (r8);
\draw[-] (r2) -- (r4); \draw[-] (r2) -- (r6); \draw[-] (r8) -- (eval);\draw[-](r6)--(r11);
\draw[-] (r4) -- (r5); \draw[-] (r4) -- (r13); \draw[-] (r13) -- (HapMap);\draw[-] (r11)--(OMIM);
\draw[-](r5)--(r10); \draw[-] (r10) -- (threshold);
\end{tikzpicture}
\caption{An \rpt for query $\{HapMap, OMIM\}$.}
\end{figure}
}

{\bf Contributions.}  The contributions of our paper include:
\begin{itemize}

\item {\em A model of search and result presentation over
  keyword-annotated context-free bag grammars}
  (Section~\ref{sec:model}): Although the motivation for our model is
  derived from the problem of searching workflow repositories, it is
  applicable to any other scenario involving search over context-free
  bag grammars, e.g., business processes, call structure of programs,
  or other hierarchical graph applications.

\item {\em Search and ranking algorithms} (Sections~\ref{sec:algo}
  and~\ref{sec:ranking}): We give a bottom-up {\em match} algorithm,
  and develop an optimization, which borrows ideas from semi-naive
  datalog evaluation to avoid unproductive calculations.  Next,
  translating {\em probabilistic context-free grammars} to our
  setting, we develop efficient {\em search} and {\em ranking}
  algorithms, and use them to identify top-$k$ grammars.

\item {\em Novel result presentation methods}
  (Section~\ref{sec:resultPresentation}): Since a workflow may match a
  query in many different ways, we develop a presentation mechanism to
  help the user understand how the keyword matches are most likely to
  occur.  The mechanism is based on a novel notion of {\em
  representative parse trees}, which are the most probable parse trees
  that are structurally irredundant.

\item {\em Extensive experimental evaluations}
  (Section~\ref{sec:experiment}): We use synthetic datasets to
  demonstrate the effectiveness of our approach.  Our synthetic data
  is generated in a way that resembles characteristics of workflows in
  \myExp, in terms of keyword assignment and workflow size.  However,
  since current scientific workflow management systems do not yet
  allow that workflows be expressed as grammars, we are unable to use
  the \myExp dataset directly in our experiments.
\end{itemize}

{\bf Related Work.}
Much effort \cite{DBLP:conf/ssdbm/ChiuHKA11,
DBLP:conf/ssdbm/GandaraCSWSC11, Stoyanovich:2010:ERS:1833398.1833405}
has been made recently to annotate scientific workflows to enable
keyword search. As observed in \cite{Liu:2010:SWH:1920841.1920958},
since scientific workflows are usually modeled as a three-dimensional
graph structure when considering the expansions of composite modules
(dashed edges in Figure~\ref{fig:graphWorkflow}), results on searching
relational and XML data \cite{Chen:2009:KSS:1559845.1559966,
Liu:2011:PKS:2036574.2036580, Yu_keywordsearch} or graph data
\cite{Dalvi:2008:KSE:1453856.1453982, He:2007:BRK:1247480.1247516,
Kacholia:2005:BEK:1083592.1083652, Tran:2009:TEQ:1546683.1547416} can
not be easily
extended. \cite{DBLP:conf/amw/DavidsonLS11,Liu:2010:SWH:1920841.1920958,4812555}
consider the scenario when alternation or recursion is not present in
workflows. \cite{5767875,Nierman02protdb:probabilistic} consider the
scenario where nodes of XML documents exist with a probability
(analogous to choice), however there is no recursion.\eat{
\cite{DBLP:conf/ssdbm/ChiuHKA11,DBLP:conf/ideas/ChiuHKA11} use an
information retrieval approach to compose a workflow when users issue
a keyword query.}

\eat{
{\em Context-free string grammars (CFSG)} have also been well
studied. The \emph{CYK} algorithm provides an efficient membership
test for CFSGs \cite{DBLP:books/daglib/0011126}. The substring parsing
problem has been shown to be solvable in
PTIME~\cite{Rekers:1991:SPA:122501.122505}, based on the fact that a
string is a concatenation of two \emph{smaller} substrings, and there
is a polynomial number of such substring pairs.  These results do not
immediately extend to bags.
}

There are also extensive results in (context-sensitive) {\em
commutative
grammars}~\cite{Esparza94decidabilityissues,Esparza97petrinets}, which
have been cast as vector addition systems \cite{Karp1969147} and Petri
nets \cite{Crespi-Reghizz74}. Decidability issues of Petri nets are
surveyed in~\cite{Esparza94decidabilityissues}, and are shown to be
$\mathcal{NP} \mbox{-} complete$, directing focus towards more
specific problems.  We focus on a novel sub-problem,i.e., on whether
there exists a bag that contains a keyword set and is accepted by a
commutative grammar.

Most related to our work on matching and ranking is
\cite{DBLP:journals/jcss/DeutchM12}, which considers the more general
problem of querying the structure of a specification using graph
patterns\eat{; they also consider temporal queries on the potential
behavior of the defined process}.  The paper gives a query evaluation
algorithm of polynomial data complexity; the authors also consider the
probability of the match in~\cite{DBLP:journals/pvldb/DeutchMPY10}.
By considering each permutation of the keyword query $Q$ as a simple
graph pattern, where each node represents a keyword and nodes form a
chain connected by transitive edges, and taking the union of matching
specifications for each permutation, these results could be used in
our setting.  However, our matching algorithm is optimized for queries
that are sets of keywords, and is therefore simpler and considerably
more efficient than \cite{DBLP:journals/jcss/DeutchM12} for this
setting (see results in Section~\ref{sec:experiment}).  Importantly,
unlike~\cite{DBLP:journals/jcss/DeutchM12}, our solution does not
require that the input grammar be transformed for each query.  For
this reason, our solution can be tailored to present results using
representative parse trees (Section~\ref{sec:resultPresentation}),
while the solution of~\cite{DBLP:journals/jcss/DeutchM12} cannot.

Also closely related to our match algorithm is~\cite{1964}, which
gives a polynomial time algorithm for checking if the intersection of
a context-free string grammar (which could represent the workflow
specification) and a finite automaton (which could represent the
query) results in an empty grammar; however, our match algorithm has a
much better average case performance since it can terminate early if a
match is found.

\section{Model}
\label{sec:model}
 
In this section, we give background and definitions that will be used throughout  the paper.  We start with the definition of a context-free bag grammar, its language, and what 
it means for a search query with ``and''  semantics to match a grammar.  
We then introduce parse trees and derivation sequences.
\eat{We then define keyword-annotated graph grammars, give the translation to context-free bag grammars, and define the notion of a {\em repository} of keyword-annotated graph grammars.}
Finally, we define the notion of a repository of context-free bag grammars.

\label{sec:model:workflow}

\begin{definition}
\textbf{(Context-free Bag Grammar)} \\
 A \emph{context-free bag grammar} is a grammar $G=(\Sigma, \Delta, S, R)$ where $\Sigma$ is the set of symbols (variables and terminals), $\Delta \subset \Sigma$ is the set of terminals, $S$ is the start variable and $R$ is the set of production rules. For production $r \in R$, we denote by $head(r)\in \Sigma/\Delta$ the head of the production and  $body(r)$ the bag of symbols in the body of the production. The language of the grammar, $L(G)$, is a set of bags whose elements are in $\Delta$:
$L(G)= \; \{w \in \Delta^*| S \stackrel{*}{\Longrightarrow} w\} $
\end{definition}

We use context-free bag grammars to represent keyword-annotated
scientific workflows, of the kind described in
Figure~\ref{fig:graphWorkflow}, and with the corresponding grammar
given in Figure~\ref{fig:grammar}.  This grammar was derived by
replacing each composite module (variable) with production rules that
emit their keywords (as in rules $r_6$ and $r_7$ for module $M_4$),
and adding a production rule for each atomic module (terminal) that
emits its keyword (as in rules $r_8$ through $r_{13}$).

In the remainder of the paper, we will refer to a context-free bag
grammar simply as a grammar, and will use $\\productions(M)\subseteq
R$ to denote the set of productions with $M$ in the head. For clarity,
we also label productions.

\begin{definition}
\textbf{(Match)} Given a grammar \\
$G=(\Sigma, \Delta, S, R)$ and a keyword query $Q\subseteq \Delta$, we say that $G$  \emph{matches} $Q$ iff there exists some $X \in L(G)$ such that $Q
\subseteq X$.
\end{definition}

\begin{example}\label{ex:ex1}
Consider the grammar $\\G=(\{S,A,B,C,s_1,s_2,b,c\}, \{s_1,s_2,b,c\}, S, R)$, where $R$ is:
\vspace{-0.7cm}
 \begin{figure}[h]
\small
 \centering
 \begin{minipage}{0.4\linewidth}
\[
\begin{split}
r_1: S&\Longrightarrow \{A,S\}\\     
r_2:S&\Longrightarrow \{s_1\}\\
r_3:S&\Longrightarrow \{s_2\}\\
r_4: A&\Longrightarrow \{B,C\}\\
\end{split}
\]
\end{minipage}
 \begin{minipage}{0.4\linewidth}
\[
\begin{split}
r_5: B&\Longrightarrow \{b\}\\
r_6:C&\Longrightarrow \{B\}\\
r_7: C&\Longrightarrow \{c\}
\end{split}
\]
\end{minipage}
\end{figure}
\vspace{-0.4cm}

\eat{
\noindent
$r_1: S \Longrightarrow \{A,S\}$\\     
$r_2:S \Longrightarrow \{s_1\}$\\
$r_3:S \Longrightarrow \{s_2\}$\\
$r_4: A \Longrightarrow \{B,C\}$\\
$r_5: B \Longrightarrow \{b\}$\\
$r_6:C \Longrightarrow \{B\}$\\
$r_7: C \Longrightarrow \{c\}$
}

\eat{
 \begin{figure}[h]
 \centering
 \begin{minipage}{0.4\linewidth}
\[
\begin{split}
r_1: S&\Longrightarrow \{A,S\}\\     
r_2:S&\Longrightarrow \{s_1\}\\
r_3:S&\Longrightarrow \{s_2\}\\
r_4: A&\Longrightarrow \{B,C\}\\
r_5: B&\Longrightarrow \{b\}\\
r_6:C&\Longrightarrow \{B\}\\
r_7: C&\Longrightarrow \{c\}
\end{split}
\]
\end{minipage}
\begin{minipage}{0.4\linewidth}
\begin{tikzpicture}[node distance=15]
\node(r1) {$r_1$};
\node(center)[below of=r1] {};
\node(r4)[left of=center, node distance=25] {$r_4$};
\node(r5)[below of=r4, left of=r4] {$r_5$};
\node(r6)[below of=r4, right of=r4] {$r_6$};
\node(b)[below of=r5]{$b$};
\node(r51)[below of=r6]{$r_5$};
\node(b1)[below of=r51]{$b$};
\node(r2)[right of=center]{$r_2$};
\node(s1)[below of=r2]{$s_1$};
\draw[-](r1)--(r4);
\draw[-](r1)--(r2);
\draw[-](r4)--(r5);
\draw[-](r4)--(r6);
\draw[-](r5)--(b);
\draw[-](r6)--(r51);
\draw[-](r51)--(b1);
\draw[-](r2)--(s1);

\node(T) [below of=b1, node distance=20] {$T$};
\end{tikzpicture}
\end{minipage}
\end{figure}
}

The language of $G$ is $L(G)=\{(s_1+s_2)(bc+bb)^n|n\in \mathbb{N}_{\geq 0}\}$.
$G$ matches query $Q_1=\{s_1,b \}$ since $Q_1\subseteq \{s_1,b,b\} \in L(G)$. However, $G$ does not match $Q_2=\{s_1,s_2\}$.
\end{example}

\eat{In the remainder of this paper,  we will eliminate  the curly brackets of the right-hand side of productions.}

\label{sec:model:derivation}

Since the language of a grammar can be infinite, we will need to focus our attention on a small sample of its elements in which a match can be found.  For this, we will use the notion of parse trees and derivation sequences.

\begin{definition}
\textbf{(Parse Tree)} A parse tree $T$ associated with a grammar $G=(\Sigma,\Delta, S,R)$ is a finite unordered tree where each interior node
represents a {\em production} $r\in R$ and whose children represent $body(r)$, i.e. each child is either a terminal in $body(r)$  (in which case it is a leaf)
or is a production whose head is a variable in $body(r)$.  If $T$ consists of a single node, then it represents a terminal in $\Delta$.  We use {\em root(T)},
and {\em leaves(T)} to denote the root production and   leaves of $T$, respectively.
\end{definition}

For our purposes, a parse tree can be rooted at {\em any} terminal or production rather than just those whose head is $S$. 
Given a parse tree $T$ of a  grammar, we denote by $paths(T)$ the bag of all root-to-leaf paths in $T$, $productions(T)$ the bag of productions applied in $T$, and $symbols(T)$  the set of symbols that appear in $productions(T)$. 

We now define what portions of a keyword query a parse tree $T$ matches in terms of the query-relevant keywords generated by $T$, and adapt the notion of 
 \emph{derivation sequence} to our setting. 

\begin{definition} 
\textbf{(Generates)}
Given a grammar $G=(\Sigma,\Delta,S,R)$ and a query $Q\subseteq \Delta$, we say a parse tree $T$ \emph{generates} the set of matching keywords $leaves(T)\cap
Q$. A symbol $M\in\Sigma$ {\em generates} a set $X\subseteq Q$ iff one of its parse trees generates $X$.
\end{definition}

\eat{
We denote by $F(M)$ ($M\in\Sigma$) the set of sets of matching keywords that $M$ can generate, i.e. $F(M)=\{leaves(T)\cap Q|root(T)\in productions(M)\}$. It is easy to check that $G$ matches $Q$ iff $S$ can generate $Q$, i.e. $Q\in F(S)$.
}

\begin{definition}
\textbf{(Derivation Sequence)} Given a  grammar 
$G=(\Sigma, \Delta, S, R)$, we say a variable $A\in \Sigma / \Delta$ \emph{derives} a symbol $B\in \Sigma$ iff there exists a parse tree $T$ where $root(T)\in
productions(A)$, and $B\in symbols(T)$. Each path from $A$ to $B$ in  $T$ is a sequence of productions called a \emph{derivation sequence}. 
sequence $seq(A\mapsto B)$, we denote by $Set(seq(A\mapsto B))$ the bag of all productions applied in it.
\end{definition}

If a variable $A$ derives a symbol $B$, we say $A$ is an ancestor of $B$ and $B$ is a descendant of $A$. 

\begin{definition}
\textbf{(Simple Derivation Sequence)} A derivation sequence is \emph{simple} iff there is no production that appears in it more than once.  
\end{definition}

Intuitively, a simple derivation sequence is a derivation sequence where a recursion is fired at most once. It disregards unnecessary recursions and provides an upperbound for the complexity results in Section~\ref{subsec:data complexity}.

\begin{definition}\label{def:distance}
\textbf{(Distance)} Given a  grammar \\
$G=(\Sigma, \Delta, S,R)$,  the \emph{distance} from a variable $A\in \Sigma/\Delta$ to symbol $B\in \Sigma$ is the length of the longest simple derivation
sequence for $A$ deriving $B$, denoted  $d(A\mapsto B)$.The \emph{distance} of a grammar is defined as the longest distance between any two symbols, denoted
$d(G)$. It is easy to see $d(A \mapsto B) \leq d(G) \leq |R|$.
\end{definition}
\begin{example}
For the grammar in Example~\ref{ex:ex1} there is only one simple derivation sequence for $S$ to derive $C$, i.e. $r_1r_4$.  However, there are two simple derivation sequences to derive $B$, i.e. $r_1r_4$ and $r_1r_4r_6$. Hence $d(S\mapsto C)=2$ and $d(S\mapsto B)=3$. It is also easy to check that $d(G)=4$.
\end{example}

\eat{
\begin{definition}
\textbf{\em(Proper Grammar)} {\em A grammar $G=(\Sigma,\Delta,S,R)$ is said to be \emph{proper} if it has (1) no \emph{underivable} symbols: $\forall M \in \Sigma, \exists T$, $root(T)\in productions(S), M\in symbols(T)$; (2) no $unproductive$ variables: $\forall M \in \Sigma/\Delta, \exists T$, $root(T) \in productions(M)$.}
\end{definition}
}

We end this section by discussing the notion of a {\em repository} of  bag grammars.    A repository of grammars is essentially a set of grammars in which symbols (modules) may be {\em shared}, but must be done so consistently.  We assume that all grammars are {\em proper}, i.e. have no underivable symbols  or unproductive variables.

\begin{definition}
\textbf{(Bag Grammar Repository)} A \emph{bag grammar repository} is a set of bag grammars such that for any two grammars $G_1=(\Sigma_1,\Delta_1,S_1,R_1)$ and
$G_2=(\Sigma_2,\Delta_2,S_2,R_2)$, $\forall M\in (\Sigma_1\cap\Sigma_2)$, $\bigcup\limits_{\substack{r\in R_1\\head(r)=M}}r=\bigcup\limits_{\substack{r\in
R_2\\head(r)=M}}r$
\end{definition}

\section{Matching and Searching}
\label{sec:algo}
In this section we give an efficient algorithm that, given a grammar
$G$ and keyword query $Q$, determines whether $G$ matches $Q$.  We
start by presenting the basic algorithm, {\em Match}, which computes a
fixpoint of subsets of matching keywords using a bottom-up approach
over the hierarchy of nonterminals (Section~\ref{subsec:data
  complexity}).  We then give an optimized match algorithm
(Section~\ref{subsec:optimization}), and extend the results to
searching a {\em repository} of grammars.

Note that the matching problem is NP-complete with respect to combined
(data and query) complexity, as shown in
\cite{DBLP:journals/jcss/DeutchM12}.  However, since query size is
typically small (6 or fewer keywords), we focus on the {\em data
  complexity} 
and give a matching algorithm that is polynomial in the size of the
grammar.  As observed in the introduction, the algorithm
in~\cite{DBLP:journals/jcss/DeutchM12} could also be used here, since
it considers a (more general) graph pattern query.  However, our
algorithm is optimized for keyword queries and is therefore simpler
and considerably more efficient in our setting (see results in
Section~\ref{sec:experiment}).

\eat{
  efficient {\em Match} algorithm, which builds parse trees bottom up
  until either a match is found or a fixpoint of keyword sets is
  reached.  An optimization of match, {\em OptMatch}, is then developed
  which borrows ideas from semi-naive Datalog evaluation to avoid
  unproductive calculations, and trims the size of keyword sets
  maintained.  We then use the match algorithm to efficiently {\em
    search} a repository of bag grammars. }

\eat{we study the complexity of the decision problem of whether a search query $Q$ matches a grammar $G$. We show that the problem is $\mathcal{NP} \mbox{-} complete$ with respect to {\em combined} (query-and-data) complexity (Section~\ref{subsec:combinedComplexity}).

However, since the query size is typically small (fewer than 5 keywords) we focus on the  {\em data complexity} 
and show that the problem can be solved in  polynomial time in the size of the grammar
(Section~\ref{subsec:data complexity}). We then give an optimized match algorithm  (Section~\ref{subsec:optimization}), and extend the results to searching a {\em repository} of grammars.}

\eat{
\subsection{Combined complexity}
\label{subsec:combinedComplexity}
The decision problem is complicated by several factors:  the presence of multiple productions with the same head ({\em choice}), the fact that the same symbol may occur in the right-hand side of multiple productions or occur multiple times in the right-hand side of some production ({\em sharing}), and {\em recursion} in the grammar.  Note that in the absence of choice, checking whether or not a parse tree $T$ matches a query $Q$ can be done in polynomial time by testing whether or not the set of leaves of $T$ contains $Q$.  If the grammar contains sharing and recursion but no choice, the problem is still easy:  A parse tree generates all terminals in $\Delta$, but potentially multiple times.  However,  the number of times a terminal occurs is not relevant, so it  still suffices to check whether $\Delta$ contains $Q$.

However, when choice and sharing are added the problem becomes hard.
\begin{theorem}
Given a bag grammar $G$  (with choice, sharing, but no recursion) and  a search query $Q$, deciding whether $G$ matches $Q$ is $\mathcal{NP} \mbox{-} complete$.
\end{theorem} \label{T4}

The proof is by reduction from 3-SAT. 



}

\subsection{Keyword query match}\label{subsec:data complexity}
We now give an algorithm, {\em Match} (Algorithm~\ref{algo:AF}), which determines whether or not a grammar $G$ matches a keyword query $Q$. 
To do this,  {\em Match} builds a parse tree  bottom-up until some symbol generates $Q$, in which case $G$ matches $Q$, or until a fixpoint of query-relevant keywords is reached for each symbol, in which case $G$ does not match $Q$.   It can be shown that the fixpoint will be reached at or before height $O(d(G))$, and that therefore the algorithm is polynomial in the size of the grammar (although exponential in the query size due to the cost of generating sets of query-relevant keyword sets).

{\em Match} generates for each symbol $M \in \Sigma$ the set $F(M)$ of sets of query-relevant keywords for $M$.   
\eat{ $F(M)=\{leaves(T)\cap Q|root(T)\in productions(M)\}$.
It is easy to check that $G$ matches $Q$ iff $S$ can generate $Q$, i.e. $Q\in F(S)$.}
It does so by considering parse trees of increasing height $i$, and calculating for each $M \in \Sigma$ the set $F_i(M)$ of sets of query-relevant keywords for $M$  that are derived by 
parse trees of height $\leq i$. For terminals $M \in \Sigma$ (which are the leaves in a parse tree), $F_0(M)$ can be calculated directly (line 2, ignore for now line 3).  
For variables $M \in \Sigma/\Delta$, $F_0(M)=\emptyset$ (line 4). It then calculates $F_i(M)$ by initializing it to $F_{i-1}(M)$
(line 7) or $\emptyset$, 
then considering each production $r$  with $M$ as head, taking the  ``product-union" of $F_{i-1}(\alpha_j)$ for each $\alpha_j$ in $body(r)$, and adding the resulting set of query-relevant keywords to $F_{i}(M)$ (lines 9-11).  This continues until the query is matched by some $M$ (line 12), or until a fixpoint is reached (for each symbol $M$, $F_i(M)=F_{i-1}(M)$, line 13).

\eat{Given a grammar $G=(\Sigma,\Delta,S,R)$ and a query $Q$, the \emph{fixpoint} approach (Algorithm~\ref{algo:AF}) leads to build parse trees bottom-up till some symbol generates $Q$ (in which case $G$ matches $Q$) or a fixpoint is reached  (in which case $G$ does not match) i.e. when Algorithm~\ref{algo:AF} has enumerated all possible sets that each symbol can generate. More precisely, in the $i^{th}$ iteration (loop $L$), Algorithm~\ref{algo:AF}  calculates for each $M\in \Sigma$ $F_i(M)$, which represents the set of sets of query-relevant terminals that $M$ matches, and that are generated by parse trees of height $\leq i$ i.e. $F_i(M)=\{leaves(T)\cap Q|root(T)\in productions(M), height(T)\leq i\}$. It can be proved that the fixpoint is when $F_i(M)=\lim\limits_{j\to \infty}F_j(M)=F(M)$ for all $M\in\Sigma$, which guarantees the correctness of Algorithm~\ref{algo:AF}.}

\begin{algorithm}
\small
\DontPrintSemicolon
\KwIn{a grammar $G=(\Sigma,\Delta,S,R)$ and a query $Q\subseteq \Delta$\\}
\KwOut{boolean}
\tcc{\scriptsize$F(M):$  set of sets of query-relevant keywords for $M$}
\tcc{\scriptsize$F_i(M)=\{leaves(T)\cap Q|root(T)\in productions(M)$, $height(T)\leq i\}$} \Begin{        
\nl\ForEach{$M\in \Delta$}{
\nl\lIf {$F(M)=\emptyset$} {$F_0(M) \gets \{\{M\}\cap Q\}$}\\
\tcc{\scriptsize for OptMatch (Algorithm~\ref{algo:OAF})}
\nl \lElse {
    $F_0(M)\gets F(M)$}}
\nl\lForEach{$M\in \Sigma/\Delta$}{$F_0(M)\gets\emptyset$}\\
\nl$i \gets 0$ \tcp{\scriptsize a counter for iteration times}
\nl\textbf{L:}~$i\gets i+1$\\
\nl \lForEach{$M\in\Delta$}{$F_i(M)\gets F_{i-1}(M)$}\\
\nl\lForEach{$M\in\Sigma/\Delta$}{$F_i(M)\gets \emptyset$}\\
\nl\ForEach{production $r: M \rightarrow \alpha_1\alpha_2 \ldots \alpha_n \in R$} {
\tcc{\scriptsize Construct parse trees rooted at $r$} 
\nl\ForEach{$X_1\in F_{i-1}(\alpha_1)$, $\ldots X_n\in F_{i-1}(\alpha_n)$}{
\nl$F_i(M).addElement(\bigcup\limits_{j=1}^n{X_j})$
}
}
\nl\lIf{$\exists M\in\Sigma,Q \in F_i(M)$} {
	  \Return{true}
  }\\
\nl\lIf{$\exists M\in\Sigma/\Delta, ~F_{i}(M)\not=F_{i-1}(M)$} {goto $L$}\\
\nl\Return{false}
}
\caption{{\em Match}}
\label{algo:AF}
\end{algorithm}

\begin{example}
\label{ex:algo1}
Consider the grammar in Example~\ref{ex:ex1} and query $Q=\{b,c\}$. Initially, $F_0(s_1)=F_0(s_2)=\{\emptyset\}$, $F_0(b)=\{\{b\}\}$, $F_0(c)=\{\{c\}\}$, and $F_0(S)=F_0(A)=F_0(B)=F_0(C)=\emptyset$ (lines 1-4). 
In the first iteration of {\bf L} ($i=1$), we add to $F_1(S)$ the set $\emptyset$ (after processing rule $r_1,r_2,r_3$). 
After processing all other productions, we have $F_1(A)=\emptyset$, $F_1(B)=\{\{b\}\}$, and $F_1(C)=\{\{c\}\}$. We then proceed to the second iteration ($i=2$). During this iteration, when rule $r_4$ is processed, we add to $F_2(A)$ (which was initialized to $F_2(A)=\emptyset$) the product of $F_1(B)$ and $F_1(C)$, at each step taking the union of the two elements (e.g. $\{b\} \cup \{c\}= \; \{b, c\}$), resulting in $F_2(A)= \{\{b, c\}\}$.  Since $Q \in F_2(A)$, {\em Match} terminates.
If this early termination condition were omitted, the fixpoint would have been reached in the fifth iteration, since $F_5(S)=F_4(S)$,  $F_5(A)=F_4(A)$, $F_5(B)=F_4(B)$, $F_5(C)=F_4(C)$.


\eat{
calculate $F_1(s_1)=F_0(s_1)$ and the same for $s_2,b,c$. Next we process each production i.e. Line 9$\sim$11. For $r_1$, we get $F_1(S)=\emptyset$. For $r_2$, we get $F_1(S)=\{\{s_1\}\}$. Further after processing $r_3$, we have $F_1(S)=\{\{s_1\},\emptyset\}$. Similarly, after processing all productions, we have $F_1(S)=\{\{s_1\},\emptyset\}$, $F_1(A)=\emptyset$, $F_1(B)=\{\{b\}\}$, $F_1(C)=\{\{c\}\}$. We then proceed to the next iteration i.e. Line 13. It is easy to check that Algorithm~\ref{algo:AF} will terminate in the 3rd iteration since $F_3(S)=\{\{s_1,b,c\},\{bc\},\{s_1\},\emptyset\}\ni Q$ (Line 12). We can also check that the fixpoint will be reached in the $5^{th}$ iteration, since $F_5(S)=F_4(S)$,  $F_5(A)=F_4(A)$, $F_5(B)=F_4(B)$, $F_5(C)=F_4(C)$.
}

 \eat{
 \begin{table}[h]
\centering
\begin{tabular}{|c|c|c|c|c|}
\hline
&$s_1$&$s_2$&b&c\\
\hline
F&$\{\{s_1\}\}$&$\{\emptyset\}$&$\{\{b\}\}$&$\{\{c\}\}$\\
\hline
&S&A&B&C\\
\hline
$F_0$&$\emptyset$&$\emptyset$&$\emptyset$&$\emptyset$\\
\hline
$F_1$&$\{\{s_1\},\emptyset\}$&$\emptyset$&$\{\{b\}\}$&$\{\{c\}\}$\\
\hline
$F_2$&$\{\{s_1\},\emptyset\}$&$\{\{b,c\}\}$&$\{\{b\}\}$&$\{\{c\},\{b\}\}$\\
\hline
$F_3$&$\{s_1bc,bc,s_1,\emptyset\}$&$\{bc,b\}$&$\{b\}$&$\{c,b\}$\\
\hline
$F_4$&$\{s_1bc,bc,s_1,\emptyset,s_1b,b\}$&$\{bc,b\}$&$\{b\}$&\{c,b\}\\
\hline
$F_5$&$\{s_1bc,bc,s_1b,b,\emptyset,s_1\}$&$\{bc,b\}$&$\{b\}$&$\{c,b\}$\\
\hline
\end{tabular}
\caption{A running example of Algorithm~\ref{algo:AF}}
\label{tab:queries}
\end{table}}
\end{example}

We now show that the data complexity of {\em Match} 
is $O(|G|*d(G))$.   We start by 
showing that it will reach a fixpoint in $d(G)*(|Q|+1) + |Q|$ iterations by proving that if a symbol $M\in\Sigma$ generates a set $X\subseteq Q$, then there exists a parse tree $T$ rooted at $r\in productions(M)$ of height at most $ d(G)*(|Q|+1) + |Q|$ such that $leaves(T)\cap Q=X$.

\eat{ for each production $r\in R$, i.e. $O(|G|)$ in total for each iteration. Next, we will show that Algorithm~\ref{algo:AF} will reach the fixpoint in $ d(G)*(|Q|+1) + |Q|$ iterations by proving that if a symbol $M\in\Sigma$ generates a set $X\subseteq Q$, then there exists a parse tree $T$ rooted at $r\in productions(M)$ of height at most $ d(G)*(|Q|+1) + |Q|$ such that $leaves(T)\cap Q=X$.
}

\begin{lemma}\label{part2}
Given a grammar $G=(\Sigma, \Delta, S,R)$ and  a query $Q\subseteq \Delta$, $\forall M\in\Sigma,~\forall X\subseteq Q$, $height(M,X)\leq d(G)*(|Q|+1) + |Q|$.
\end{lemma}

\begin{proof}
Note that if  $M\in\Sigma$ generates $X\subseteq Q$, then the parse trees that generate $X$ fall in one of  two classes: (1) each child subtree generates a {\bf subset} of $X$, in which case we say the parse tree \emph{produces} the set $X$; (2) some child subtree generates $X$ {\bf by itself}, in which case we say the parse tree \emph{broadcasts} $X$.

We denote by $height(M,X)$ ($M\in \Sigma/\Delta, X\subseteq Q$) the minimum height of parse trees of $M$ that generate $X$; \\
$height(M,X)=-1$ if $M$ cannot generate $X$. Similarly, we denote by  $height_p(M,X)$ ($height_b(M,X)$) the minimum height of parse trees of $M$ that produce (broadcast) $X$.
For $height_p(M,X)$ we have the following:

\begin{equation}\label{eq:3.0}
\small
 height_p(M,X)=\begin{cases} 0 &\text{if } M\in \Delta, X=\{M\} \cap Q \\ -1 &\text{if } M\in \Delta, X\not=\{M\}\cap Q  \\
-1 & \text{if } M\in \Sigma/\Delta, |X|\leq 1\end{cases}
\end{equation}

\begin{equation}\label{eq:3.1}
\small
 height_p(M,X)\leq 1+\max\limits_{M'\in \Sigma,X'\subset X}height(M',X')
\end{equation}

\noindent
Turning to $height_b(M,X)$, we have:

\begin{equation}\label{eq:3.2}
\small
height_b(M,X) \leq  d(G) + \max\limits_{M^\prime\in \Sigma}{height_p(M^\prime ,X)}
\end{equation}

\eat{
\begin{equation}\label{eq:3.3}
\small
\begin{split}
&height(M,X)\\
 =&\begin{cases} \min(height_p(M,X), height_b(M,X)) & \mbox{if} \; \substack{height_p(M,X) \not= -1\\ height_b(M,X)\not= -1}\\
 \max(height_p(M,X), height_b(M,X))&\mbox{otherwise}
  \end{cases}\\
 \leq &d(G) + \max\limits_{M^\prime \in \Sigma}{ height_p(M^\prime ,X)}
  \end{split}
 \end{equation}}
 \begin{equation}\label{eq:3.3}
 \small
 height(M,X)\leq \max(height_p(M,X),height_p(M,X))
 \end{equation}
 
\noindent
So we have $height(M,X)\leq d(G)*(|Q|+1) + |Q|$.
\end{proof}

\eat{ 
\begin{lemma}\label{part1}
Given a grammar $G=(\Sigma, \Delta, S,R)$ and a query $Q\subseteq \Delta$, $\forall M\in\Sigma,~\forall X\subseteq Q$,
 $height_p(M,X)\leq |X|*(d(G)+1)$. 
\end{lemma}

\begin{proof} 
 Note that $\forall M\in \Sigma$, $height_p(M,X)= 0$ when $|X|=0$ or $|X|=1$. Moreover, $\forall M\in\Sigma, X\subseteq Q (|X|>1)$, $height_p(M,X)\leq 1 + \max\limits_{M'\in \Sigma,X'\subset X} height(M',X')$
\end{proof}
}

\eat{
\begin{example}
For example,  given the grammar in Example~\ref{ex:ex1}, and query $Q=\{s_1,b,c\}$, $T_1$ produces $\{b,c\}$ while $T_2$ broadcasts $\{b,c\}$. It is easy to check that $height_b(S,Q)=3$, $height_p(S,Q)=4$.
\begin{figure}[h]
\centering
\begin{tikzpicture}[node distance=15]
\node (r4) {$r_4$};
\node(center)[below of = r4] {};
\node(r5) [left of = center, node distance = 15] {$r_5$};
\node(r7) [right of = center,node distance=15] {$r_7$};
\node(b) [below of=r5, node distance = 15] {$b$};
\node(c) [below of=r7, node distance = 15] {$c$};
\draw[-] (r4) -- (r5);
\draw[-] (r4) -- (r7);
\draw[-] (r5) --(b);
\draw[-] (r7) --(c);
\node(T) [below of = r4, node distance = 60] {$T_1$};

\node(r1)[right of=r4, node distance=80]{$r_1$};
\node(0center)[below of=r1] {};
\node(r3)[right of=0center]{$r_3$};
\node(s2)[below of=r3]{$s_2$};
\node (1r4) [left of=0center, node distance=25]{$r_4$};
\node(1center)[below of = 1r4] {};
\node(1r5) [left of = 1center, node distance = 15] {$r_5$};
\node(1r7) [right of = 1center,node distance=15] {$r_7$};
\node(1b) [below of=1r5, node distance = 15] {$b$};
\node(1c) [below of=1r7, node distance = 15] {$c$};
\draw[-] (1r4) -- (1r5);
\draw[-] (1r4) -- (1r7);
\draw[-] (1r5) --(1b);
\draw[-] (1r7) --(1c);
\draw[-] (r1) --(r3);
\draw[-] (r1) --(1r4);
\draw[-] (r3) --(s2);
\node(1T) [below of = r1, node distance = 60] {$T_2$};

\end{tikzpicture}
\end{figure}

\end{example}
}

Using Lemma~\ref{part2} we can prove the following.
 
\begin{theorem}
The data complexity of {\em Match}
is \\$O(|G|*d(G))$.
\end{theorem}
 
\begin{proof}
The size of a grammar  $G=(\Sigma, \Delta, S,R)$
is defined as the sum of the sizes of its productions, \\
$|G|=\sum\limits_{r\in R}{(1+|body(r)|)}$. 
 By Lemma~\ref{part2} we know that the number of iterations of Algorithm~\ref{algo:AF} is bounded by 
$d(G)*(|Q|+1) + |Q|$.  
Each iteration (loop {\bf L}) of Algorithm~\ref{algo:AF} processes all productions, 
each of which takes $O(|body(r)|*2^{|Q|})$. 
Since the query size is considered as a constant, this yields a total time
of $O(|G|*d(G))$.
\end{proof}

\subsection{Optimized keyword query match}
\label{subsec:optimization}
We now introduce two improvements to {\em Match}, one of which
avoids unproductive calculations introduced by variables that are
fixed early, and the other of which reduces the size of $F_i(M)$.

\textbf{Optimization 1: Symbol Dependencies}
Borrowing ideas from semi-naive 
Datalog evaluation, we observe that a grammar yields symbol dependencies through the head and body
structure of rules, e.g.  that in 
{\em Match} the start variable $S$ will not be fixed until $A$ is fixed (a variable $M$ is \emph{fixed} in the $i^{th}$ iteration, iff $\forall j> i,F_j(M)=F_i(M)$, $\forall j< i,F_j(M)\subset F_i(M)$). 
The first optimization is to avoid unproductive calculations introduced by variables which are
fixed early.

Given a grammar $G=(\Sigma,\Delta,S,R)$, we create a \emph{precedence
  graph} $\cal{G=(V,E)}$ of symbols as follows: $\cal{V}$= $\Sigma$,
and a directed edge $(M, M^\prime)$ is added to $\cal{E}$ iff $\exists
r\in R, head(r)=M, body(r)\ni M'$. $\cal G$ has a directed cycle iff
$G$ is recursive. Two symbols are \emph{mutually recursive} iff they
participate in the same cycle of $\cal G$. Mutual recursion is an
equivalence relation on $\Sigma$, where each equivalence class
corresponds to a strongly connected component of $\cal G$. Denote by
$[M]$ the symbol equivalence class of $M$, and perform a topological
sort of $\cal G$ to construct a list $[M_1],\,[M_2]\ldots
[M_n]$. 
Clearly, if there is a path from $[M_j]$ to $[M_i]\,(i < j)$, then
symbols in $[M_i]$ are fixed earlier than symbols in $[M_j]$; symbols
in the same equivalence class are fixed in the same iteration.

\textbf{Optimization 2: Set Domination} We note that an element in $F_i(M)$ 
is \emph{useless} if it is a subset of (\emph{dominated by}) another element in $F_i(M)$. 
For example, $\{b\}$ is useless in $F_3(A)=\{\{b\},\{b,c\}\}$
 of Example~\ref{ex:algo1}.
Using set domination, $|F(M)|$ decreases from $2^{|Q|}$ to $\binom{|Q|}{\lfloor |Q|/2\rfloor}$.


\begin{algorithm}[h]
\small
\DontPrintSemicolon
\KwIn{a grammar $G=(\Sigma,\Delta,S,R)$ and a query $Q\subset \Delta$\\
}
\KwOut{boolean}
\textbf{Initialization}\\
Precompute the symbol equivalence classes $[M]$\\
Precompute the topological order of the symbol equivalence classes, $T\_order$\\

\Begin{
\nl\lForEach{$M \in \Sigma$}{$F(M)= \emptyset$}  \\
\nl\ForEach{$[M]\in T\_order$} {
	\tcp{\scriptsize Find the relevant productions of $[M]$}
\nl	$R^\prime\gets \{r|r\in R, head(r)\in [M]\}$\\
\nl	$\Sigma'\gets \bigcup\limits_{r\in R'}(head(r)\cup body(r))$\\
\tcc{\scriptsize all descendants  have been fixed and hence are treated as terminals}
\nl	$\Delta'\gets \Sigma'/\ ([M]/\Delta)$ \\
	\tcc{\scriptsize $F(\cdot)$ is a global variable visible for both $OptMatch()$ and $Match()$}
\nl	\lIf{$Match(Q,\Sigma',\Delta', M,R^\prime)$}{\Return{true}}\\
}
\nl\Return {false}
}

\caption{{\em OptMatch}}
\label{algo:OAF}
\end{algorithm}

Algorithm~\ref{algo:OAF} gives the optimized algorithm,
{\em OptMatch}. Note that the topological order of symbol equivalence
classes is query-independent and can be precomputed.  Line 6 of
{\em OptMatch} calls {\em Match} for the current equivalence class.  $F(M)$ is
global, and is used in lines 2-3 of {\em Match}.  Set domination is
captured in the \emph{addElement} method in {\em Match}.

\begin{example}\label{ex:algo2}
Consider the grammar in Example~\ref{ex:ex1} and query $Q=\{b,c\}$. One topological order of the symbol equivalence classes is $[b],[c],[B],[C],[s_1],[s_2],[A],[S]$ where each of the classes consists of the symbol itself. For $[b]$, $R'=\emptyset,~\Sigma'=\{b\},~\Delta'=\{b\}$ (lines 3-5); after calling {\em Match} (line 6), $F(b)=\{\{b\}\}$. Similarly, $F(c)=\{\{c\}\}$, $F(B)=\{\{b\}\}$. Now we process $[C]$ where  $R'=\{r_6,r_7\},~\Sigma'=\{C,B,c\},~\Delta'=\{B,c\}$. When calling {\em Match} for $[C]$, the initialization results in $F_0(B)=F(B)=\{\{b\}\},~F_0(c)=F(c)=\{\{c\}\}$ (line 3 of Algorithm~\ref{algo:AF}). At the end of {\em Match}, we get $F(C)=\{\{b\},\{c\}\}$. We can check that  {\em OptMatch} terminates after processing $[A]$ since $Q \in F(A)=\{\{b,c\}\}$.
\end{example}

\newpage
 \subsection{Searching a bag grammar repository}
\label{subsec:searchrepo}

Given a bag grammar repository and query $Q$, we must retrieve all
grammars that match $Q$. A straightforward way to do this is to run
$OptMatch$
over each grammar.  This way a grammar that is reused by other
grammars will be processed multiple times. One solution is to union
productions of all grammars to form a universal grammar.  However,
this grammar would be too large to fit in memory.

We therefore process grammars one by one while recognizing grammar
reuse; an individual grammar can be assumed to fit in main memory.  We
first build inverted indexes to help identify {\em candidate
grammars}, i.e. grammars in which every keyword of the query appears
(although they may not simultaneously occur in some bag in the
language of the grammar).  Each index maps a keyword to a list of
grammars in which the keyword appears.  Given a query, we find
candidate grammars by intersecting the corresponding lists.

We then process candidate grammars (using $OptMatch$) so that a
grammar is always processed earlier than the grammars that reuse it.
Specifically, we create a precedence graph where nodes of the graph
are grammars, and there is an edge from $G_i$ to $G_j$ iff $G_j$
reuses $G_i$ (i.e. the start module of $G_i$ appears in $G_j$ as a
variable). The graph is a $DAG$, since there is a temporal dimension
to reuse. A topological order of the $DAG$ is the order in which
grammars are processed.

We cache intermediate results ($F(S)$) for grammars that are reused
and clean them from memory when there are no unprocessed grammars that
reuse them. In this way, we balance memory size and overhead of
redundant computations.

   
%
\eat{
\begin{algorithm}[h]
\DontPrintSemicolon
\KwIn{A keyword query $Q$ and a repository $R$}
\KwOut{A set of grammars}
\textbf{Initialization}\\
Precompute the topological order of all grammars, denoted as $L_G=G_1$, $G_2\ldots G_n$\\
Denote by $order(G)$ the order of $G$ i.e. its index in $L_G$\\
Denote by $reuse(G)$ the set of grammars that reuse $G$\\
Denote by $L(q)$ the grammar list  indexed by keyword $q$\\
\emph{$search(Q,R)$}\\
\Begin{
  $result \gets \emptyset $  \\
\tcc{\scriptsize Find the candidate grammars $L_c$, grammars of which are in the topological order}
  $L_c = \bigcap\limits_{q\in Q}{L(q)}$\\
 \ForEach{grammar $G=(\Sigma, \Delta, S, P) \in L_c$}{
      Load $G$ into the memory\\
      \If{$optMatch(Q,\Sigma, \Delta, S, P)$}{
	$result \gets result \cup reuse(G)$\\
	$L_c.remove(reuse(G))$
      }\lElse {
	Store the pair $(G, \,F(S))$ in memory
      }
      \ForEach{$(G^\prime, F(S^\prime))$ in the memory} {
	\If{ $\max\limits_{G''\in reuse(G')}(G'') \leq order(G)$} {
	    Clear $(G^\prime, F(S^\prime))$  from the memory
      }}}
      \Return{$result$}
}
\caption{Pseudo Code \textbf{Search a repository}}
\label{algo:AR}
\end{algorithm}}


\section{Ranking}
\label{sec:ranking}
\eat{In the previous sections we described our workflow model that is
  based on context-free graph grammars.  We also studied the
  complexity of the following decision problem: given a keyword query,
  and a keyword-labeled workflow specification, does there exist a
  workflow run in which all keywords are assigned?  We showed that the
  decision problem can be answered in time polynomial in the size of
  the grammar and exponential in the size of the query, and presented
  several efficient algorithms.}

For large repositories of grammars, more grammars may match a query
than can be shown to a user, motivating ranking.  In this section we
describe a relevance measure based on probabilistic grammars
(Section~\ref{sec:ranking:semantics}) and develop an algorithm for
computing the relevance of a grammar to a query
(Section~\ref{sec:ranking:algo}).  We also describe an efficient
top-$k$ algorithm (Section~\ref{sec:ranking:topK}).

\subsection{Semantics of relevance} 
\label{sec:ranking:semantics}

\eat{
Recall that workflows are represented by context-free bag grammars in our formalism.  
It is therefore natural to turn to probabilistic (or
stochastic) context-free grammars (PCFG) for ranking. PCFGs have been
used to analyze the probability that a string is generated by a
particular grammar, with applications in Natural Language
Processing~\cite{DBLP:books/daglib/0001548} and beyond.  Although our
grammars generate bags rather than strings, the formalism applies
literally in our setting.}

Probabilistic \eat{(or stochastic) }context-free grammars (PCFG) have
been used in applications such as natural language processing (NLP) to
analyze the probability that a string is generated by a particular
grammar.  It is therefore natural to use them for ranking.
Although our grammars generate bags rather than strings, the formalism applies literally in our setting.

\begin{definition}
  \emph{ \textbf{(Probabilistic Context-Free Bag\\Grammar)} A
    \emph{probabilistic context-free bag grammar} $G=(\Sigma,\Delta,
    S,R,\rho)$ is a context-free bag grammar in which each production
    is augmented with a probability $\rho: R\rightarrow (0,1]$ such
    that}
$\forall M \in \Sigma/ \Delta, ~\sum\limits_{r \in productions(M)}\rho(r)=1$.
\end{definition}

\eat{ The PCFG formalism will allow us to estimate the probability
  that a particular grammar will generate an execution that matches
  the keyword query.  Then, for a given keyword query, and for a
  workflow repository, our goal will be to rank workflow
  specifications on the probability that they will generate an
  execution matching the query.  }

In a PCFG, which production $r\in productions(M)$ is chosen at a given
composite module $M$ is independent of the choices that lead to $M$.
Thus, the probability of a parse tree $T$, which we will denote
$\rho(T)$, is the product of the probabilities of productions used in
the derivation, i.e.,
  $\rho(T)=\prod\limits_{r \in productions(T)}{\rho(r)}$.

Probabilities of productions in $G$ may be given by an expert or mined
from the corpus.  In this paper, we do not consider how these
probabilities are derived, but note that much work on the topic has
been done in the NLP community~\cite{DBLP:books/daglib/0001548}, and
the techniques are likely applicable here.\eat{ Without loss of
  generality, in the remainder of this paper we assume that all
  productions of a given module $M$ have equal probabilities.}

Given a grammar $G$ and a query $Q$, our goal is to
compute a relevance score, denoted $score(G,Q) \in [0,1]$,
representing the likelihood of $G$ to generate a parse tree that
matches $Q$.  We want these scores to be comparable across grammars,
which would enable us to say that, if $score(G,Q) > score(G',Q)$, then
$G$ is more query-relevant than $G'$.  Generating scores that are
comparable across grammars turns out to be tricky, because, as we show
next, parse trees of probabilistic context-free bag grammars may not
form a valid probabilistic space.

\eat{
Suppose that query $Q=\{ a \}$ is given, and consider grammar $G$
 in Figure~\ref{fig:nonrec}.
This is a non-recursive grammar that generates exactly three parse
trees with probabilities: $\rho(T_1)= \rho(T_2) =
\rho(T_3) = 0.25$ (Recall that, because we are working with bag
grammars, parse trees are unordered.)

\begin{figure}[h]
\begin{minipage}{.4\linewidth}
 \begin{eqnarray*}
  r_1:S\rightarrow & \{A, A\} & (1)\\
  r_2:A\rightarrow & \{a\}& (0.5)\\
  r_3:A\rightarrow & \{b\} & (0.5)\\
\end{eqnarray*}
\end{minipage}
\begin{minipage}{0.6\linewidth}
\centering
\begin{tikzpicture}[node distance=15]
  \node(r1) {$r_1$};
\node(center)[below of=r1] {};
\node(r21) [left of = center,node distance=10]{$r_2$};
\node(r22)[right of = center,node distance=10]{$r_2$};
\node(a1)[below of=r21]{$a$};
\node(a2)[below of=r22]{$a$};
\node(T1)[below of = center, node distance=30]{$T_1$};
\draw[-] (r1)--(r21);
\draw[-] (r1)--(r22);
\draw[-] (r21)--(a1);
\draw[-] (r22)--(a2);

  \node(1r1) [right of = r1, node distance=40]{$r_1$};
\node(1center)[below of=1r1] {};
\node(1r21) [left of = 1center,node distance=10]{$r_2$};
\node(1r22)[right of = 1center,node distance=10]{$r_3$};
\node(1a1)[below of=1r21]{$a$};
\node(1a2)[below of=1r22]{$b$};
\node(T2)[below of = 1center, node distance=30]{$T_2$};
\draw[-] (1r1)--(1r21);
\draw[-] (1r1)--(1r22);
\draw[-] (1r21)--(1a1);
\draw[-] (1r22)--(1a2);

  \node(2r1) [right of = r1, node distance=80]{$r_1$};
\node(2center)[below of=2r1] {};
\node(2r21) [left of = 2center,node distance=10]{$r_3$};
\node(2r22)[right of = 2center,node distance=10]{$r_3$};
\node(2a1)[below of=2r21]{$b$};
\node(2a2)[below of=2r22]{$b$};
\node(T3)[below of = 2center, node distance=30]{$T_3$};
\draw[-] (2r1)--(2r21);
\draw[-] (2r1)--(2r22);
\draw[-] (2r21)--(2a1);
\draw[-] (2r22)--(2a2);

\end{tikzpicture}
\end{minipage}
\caption{Non-recursive bag grammar}
\label{fig:nonrec}
\end{figure}

\eat{
\begin{verbatim}
 T1:   r1             T2:  r1       T3:   r1
      / \                 / \            /  \
     r2  r2              r2  r3         r3   r3
     |   |               |   |          |    |
     a   a               a   b          b    b
\end{verbatim}}

The parse trees do not form a valid probabilistic space, because their
probabilities do not sum to $0.75 < 1$.  Such PCFGs are referred to
as {\em improper}.  Note that, while non-recursive string PCFGs are
always proper, non-recursive bag PCFGs may be improper.  An easy fix,
which works well for grammar $G$, is to normalize the probabilities of
parse trees by their sum.  However, as we will see next, this does not
work for recursive PCFGs.
}

Consider the grammar $G'$ in Figure~\ref{fig:rec} that consists of two
productions, chosen with equal probability (probabilities are
indicated in parentheses).  Since $G'$ is recursive, it generates an
infinite number of parse trees that match query $Q=\{a\}$.  Two of
these are shown in Figure~\ref{fig:rec}.

\begin{figure}[h]
\small
 \centering
\begin{minipage}{0.55\linewidth}
\begin{eqnarray*}
r_1:S\Longrightarrow &\{S,S,S\}&(0.5)\\
r_2:S\Longrightarrow&\{a\}&(0.5)
\end{eqnarray*}
\end{minipage}
\begin{minipage}{0.4\linewidth}
\begin{tikzpicture}[node distance=15]
 \node(r2) {$r_2$};
\node(a)[below of = r2] {$a$};
\node(T1)[below of = r2, node distance=30]{$T_1$};
\draw[-] (r2)--(a);

\node(1r1)[right of =r2, node distance=40] {$r_1$};

\node(1r21)[below of =1r1]{$r_2$};
\node(1r20)[left of=1r21] {$r_2$};
\node(1r22)[right of =1r21]{$r_2$};

\node(1a0)[below of = 1r20]{$a$};
\node(1a1)[below of = 1r21]{$a$};
\node(1a2)[below of = 1r22]{$a$};
\node(T2)[below of=1r1, node distance = 40]{$T_2$};

\draw[-](1r1)--(1r20);
\draw[-](1r1)--(1r21);
\draw[-](1r1)--(1r22);

\draw[-](1r20)--(1a0);
\draw[-](1r21)--(1a1);
\draw[-](1r22)--(1a2);
\end{tikzpicture}
\end{minipage}
\caption{Recursive bag grammar}
\label{fig:rec}
\end{figure}
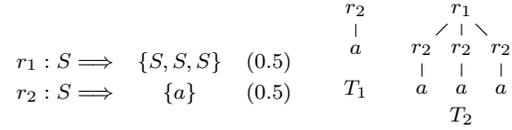

\eat{
\begin{verbatim}
 T1:  r2          T2:     r1
      |                 / | \
      a                r2 r2 r2
                       |  |   |
                       a  a   a
\end{verbatim}
}

These trees have probabilities: $\rho(T_1)=0.5$ and $\rho(T_2) = 0.5^4
= 0.0625$.  Unfortunately, it is not possible to compute a
normalization factor by summing the probabilities of the infinitely
many parse trees, because this sum is
irrational~\cite{Chi99statisticalproperties}. Generally, given a PCG
$G=(\Sigma,\Delta, S,R,\rho)$, and a query $Q\subseteq \Delta$, it is
customary to define a relevance score $score(G,Q)$ using {\em max} or
{\em sum} semantics.

\eat{{,Etessami:2009:RMC:1462153.1462154}. For $G'$, the sum is
  $score_2( G,Q)=(\sqrt{5}-1)/2$. This observation holds true for
  string and bag PCFGs alike. \julia{Need to compute the probability
    for bags.}}

\eat{
Given a probabilistic context-free grammar (PCG)
$G=(\Sigma,\Delta,R,S,\rho)$, a query $Q\subseteq \Delta$, a score
function is a function $score:G\times Q \rightarrow \mathbb{R_{\ge
    \mbox{0}}}$, which measures how much the grammar is related to the
query. \emph{max} and \emph{sum} semantics are common in ranking of
relational data, XML, business processes and etc.. We will next
discuss these two semantics.
}

\eat{
We note that the probability distribution of two grammars are
independent, which makes two grammars not comparable no matter we use
max or sum semantics. More clearly, given one PCG, we are able to tell
which one of two parse trees is more likely by comparing their
probabilities. Yet given two parse trees each from a different PCG, it
is not reasonable to say one parse tree is more likely than the other,
due to the fact that the probabilities of these two parse trees are
from two probability distributions. To make the probabilities from two
grammars comparable, we do normalization.
}

\eat{
\begin{enumerate}
\item $score_{sum}(G,Q)=\sum\limits_{\substack{r(T)\in
      R(S)\\Q \subseteq leaves(T)}}{\rho(T)}/ \sum\limits_{r(T)\in
    R(S)}{\rho(T)}$
\item $score_{max}(G,Q)=\max\limits_{\substack{r(T)\in
      R(S)\\Q \subseteq leaves(T)}}{\rho(T)}/ \max\limits_{r(T)\in
    R(S)}{\rho(T)}$
\end{enumerate}
}

\begin{equation}
\small
   score_{sum}(G,Q)=\sum\limits_{\substack{r(T)\in
      R(S)\\Q \subseteq leaves(T)}}{\rho(T)}/ \sum\limits_{r(T)\in
    R(S)}{\rho(T)}
\end{equation}
\begin{equation}
\small
    score_{max}(G,Q)=\max\limits_{\substack{r(T)\in
      R(S)\\Q \subseteq leaves(T)}}{\rho(T)}/ \max\limits_{r(T)\in
    R(S)}{\rho(T)}
\label{eq:score}
\end{equation}

Sum semantics is intuitive: we normalize the total probability of all
parse trees matching the query by the total probability of all parse
trees.  Unfortunately, as we argued above, this value cannot be
computed because probabilities may be irrational. It was shown
in~\cite{Etessami:2009:RMC:1462153.1462154} that $score_{sum}$ may be
approximated by solving a monotone system of polynomial equations.
However, this approach is in PSPACE, and is very expensive in
practice. 

Motivated by these considerations, we opt for max scoring semantics,
where the score of a grammar for a query is computed by dividing the
probability of the most likely parse tree matching the query by the
probability of the over-all most likely parse tree.  This semantics is
reasonable, and, as we will show next, $score_{max}(G,Q)$ can be
computed efficiently.  We refer to $score_{max}(G,Q)$ simply as
$score(G,Q)$ in the remainder of the paper.

\eat{
According to \cite{Etessami:2009:RMC:1462153.1462154}, sum semantics
may result in irrational probabilities (see example.~\ref{irrational
  probability}). Similar to\cite{Etessami:2009:RMC:1462153.1462154},
$score_2(G,Q$ could be obtained by solving a monotone system of
polynomial equations (a.k.a. MSPE), which has been proved to be in
$PSPACE$. And hence, we use max semantics as our ranking semantics. We
will give an efficient algorithm to compute the score. (see
Algorithm.\ref{algo:AF_Ranking}). Note that $ \max\limits_{r(T)\in
  R(S)}{\rho(T)}$ can be precomputed using the same algorithm. We will
refer to $ \max\limits_{r(T)\in R(S)}{\rho(T)}$ simply as
$\rho_{max}(S)$, $score_1(G,Q)$ as $score(G,Q)$.

\begin{example}\label{irrational probability}
Given a grammar 
$$G:\; S\rightarrow SSS\;(0.5) \;|\; a\;(0.5)\;$$
and a search query $Q=\{a\}$, we know that $Q$ matches all parse trees
of $G$, i.e. $score_2(
G,Q)=(\sqrt{5}-1)/2$. \cite{Etessami:2009:RMC:1462153.1462154}
\end{example}
}

\subsection{Computing the score of a workflow}
\label{sec:ranking:algo}

We now present Algorithm~\ref{algo:AF_score}, {\em Score}, which
computes the relevance score of grammar $G$ for query $Q$ per
Equation~\ref{eq:score}.  This is a fixpoint algorithm that is similar
in spirit to {\em Match}.  

Note that the algorithm of~\cite{DBLP:journals/pvldb/DeutchMPY10} can
also be used to compute the relevance score of grammar $G$ for query
$Q$.  This algorithm uses a similar framework
as~\cite{DBLP:journals/jcss/DeutchM12}, and so is very general, but is
less efficient in our particular scenario.  We will give experimental
support to this claim in Section~\ref{sec:exp:rank}.

Recall from {\em Match} that $F_i(M)$ represents the set of sets of
query-relevant keywords that module $M$ matches, and that are derived
by parse trees of height $\leq i$.  {\em Score} uses $F_i(M)$ and a
data structure $\rho_i(M,X)$, in which it stores, for each $X \in
F_i(M)$, the score of the corresponding parse tree.  Like {\em Match},
{\em Score} manipulates a global data structure $F(M)$; {\em Score}
also maintains the corresponding global $\rho(M,X)$.  These data
structures will become important when we consider an optimization,
called {\em OptScore}.

\eat{
\setlength{\algomargin}{0em}
\begin{algorithm}
\DontPrintSemicolon
\KwIn{grammar $G=(\Sigma,\Delta,S,R, \rho)$, query $Q$\\
}
\KwOut{$score(Q,\Sigma, \Delta, S,R, \rho)$\\ 
}
\Begin{
\ForEach{$M \in \Delta$} {
\If{$F(M)=\emptyset$} {$F_0(M)\gets\{\{M\}\cap Q\};~\rho_0(M,\{M\}\cap Q) \gets  1.0$ }
      \lElse{ $F_0(M)=F(M);~\rho_0(M,X)=\rho(M,X)$}}
\lForEach{$M\in \Sigma/\Delta$}{
    $F_0(M)\gets\emptyset$} 
$i\gets 0$\\
\nlset{ L:} 
$i\gets i + 1$\\
\ForEach{production $r: M \rightarrow \alpha_1\alpha_2 \ldots \alpha_n \in R$} {
	\ForEach{$X_1\in F_{i-1}(\alpha_1)\ldots X_n\in F_{i-1}(\alpha_n)$}{
		$X \gets \bigcup_{j=1}^{n}X_j$\\
		\If{$X\not\in F_i(M)$}{
			$F_i(M).addElement(X)$\\
			$\rho_i(M, X)\gets \rho(r)\times \prod_{j=1}^n{\rho_{i-1}(\alpha_j, X_j)}$\\
		}\Else{
		$\rho_i(M, X)\gets \max(\rho_{i}(M,X), \rho(r)\times	\prod_{j=1}^n{\rho_{i-1}(\alpha_j, X_j))}$\\
			}
		}
}
\If{$\forall M\in \Sigma/\Delta,~F_{i}(M)=F_{i-1}(M)$} {
	\lIf{$Q\not\in F_i(S)$}{\Return{$0.0$}}\\
	\If{$\forall M\in \Sigma,\forall X\subseteq Q,~\rho_i(M,X)=\rho_{i-1}(M,X)$} {
	  $F(M)=F_i(M);~\rho(M,X)=\rho_i(M,X)$\\
	  \Return{$\rho(S,Q)/\rho_{max}(S)$}}
}
goto $L$
}
\caption{Score}
\label{algo:AF_score}
\end{algorithm}
}

\begin{algorithm}[t]
\small
\DontPrintSemicolon

\KwIn{grammar $G=(\Sigma,\Delta,S,R)$ and query $Q\subseteq \Delta$\\}
\KwOut{double}
\Begin{          
\nl\ForEach{$M\in \Delta$}{
\nl\If {$F(M)=\emptyset$}{$F_0(M) \gets \{\{M\}\cap Q\}$;~$ \rho_0(M,\{M\}\cap Q)\gets 1.0$}
\nl \Else {
    $F_0(M)\gets F(M)$ $\forall X\in F_0(M), \rho_0(M,X)\gets \rho(M,X)$}}
   
\nl \lForEach{$M\in \Sigma/\Delta$}{$F_0(M)\gets\emptyset$}\\

\nl$i \gets 0$ \tcp{\scriptsize a counter for iteration times}
\nl\textbf{L:}~$i\gets i+1$\\
\nl \ForEach{$M\in\Delta$}{$F_i(M)\gets F_{i-1}(M)$\\
	$\forall X\in F_i(M),\rho_i(M,X)\gets \rho_{i-1}(M,X)$}
\nl\lForEach{$M\in\Sigma/\Delta$}{$F_i(M)\gets \emptyset$}\\
\nl\ForEach{production $r: M \rightarrow \alpha_1\alpha_2 \ldots \alpha_n \in R$} {
\nl\ForEach{$X_1\in F_{i-1}(\alpha_1)$, $\ldots X_n\in F_{i-1}(\alpha_n)$}{
$X\gets \bigcup\limits_{j=1}^n{X_j}$
 $prob\gets \rho(r)\times \prod_{j=1}^n{\rho_{i-1}(\alpha_j, X_j)}$\\
\lIf{$X\not\in F_{i}(M)$}{
$\rho_{i}(M,X)\gets prob$}\\
\lElse{$\rho_i(M, X)\gets \max(\rho_{i}(M,X), prob)$}
\nl$F_i(M).addElement(X)$}
}
\LinesNotNumbered
\nlset{13} \lIf{$\exists M\in\Sigma/\Delta, ~F_{i}(M)\not=F_{i-1}(M)$} {goto $L$}\\
 \lIf{$Q\not\in F(S)$}{\Return{$0.0$}}~\tcp{not matching}
\lIf{$\exists M\in \Sigma,\exists X\in F_i(M),~\rho_i(M,X)>\rho_{i-1}(M,X)$}{goto $L$}\\
$F(M)=F_i(M);~\rho(M,X)=\rho_i(M,X)$\\
 \Return{$\rho(S,Q)/\rho_{max}(S)$}
}

\caption{$Score$}
\label{algo:AF_score}
\end{algorithm}

Algorithm {\em Score} starts by storing the set of query-relevant
keywords annotating each terminal module $M$ in $F_0(M)$, and by
recording the probability score of 1.0 in $\rho_0$.  Next, for
non-terminal modules, $F_0(M)$ is initialized to an empty set.  The
bulk of the processing happens next, where, at each iteration $i$, we
consider each production $r$ rooted at $M$, and generate all sets of
query terminals resulting from parse trees of height at most $i$
rooted at $M$.  We record the resulting subset of query keywords $X$
if this subset has not been seen before, or it if is the
highest-scoring parse tree for this subset at the current round.  We
compute the probability score of a parse tree resulting from
production $r$ as the product of the probabilities of its subtrees and
the probability of $r$.

{\em Score} terminates when either no new subsets of the query are
generated, or no better (more probable) derivations of existing
subsets are found.  {\em Score} returns the normalized probability of
the most probable derivation tree matching $Q$. Note that
normalization factor $\rho_{max}(S)$ is query-independent and can be
computed by invoking $Score(G,\emptyset)$.  We compute $\rho_{max}(S)$
offline and store it for future use.

\begin{example}\label{ex:algo3}
  Consider the grammar in Example~\ref{ex:ex1} and assume that
  productions with the same head are equally likely. Given the query
  $Q=\{b,c\}$, {\em Score} calculates $\rho_{max}(S)=1/3$,
  $\rho(S,Q)=\frac{1}{3}*\frac{1}{3}*\frac{1}{2}$ and
  returns 
  $\frac{1}{3}*\frac{1}{3}*\frac{1}{2}/\rho_{max}(S)=\frac{1}{6}$ when
  it terminates at the end of the $5^{th}$ iteration.
\end{example}

The worst-case running time of {\em Score} is polynomial in the size
of the grammar.  This can be shown by a similar argument as for {\em
  Match} (Section~\ref{subsec:data complexity}), and is based on the
observation that, for any symbol $M$, the height of the most probable
parse tree generating any subset of $Q$ is bounded by $\leq
d(G)*(|Q|+1) + |Q|$.  

\eat{
Similar to the decision problem, we can prove that, given a grammar
$G=(\Sigma,\Delta,S,R)$ and a query $Q\subseteq \Delta$, the most
probable parse tree for $M\in \Sigma$ to generate any $X\in F(M)$ has
the height $\leq d(G)*(|Q|+1) + |Q|$, which implies a polynomial data
complexity.}

We also developed an optimized version of {\em Score}, which we call
{\em OptScore}.  We do not detail the {\em OptScore} algorithm here,
but note that it is based on the two optimizations performed in {\em
  OptMatch}.  The first optimization, {\bf symbol dependencies},
identifies equivalence classes of modules of $G$ based on their reuse
of variables on the right-hand-side of productions.  {\em OptScore}
computes a topological ordering among equivalence classes, and runs
algorithm {\em Score} for each class in reverse topological order,
saving intermediate results in global data structures $F(M)$ and
$\rho(M,X)$.  The second optimization, {\bf set domination}, includes
probabilities in the notion of domination: a set $X_1\in F_i(M)$
dominates $X_2\in F_i(M)$ iff $X_1\supset X_2$ and $\rho_i(M,X_1)\geq
\rho_i(M,X_2)$.


\eat{
\begin{algorithm}
\DontPrintSemicolon
\KwIn{\\A probabilistic grammar $G=(\Sigma,\Delta,S,R, \rho)$\\
A keyword query $Q\subseteq \Delta$\\}
\KwOut{$double$}
\textbf{Initialization}\\
Compute the variable equivalence classes $[M]$\\
Compute the topological order of the variable equivalence classes, $T\_order$\\
Mark all variables as unfixed\\
\emph{$optScore(Q, \Sigma, \Delta, S, R,\rho)$}\\
\Begin{
\ForEach{$[M]\in T\_order$} {
	\tcp{\scriptsize Find the relevant productions of $[M]$}
	$R^\prime\gets \{r|r\in R, head(r)\in [M]\}$\\
	$\Sigma'\gets \bigcup\limits_{r\in R'}(head(r)\cup body(r))$\\
	$\Delta'\gets \Sigma'/\ [M]$\\
	\tcc{\scriptsize $F(\cdot)$ and $\rho(\cdot)$ are global variables visible for both $optScore()$ and $Score()$}
	$score(Q,\Sigma',\Delta', M,R^\prime,\rho)$\\
	
}
\lIf{$Q\in F(S)$}{\Return {$\rho(S,Q)/\rho_{max}(S)$}}\\
\Return {$0.0$}
}
\caption{Optimized algorithm of Algorithm.\ref{algo:AF_ranking}(almost the same as Algorithm.\ref{algo:OAF})}
\label{algo:OAF_ranking}
\end{algorithm}
}

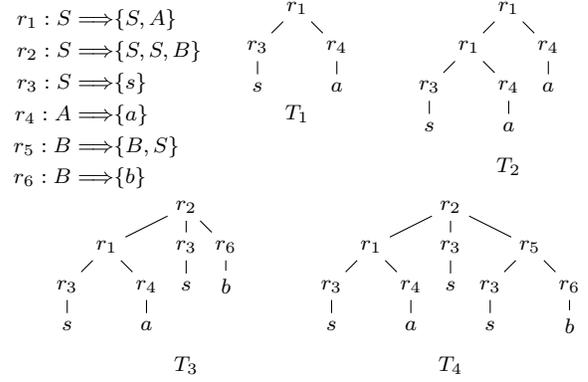
\begin{figure}[t!]
\small
 \begin{minipage}[!t]{0.4\linewidth}
\[
\begin{split}
r_1: S\Longrightarrow & \{S,A\}  \\
r_2: S\Longrightarrow & \{S,S,B\}\\
r_3: S\Longrightarrow & \{s\}    \\
r_4: A\Longrightarrow & \{a\}    \\
r_5: B\Longrightarrow & \{B,S\}  \\
r_6: B\Longrightarrow & \{b\}    \\
\end{split}
\]
\end{minipage}
\begin{minipage}{0.6\linewidth} 
\begin{tikzpicture}[node distance=15]
\node(r1) {$r_1$};
\node(center)[below of=r1]{};
\node(r3) [left of=center] {$r_3$};
\node(r4) [right of=center] {$r_4$};
\node(s) [below of=r3] {$s$};
\node(a)[below of=r4]{$a$};
\node(T1)[below of=r1, node distance=40] {$T_1$};
\draw[-] (r1)--(r3);
\draw[-] (r1)--(r4);
\draw[-] (r3)--(s);
\draw[-] (r4)--(a);

\node(1r1) [right of=r1, node distance=80] {$r_1$};
\node(1center)[below of=1r1]{};
\node(1r11)[left of=1center] {$r_1$};
\node(1r4)[right of=1center]{$r_4$};
\node(1center1)[below of=1r11] {};
\node(1r3) [left of=1center1] {$r_3$};
\node(1r41) [right of=1center1] {$r_4$};
\node(1s) [below of=1r3] {$s$};
\node(1a1)[below of=1r41]{$a$};
\node(1a)[below of=1r4]{$a$};
\node(T2)[below of=1r1, node distance=60] {$T_2$};
\draw[-] (1r1)--(1r11);
\draw[-] (1r1)--(1r4);
\draw[-] (1r11)--(1r3);
\draw[-] (1r11)--(1r41);
\draw[-] (1r4)--(1a);
\draw[-] (1r3)--(1s);
\draw[-] (1r41)--(1a1);
\end{tikzpicture}
\end{minipage}
\begin{minipage}[!t]{0.1\linewidth}
\[
\]
\end{minipage}
\begin{minipage}[!t]{0.9\linewidth}
\begin{tikzpicture}[node distance=15]
\node (S) {$r_2$};
\node(S1)[below of =S] {$r_3$};
\node (S0)[left of = S1, node distance  = 30] {$r_1$};
\node (B)[ right of = S1] {$r_6$};
\node (S3)[below of = S0, left of = S0] {$r_3$};
\node (A)[below of = S0, right of = S0] {$r_4$};
\node (s1)[below of = S1] {$s$};
\node (b)[below of = B] {$b$};
\node (a)[below of = A] {$a$};
\node (s2)[below of = S3] {$s$};
\node (T1)[below of = S, node distance = 60] {$T_3$};

\draw[-] (S)--(S0);
\draw[-] (S)--(S1);
\draw[-] (S)--(B);
\draw[-] (S0)--(S3);
\draw[-] (S0)--(A);
\draw[-] (S3)--(s2);
\draw[-] (A)--(a);
\draw[-] (S1)--(s1);
\draw[-] (B)--(b);

 \node (1S) [right of= S, node distance = 100]{$r_2$};
\node(1S1)[below of =1S] {$r_3$};
\node (1S0)[left of = 1S1, node distance  = 30] {$r_1$};
\node (1B)[ right of = 1S1,node distance = 30] {$r_5$};
\node (1S3)[below of = 1S0, left of = 1S0] {$r_3$};
\node (1A)[below of = 1S0, right of = 1S0] {$r_4$};
\node (1s1)[below of = 1S1] {$s$};

\node (1B1)[below of = 1B,left of=1B] {$r_3$};
\node (1S4)[below of = 1B,right of=1B] {$r_6$};
\node (1b)[below of = 1B1] {$s$};
\node (1s3)[below of = 1S4] {$b$};

\node (1a)[below of = 1A] {$a$};
\node (1s2)[below of = 1S3] {$s$};
\node (T2)[below of = 1S, node distance = 60] {$T_4$};

\draw[-] (1S)--(1S0);
\draw[-] (1S)--(1S1);
\draw[-] (1S)--(1B);
\draw[-] (1S0)--(1S3);
\draw[-] (1S0)--(1A);
\draw[-] (1S3)--(1s2);
\draw[-] (1A)--(1a);
\draw[-] (1S1)--(1s1);
\draw[-] (1B)--(1B1);
\draw[-] (1B)--(1S4);
\draw[-] (1B1)--(1b);
\draw[-] (1S4)--(1s3);
\end{tikzpicture}
\end{minipage}
\caption{An example of a grammar.}
\label{fig:rep}
\vspace{-0.4cm}
\end{figure}

\subsection{Identifying the top-$k$ workflows}
\label{sec:ranking:topK}

We conclude our discussion of ranking by presenting an efficient way
to retrieve, and compute the scores of, the top-$k$ workflows in a
repository.  Given a repository, a query $Q$, and an integer $k$, a
naive approach is to compute $score(G,Q)$ using, e.g., {\em OptScore},
sort grammars in decreasing order of score, and return the top-$k$.
Here, $OptScore$ may be executed on the entire repository, or only on
its promising subset, leveraging an inverted index that maps each
keyword to the set of workflows in which it appears.  Even assuming
that only the grammars matching $Q$ are considered (i.e., grammars for
which {\em Match(G,Q)} returns true), this naive approach will still
require us to compute $score(G,Q)$ for many more than $k$ grammars.

We use the {\em Threshold Algorithm} (TA)~\cite{FaginLN01} to limit
the number of score computations.  Our use of TA is based on the
observation that $score(G,Q) \leq \min_{X \subset Q} score(G,X)$.
In particular, $score(G,Q)$ is at most as high as the score of $G$ for
any single-keyword subset of $Q$.

We leverage this observation and build inverted lists, one per keyword
$a$, storing all grammars $G$ that match $a$, in decreasing order of
$score(G,\{a\})$.  Then, given a multi-keyword query $Q$, we access the
query-relevant lists sequentially in parallel, and compute
$score(G,Q)$ for the first $k$ grammars.  We refer to the current
$k^{th}$ highest score as $\theta$, and we update $\theta$ as the
algorithm proceeds.

We consider grammars in inverted list order, and, when an unseen
grammar $G$ is encountered, retrieve its entries from all inverted
lists with random accesses, and compute the score upper-bound $ub(G,Q)
= \min_{a \in Q} score (G,\{a\})$.  If $ub(G,Q) > \theta$, we compute
$score(G,Q)$ using an algorithm from Section~\ref{sec:ranking:algo}
and update $\theta$ if necessary.  TA terminates when the score
upper-bound of unseen grammars, computed as the minimum of current
scores in the relevant inverted lists, is lower than $\theta$.

\eat{
Given a query $Q$, TA algorithm sequentially accesses the index lists
each of which represents a keyword in $Q$. For each grammar $G$ seen
currently, it does a random access to retrieve all values from all
lists, which gives an upper bound. The threshold $\theta$ of TA
algorithm is the $k^{th}$ minimum score of grammars seen. For each
grammar the upper bound of which is greater than $\theta$, run
Algorithm~\ref{algo:AF_ranking} to get its score and update
$\theta$. Once all the values of some sequential access are less or
equal to $\theta$, which means the upper bound of all unseen grammars
is less than $\theta$, TA algorithm terminates.
}

\section{Result Presentation}
\label{sec:resultPresentation}
Since grammars may be large and complex objects, it is important to
develop presentation mechanisms that help the user understand where
keyword matches occur in the result
grammars. 
 Interestingly, a single grammar may match a query in many ways, more
than can be shown to a user. In Section~\ref{sec:ranking} we proposed
to compute probabilities of parse trees, and to use these
probabilities to rank grammars.  In this section, we build on this
idea and propose to choose the most probable parse trees that are
structurally irredundant. We refer to such trees as {\em
  representative parse trees}, and describe them in
Section~\ref{sec:pres:model}.  We then give an algorithm for finding
the top-$k$ representative parse trees for a given grammar in
Section~\ref{sec:pres:algo}.

\subsection{Representative parse trees}
\label{sec:pres:model}

Recursion gives rise to structural redundancy.  Consider the grammar
in Figure~\ref{fig:rep}, and its parse trees $T_1$ and $T_2$.  These
trees both match query $Q = \{s, a\}$.  Both trees fire the same
productions in the same order, and, while $T_1$ cuts through the
chase, $T_2$ loops by firing $r_1$ twice in sequence, and by
generating $a$ along the path $S \rightarrow A \rightarrow a$ twice.
The concept of a {\em representative parse tree} (\rpt for short),
which we define next, models the intuition that, while both trees
match the query in the same way (by firing the same productions in the
same order), $T_1$ is more concise than $T_2$.

For convenience, we will sometimes represent parse trees as bags of
paths, denoted $paths(T)$.  Recall that a path is a sequence of
productions that ends with a terminal, e.g., $r_1 r_3 s$ is a path in
$T_1$ in Figure~\ref{fig:rep}.  We can represent $T_1$ as $paths(T_1)
= \{ r_1 r_3 s,~r_1 r_4 a \}$. Note that the paths representation may
be ambiguous i.e. the same bag of paths may correspond to two
different parse trees (See Figure~\ref{fig:ambiguity}).

\begin{figure}[h]
 \begin{minipage}[!t]{0.3\linewidth}
\[
\begin{split}
r_1:S\rightarrow & SS\\
r_2: S\rightarrow &AB\\
r_3: A\rightarrow &a_1\\
r_4: A\rightarrow &a_2\\
r_5: B\rightarrow &a_1\\
r_6: B\rightarrow &a_3\\
\end{split}
\]
 \end{minipage}
\begin{minipage}{0.5\linewidth}
\begin{tikzpicture}
 \node (S) {$r_1$};
\node(center)[below of =S,node distance = 20] {};
\node (S1)[left of = center] {$r_2$};
\node (S2)[ right of = center] {$r_2$};
\node (A1)[below of = S1, left of = S1,node distance = 20] {$r_3$};
\node (B1)[below of = S1, right of = S1,node distance = 20] {$r_5$};
\node (A2)[below of = S2, left of = S2,node distance = 20] {$r_4$};
\node (B2)[below of = S2, right of = S2,node distance = 20] {$r_6$};
\node (a1)[below of = A1, node distance = 20] {$a_1$};
\node (a2)[below of = B1, node distance = 20] {$a_1$};
\node (a3)[below of = A2,node distance = 20] {$a_2$};
\node (a4)[below of = B2,node distance = 20] {$a_3$};
\node (T1)[below of = S, node distance = 80] {$T_1$};

 \node (1S) [below of =S,node distance = 100]  {$r_1$};
\node(1center)[below of =1S,node distance = 20] {};
\node (1S1)[left of = 1center] {$r_2$};
\node (1S2)[right of = 1center] {$r_2$};
\node (1A1)[below of = 1S1, left of = 1S1,node distance = 20] {$r_3$};
\node (1B1)[below of = 1S1, right of = 1S1,node distance = 20] {$r_6$};
\node (1A2)[below of = 1S2, left of = 1S2,node distance = 20] {$r_4$};
\node (1B2)[below of = 1S2, right of = 1S2,node distance = 20] {$r_5$};
\node (1a1)[below of = 1A1, node distance = 20] {$a_1$};
\node (1a2)[below of = 1B1, node distance = 20] {$a_3$};
\node (1a3)[below of = 1A2,node distance = 20] {$a_2$};
\node (1a4)[below of = 1B2,node distance = 20] {$a_1$};
\node (1T1)[below of = 1S, node distance = 80] {$T_2$};

\draw[-] (S)--(S1);
\draw[-] (S)--(S2);
\draw[-] (S1)--(A1);
\draw[-] (S1)--(B1);
\draw[-] (S2)--(A2);
\draw[-] (S2)--(B2);
\draw[-] (A1)--(a1);
\draw[-] (B1)--(a2);
\draw[-] (A2)--(a3);
\draw[-] (B2)--(a4);

\draw[-] (1S)--(1S1);
\draw[-] (1S)--(1S2);
\draw[-] (1S1)--(1A1);
\draw[-] (1S1)--(1B1);
\draw[-] (1S2)--(1A2);
\draw[-] (1S2)--(1B2);
\draw[-] (1A1)--(1a1);
\draw[-] (1B1)--(1a2);
\draw[-] (1A2)--(1a3);
\draw[-] (1B2)--(1a4);
\end{tikzpicture}
\end{minipage}
\caption{$paths(T_1)=paths(T_2)$ while $T_1\not=T_2$}
\label{fig:ambiguity}
\end{figure}
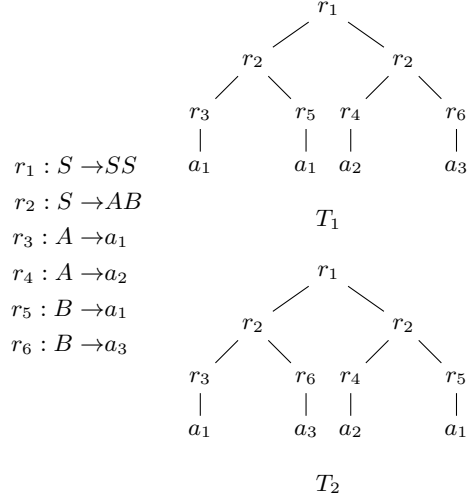

\begin{definition}\label{def:pathSubsumption}
\textbf{(Path Subsumption)} Path $p$ subsumes path $p'$,
  denoted $p\prec p'$ if $p$ is a sub-sequence of $p'$.
\end{definition}
For example, $r_1 r_3 s \prec r_1 r_1 r_3 s$.  Note that, since a path
ends with a terminal, if $p \prec p'$ then the paths must end in the
same terminal.  We use path subsumption to define the main concept of
this section, parse tree subsumption.

\begin{definition}\label{def:parseTreeSubsumption}
\textbf{(Parse Tree Subsumption)} Parse tree
    $T$ \emph{subsumes} parse tree $T^\prime$, denoted $T\prec
    T^\prime$, iff $head(root(T))=head(root(T'))$ and there exists
    an \emph{onto} mapping from $paths(T^\prime)$ to $paths(T)$, in
    which, if $p^\prime\in paths(T^\prime)$ is mapped to $p\in
    paths(T)$ then $p\prec p^\prime$.  
\end{definition}

For example, consider the parse trees in Figure~\ref{fig:rep}. Observe
that $T_1 \prec T_2$ according to
Definition~\ref{def:parseTreeSubsumption}.  The onto mapping from
$paths(T_2)$ to $paths(T_1)$ is: $r_1 r_1 r_3 s \rightarrow r_1 r_3
s$, $r_1 r_1 r_4 a \rightarrow r_1 r_4 a$ and $r_1 r_4 a \rightarrow
r_1 r_4 a$.  In contrast, no subsumption holds between parse trees
$T_2$ and $T_3$.


\begin{definition}
\textbf{(Representative Parse Tree)} A \emph{representative
      parse tree} (\rpt) $T$ of a grammar $G$ is a parse tree
    s.t. there does not exist a parse tree $T^\prime$ of $G$ that
    subsumes $T$.
\end{definition}

\eat{
\begin{example}
  Consider the grammar in Example.\ref{ex:ex1}, Fig.\ref{fig:ex3}
  shows three parse trees. $\\paths(T_2)=\{r_1r_2s_1, r_1r_4a_1,
  r_1r_4a_1, r_1r_5a_2\}$ $\\paths(T_3)=\{r_1r_2s_1, r_1r_4a_1,
  r_1r_5a_2, r_1r_5a_2\}$ $\\paths(T_4)=\\\{r_1r_1r_2s_1,
  r_1r_1r_4a_1,r_1r_1r_4a_1,r_1r_1r_5a_2, r_1r_4a_1, r_1r_4a_1,
  r_1r_5a_2\}$

  We can check that $T_2$ subsumes $T_4$. \eat{We denote by $(i,j)$
    the $i^{th}$ element of $paths(T^\prime)$ is mapped to the
    $j^{th}$ element of $paths(T)$.  There are two onto mapping
    functions from $paths(T_4)$ to $paths(T_2)$,
    i.e. $\\\{(0,0),(1,1),(2,1),(3,3),(4,2),(5,2),(6,3)\}$ and
    $\\\{(0,0),(1,1),(2,1),(3,3),(4,1),(5,2),(6,3)\}$} However, $T_2$
  does not subsume $T_3$, even though all paths in $T_3$ are
  subsumed. That's because two paths $r_1r_4a_1$ in $T_2$ subsume the
  same path $r_1r_4a_1$ in $T_3$ and hence there does not exist an
  onto mapping function from $paths(T_3)$ to $paths(T_2)$. It is easy
  to check both $T_2$ and $T_3$ are \rpts.

\end{example}
}

We now list several important properties of \rpts. Given a path $p$, we denote by $len(p)$ the length of $p$.

\begin{theorem}\label{the:matching}
  A grammar $G$ matches a query $Q$ iff there exists an \rpt $T$ in
  $G$ that generates $Q$.
\end{theorem}
\begin{proof}
The forward direction is trivial.  For the reverse, recall that if a grammar $G$ matches $Q$, then $Q$ must be contained in the leaves of some parse tree for
$G$. Then it must also be contained in the leaves of an \rpt, since an \rpt generates the same terminal set as the trees that it subsumes.
\end{proof}

\begin{lemma}
  If a parse tree is representative, then all its subtrees are also
  representative. 
\label{lm:bottomup}
\end{lemma}
\begin{proof}
Proof is by contradiction. Let $T$ be an \rpt. Suppose a subtree of $T$ denoted by $T_{sub}$ is subsumed by tree $T'_{sub}$. Let $T'$ be the tree obtained
from $T$ by replacing $T_{sub}$ with $T'_{sub}$. It is easy to see that $T'\prec T$, which contradicts that $T$ is an \rpt tree.
\end{proof}

The converse is not true, see $T_2$ in Figure~\ref{fig:rep} for a
  counter-example. One can verify that the two child subtrees of $T_2$ are \rpts. $T_2$ however is subsumed by $T_1$ and hence is not representative.

\begin{lemma}
  Given parse trees $T$ and $T'$, if $T \prec T'$ then $height(T)\leq
  height(T')$.
\label{lm:height}
\end{lemma}
\begin{proof}
 Let $p_{max}$ be one of the longest paths in $paths(T)$. Then $height(T)= len(p_{max})$.  Since $T\prec T'$, for any path $p\in paths(T)$, there is a path
$p'\in paths(T')$ s.t. $p\prec p'$. Let $p'$ be a path in $paths(T')$ that $p_{max}$ subsumes. Note that $len(p)\leq len(p')$ if $p\prec p'$. Then
$height(T)=len(p_{max})\leq len(p')\leq height(T')$. 
\end{proof}

Consider trees $T_3$ and $T_4$ in Figure~\ref{fig:rep}.  These trees
are of the same height, yet $T_3 \prec T_4$.  

\begin{lemma} If $T \prec
T'$ and $height(T) = height(T')$, then $T$ and $T'$ must be rooted at
the same production.
\end{lemma}
\begin{proof}
 Let $p_{max}$ be one of the longest paths in $T$. Since $T\prec T'$, there exists a path $p'\in paths(T')$ s.t. $p_{max}\prec p'$. Thus
$height(T)=len(p_{max})\leq len(p')\leq height(T')$.  Since $height(T)=height(T')$, $len(p_{max})= len(p')$. Recall that $p_{max}\prec p'$. Thus $p_{max}=p'$.
Since $p_{max}~(p')$ starts with the root production of $T~(T')$, $T$ and $T'$ are rooted at the same production.
\end{proof}

\begin{theorem}\label{the:subsumptionComplexity}
  Given trees $T$ and $T'$, the time complexity of checking if
  $T \prec T'$ is polynomial in tree size.
\end{theorem}
\begin{proof}
We prove the theorem by giving a polynomial algorithm of checking if $T\prec T'$. 

The algorithm starts by computing a mapping $f:paths(T')\rightarrow paths(T)^*$ such that for any path $p'\in paths(T')$, $f(p')\ni p$ iff $p\prec p'$. Note that if $\exists p'\in paths(T'),f(p')=\emptyset$ then $T\not\prec T'$. 

 We then build a flow network s.t. the network has maximum flow $|paths(T)|$ iff there exists a {\bf surjective} mapping $g:paths(T')\rightarrow
paths(T)$ s.t. if $g(p')=p$, $p\prec p'$. The flow network is built as follows.  Let $N=(V,E)$ be a network with $s,t$ being the source and the sink of $N$,
respectively. For each path $p\in paths(T)$ $(p'\in paths(T'))$, there is a node $p$ $(p')\in V$. For each $p\in paths(T)$, there is an edge from $p$ to sink
$t$ in $E$. For each $p'\in paths(T')$, there is an edge from source $s$ to $p'$ in $E$ and an edge from $p'$ to $p$ for any $p\in f(p')$. Every edge has a
capacity of $1$. It is easy to see that $N$ has maximum flow $|paths(T)|$ iff such $g$ exists.

It is easy to see that $T\prec T'$ iff $\forall p'\in paths(T'),f(p')\not=\emptyset$ and $g$ exists.

Note that given two paths $p,p'$, checking if $p\prec p'$ can be done in time $len(p)+len(p')$, i.e. $O(height(T'))$. It takes
$O(|paths(T)|*|paths(T')|*height(T'))$ to compute $f$. In the worst case, the size of $f$ could be $|paths(T)|*|paths(T')|$.  It takes
$O(|paths(T)|*|paths(T')|)$ to build the flow network. Using the Ford-Fulkerson algorithm \cite{Kleinberg:2005:AD:1051910}, it takes $O(|paths(T)|^2*|paths(T')|)$
to compute the maximum flow. In total, the algorithm has a time complexity $O(|paths(T)| *
|paths(T')| * height(T') + |paths(T)|^2 * |paths(T')|)$.
\end{proof}

\eat{
\begin{lemma}
 Given a terminal {\bf bag} $s$ of a
  grammar, the height of parse trees $T$ in which $leaves(T)=s$
  is bounded by $|s|\times(d(G)+1) + |s|$.
\end{lemma}
\begin{proof}
Proof is very similar to the proof of Lemma~\ref{part2}.
\end{proof}
}

\begin{lemma}
  Given a terminal {\bf set} of a grammar,
  the height of \rpts that generate it may be
  exponential in grammar size.
\end{lemma}
\begin{proof}
Consider the grammar:

\begin{tabular}{lllllllllll}
S$\rightarrow$&AS~|~s&&A$\rightarrow$&B~|~C&&B$\rightarrow$ &D\\
C$\rightarrow$&D&&D$\rightarrow$&E~|~F&&E$\rightarrow$ &a&&F$\rightarrow$&a
\end{tabular}

There are an exponential number of different paths (in the grammar size)  for $A$ to derive $\{a\}$, precisely $2^2=4$. By firing
$S\rightarrow AS$ once, we get one instance of $A$ and the
height increases by $1$. Consider a tree that is derived by firing $S\rightarrow AS$ four times where instances of $A$ derive $\{a\}$ distinctly. One can verify
that the tree is an \rpt, and that the height of the tree is exponential in the grammar size.
\end{proof}


\eat{
\begin{proof}
Suppose $T\prec T'$.  Given a node $u$ of tree $T$, let $r(u)$ be the production $u$ represents. Denote by $V(T)$ the set of nodes of $T$.  We now
show that we can find a mapping $g:V(T)\rightarrow V^*(T')$  where (1)  $\forall v'\in g(v), r(v)=r(v')$; and (2) for any two $u,v\in V(T)$, $g(u)\cap
g(v)=\emptyset$;
and (3) for each path $p=v_1\ldots v_n$ in $T$, there exists a path $p'=v_1'\ldots v_m'$ in $T'$ s.t. $\forall i=1\ldots n,  \exists  v'_{k_i}\in g(v_i) (0<
k_i\leq m)$ s.t. $v'_{k_i}$ is an ancestor of $v'_{k_j}$ if $k_i<k_j$. Observe that if $g$  exists, condition (2) guarantees $\rho(T)\geq\rho(T')$. e.g.
Revisiting trees $T_1,T_2$ in Figure~\ref{fig:rep}, one such mapping $g$ is shown below. For each node $u$ in $T_1$, $g(u)$ is the set of nodes in $T_2$ that
are underlined by the same mark as $u$.

\begin{figure}[h]
\centering
\begin{tikzpicture}[node distance=20]
\node(T1) {$T_1$};
\draw node[circle,draw, inner sep=0.5pt,minimum size=1pt, below of=T1] (r1) {$r_1$};
\node(center)[below of=r1]{};
\node(r3) [left of=center] {$\underline{r_3}$};
\node(r4) [right of=center] {$\uwave{r_4}$};
\node(s) [below of=r3] {$\uuline{s}$};
\node(p1) [below of=s,node distance=15] {$p_1$};
\node(a)[below of=r4]{$\dotuline{a}$};
\node(p2) [below of=a,node distance=15] {$p_2$};

\draw[-] (r1)--(r3);
\draw[-] (r1)--(r4);
\draw[-] (r3)--(s);
\draw[-] (r4)--(a);

\node(T2)[right of=T1, node distance=85] {$T_2$};
\node(1r1) [circle,draw, inner sep=0,minimum size=1, below of=T2] {$r_1$};
\node(1center)[below of=1r1]{};
\node(1r11)[left of=1center] {$r_1$};
\node(1r4)[right of=1center]{$\uwave{r_4}$};
\node(1center1)[below of=1r11] {};
\node(1r3) [left of=1center1] {$\underline{r_3}$};
\node(1r41) [right of=1center1] {$\uwave{r_4}$};
\node(1s) [below of=1r3] {$\uuline{s}$};
\node(1p1) [below of=1s,node distance=15] {$p'_1$};
\node(1a1)[below of=1r41]{$\dotuline{a}$};
\node(1p2) [below of=1a1,node distance=15] {$p'_2$};
\node(1a)[below of=1r4]{$\dotuline{a}$};
\node(1p3) [below of=1a,node distance=15]{$p'_3$};

\draw[-] (1r1)--(1r11);
\draw[-] (1r1)--(1r4);
\draw[-] (1r11)--(1r3);
\draw[-] (1r11)--(1r41);
\draw[-] (1r4)--(1a);
\draw[-] (1r3)--(1s);
\draw[-] (1r41)--(1a1);
\end{tikzpicture}
\end{figure}

Since $T\prec T'$, let $f$ be one onto mapping from $paths(T')$ to $paths(T)$ where $p\prec p'$ if $f(p')=p$,   and $f^{-1}:paths(T)\rightarrow paths(T')^*$
be the inverse function of $f$. We now show that we can find such $g$ when given $f$.

Given a node $v$ of $T$,  let $paths(v)$ be the bag of paths that go through $v$. Let $f^{-1}(paths(v))=\bigcup\limits_{p\in paths(v)}f^{-1}(p)$. Let
$T_{root}$ be the root node of $T$. $g$ maps $T_{root}$ to the set of nodes of $T'$ which are the highest common ancestors of paths $f^{-1}(paths(T_{root}))$
that represent production $r(T_{root})$. e.g. root of $T_1$ is mapped to root of $T_2$.
 For each node $u$ in $T$ whose ancestors have been mapped, $g$ maps $u$ to the set of nodes in $T'$ that are highest
common ancestors of  $f^{-1}(paths(u))$ while are descendants of at least one of  $g(parent(u))$ and represent $r(u)$. e.g. Consider $T_1,T_2$ above. There is
only one onto mapping $f$. Let $v$ be the node representing $r_4$ in $T_1$. Then $f^{-1}(paths(v))=\{p'_2,p_3'\}$. $g(parent(v))$ is the root of $T_2$. Thus
$g(v)$ is the set of nodes in $T_2$ that are underlined with tilde.

Note that the process guarantees for any $u,v\in paths(T)$ that $u$ is an ancestor of $v$, $g(u)\cap g(v)=\emptyset$. For each $u,v$ that neither of $u,v$
is an ancestor of each other, if $f^{-1}(paths(u))\cap f^{-1}(paths(u))=\emptyset$ then $g(u)\cap g(v)=\emptyset$. Now consider $u,v$ that $g(u)\cap
g(v)=\emptyset$.

Now consider a set of nodes $N$ where for
any $u,v\in N$, $f^{-1}(paths(u))\cap f^{-1}(paths(u))\not=\emptyset$ and for any $u\not\in N$, $\not\exists v\in N$, s.t. $f^{-1}(paths(u))\cap
f^{-1}(paths(u))\not=\emptyset$. We now prove that $|N|\leq |\bigcup\limits_{u\in N} g(u)|$

It is clear that $g$ is a one-to-one mapping iff
$paths(T)=paths(T')$.
\end{proof}}

\begin{lemma}
A tree may be subsumed by two different \rpts.
\end{lemma}
\begin{proof}
Consider grammar:

 \begin{tabular}{lllll}
$r_1: S\rightarrow$ & SS&&$r_2: S\rightarrow$ &AAA\\
$r_3: A\rightarrow$ & $a_1$&&$r_4: A\rightarrow$ & $a_2$
\end{tabular}

There are three trees $T_1=\{r_2r_3a_1,r_2r_3a_1,r_2r_4a_2\}$,
$\\T_2=\{r_2r_3a_1,r_2r_4a_2,r_2r_4a_2\}$, $T_3=\{r_1T_1,r_1T_2\}$
where $T_1$, $T_2$ are representative, and $T_1\prec T_3$, $T_2\prec
T_3$.
\end{proof}

The rest of this section is devoted to proving that all parse trees of a non-recursive grammar are representative (Theorem~\ref{the:non-recur}),  and that representative parse trees are at least as probable as any parse trees that they subsume (Theorem ~ \ref{the:linear-recur}, for linear-recursive grammars).  This forms the basis of
our result presentation algorithm, Algorithm~\ref{algo:repTrees}.

\subsubsection{Theorem~\ref{the:non-recur}}

We first extend parse tree subsumption to bags (Definition~\ref{def:bag_subsumption}) and forests (Definitions~\ref{def:forest} and \ref{def:forest_subsumption}). To obtain Theorem~\ref{the:non-recur}, we first consider a special case of non-recursive grammars whose productions have distinct symbols on the right-hand side (Lemma~\ref{lem:rpSub} and \ref{lem:rpG}) and then generalize to non-recursive grammars (Lemma~\ref{lem:frtSub}, \ref{lem:frtSize} and \ref{lem:frtG}).

\begin{definition}\label{def:bag_subsumption}
{\bf (Bag Subsumption)} We say a bag of paths {\em subsumes} another bag of paths denoted by $Bag\prec Bag'$ iff there exists an onto function $f:Bag'\rightarrow Bag$ s.t. $p\prec p'$ if
$f(p')=p$. 
\end{definition}

We say $T$ is a {\em tree of production} $r$ if $T$ is rooted at production $r$, and that $T$ is a {\em tree of symbol} $M$ if $T$ is rooted at a production of $M$. 

\begin{lemma}\label{lem:rpPath}
 Given a non-recursive grammar $G$, and any two of its trees of the same symbol, say $T,T'$, if $\exists {\cal S}'\subseteq paths(T')$ s.t. $paths(T)\prec {\cal
S}'$ then $T,T'$ must be rooted at the same production.
\end{lemma}
\begin{proof}
Note that since the grammar is non-recursive, any path in its parse trees has at most one production of the same symbol (otherwise, the grammar is recursive). 
Let $T,T'$ be rooted at production $r,r'$, respectively. Since $paths(T)\prec {\cal S}'$, $\exists p'\in paths(T'), p\in paths(T)$ s.t. $p\prec p'$. Then
$r$ must appear somewhere in $p'$. Recall that $T,T'$ are two trees of the same symbol. Thus $r=r'$. 
\end{proof}

\begin{lemma}\label{lem:rpSub}
 Consider a non-recursive grammar in which each production has distinct symbols on the right-hand side. If $T,T'$ are two of its trees of the same
symbol, then $\not\exists{\cal S}'\subset paths(T')$ s.t. $paths(T)\prec {\cal S}'$.
\end{lemma}

In other words, this lemma is saying that if $T,T'$ are of the same
symbol, then $\forall {\cal S}'\subset paths(T'), \forall g:{\cal S}'\rightarrow paths(T)$ (not necessarily a function), s.t. if $g(p')=p$ then $p\prec p'$, $g$
is not onto. Additionally, if
$T\nprec T'$,  then $\forall g:paths(T')\rightarrow paths(T)$ (not necessarily a function) s.t. $p\prec p'$ if
$g(p')=p$, $g$ is not onto.
The proof of this  will be given with the proof of Lemma~\ref{lem:rpG}. 

\begin{lemma}\label{lem:rpG}
 Given a non-recursive grammar $G$ in which each production has distinct symbols on the right-hand side, all its parse trees  are \rpts.
\end{lemma}
\begin{proof}
We prove this lemma and Lemma~\ref{lem:rpSub} simultaneously. (For this lemma, we basically prove that if two trees $T,T'$ of $G$ are s.t. $T\prec T'$ then
$paths(T)=paths(T')$). Proof is by induction over the height of $T$. W.l.o.g., we assume each production has exactly two symbols on the right-hand
side.

{\bf Basis: } Consider $height(T)=1$, i.e. $T$ is rooted a production the right-hand side of which consists of terminals. 

{\em Lemma~\ref{lem:rpSub}:} We prove Lemma~\ref{lem:rpSub} by contradiction. Let $T'$ be a tree of the same symbol as $T$ and $paths(T)$
subsumes some proper subset of $paths(T')$. By Lemma~\ref{lem:rpPath}, $T,T'$ must be rooted at the same production. Then
$paths(T)=paths(T')$, which contradicts $paths(T)\prec {\cal S}'$. 

{\em Lemma~\ref{lem:rpG}:} If $T\prec
T'$, then by Lemma~\ref{lem:rpPath}, $T'$ is rooted at the same production as $T$. Thus $paths(T')=paths(T)$.

{\bf Inductive Step: } Suppose  Lemma~\ref{lem:rpSub} and  Lemma~\ref{lem:rpG} hold for $height(T)\leq k$. Consider when $height(T)=k+1$.

{\em Lemma~\ref{lem:rpSub}:} We  prove Lemma~\ref{lem:rpSub} holds for  $height(T)= k+1$ by contradiction. Suppose $paths(T)$
subsumes some proper subset ${\cal S}'$ of $paths(T')$.  Then by Lemma~\ref{lem:rpPath}, $T,T'$ must be rooted at the same production.  Since each production
in $G$ has distinct
symbols on the right-hand side, let $r:M\Rightarrow AB$ be the root production of $T,T'$. Let $T_A,T_B$ ($T'_A,T'_B$) be the child subtrees of $T$ ($T'$) where
$T_A$ ($T'_A$) is a tree of $A$, $T_B$ is a tree of $B$, (see
Figure~\ref{fig:partOfT}). Let $f:{\cal S}'\rightarrow paths(T)$ be one onto mapping s.t. $p\prec p'$ if
$f(p')=p$.  Let $f^{-1}$ be the inverse function of $f$. 

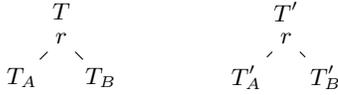
\begin{figure}[h]
\centering
\begin{tikzpicture}[node distance=15]
\node(T1) {$T$};
\draw node[below of=T1, node distance=10] (r1) {$r$};
\node(center)[below of=r1, node distance=15]{};
\node(r3) [left of=center] {$T_A$};
\node(r4) [right of=center] {$T_B$};

\draw[-] (r1)--(r3);
\draw[-] (r1)--(r4);

\node(T2)[right of=T1, node distance=85] {$T'$};
\node(1r1) [below of=T2, node distance=10] {$r$};
\node(1center)[below of=1r1]{};
\node(1r11)[left of=1center] {$T_A' $};
\node(1r4)[right of=1center]{$T_B'$};

\draw[-] (1r1)--(1r11);
\draw[-] (1r1)--(1r4);

\end{tikzpicture}
\vspace{-0.5mm}
\caption{Part of $T,T'$}
\label{fig:partOfT}
\end{figure}

We now list all possible situations below and then show all these cases will lead to contradictions.

\begin{enumerate}\centering
\item[Case $X_1:$] $ T_A\prec T_A' \wedge T_B\prec T_B'$
\item[Case $X_2:$] $T_A\nprec T_A' \wedge T_B\prec T'_B$
\item[Case $X_3:$] $T_A\prec T_A' \wedge T_B\nprec T'_B$
\item[Case $X_4:$] $T_A\nprec T_A' \wedge T_B\nprec T_B'$
\end{enumerate}

\begin{enumerate}
 \item[Case $X_1.$] Since $height(T_A)\leq k$ and $height(T_B)\leq k$,  by hypothesis of Lemma~\ref{lem:rpG}, $paths(T_A)=paths(T_A')$ and
$paths(T_B)=paths(T_B')$, i.e. $paths(T)=paths(T')$. Clearly, it contradicts the assumption that $paths(T)$
subsumes ${\cal S}'\subset paths(T')$.

\item[Case $X_2$.]  Since $T_A\nprec T_A'$, and by hypothesis of Lemma~\ref{lem:rpSub},  $f$ has to map some path of $T_B'$ to some path in $T_A$, i.e $B$
(indirectly) derives $A$. However, since $height(T_B)\leq k$,  by hypothesis of Lemma~\ref{lem:rpG},
$paths(T_B)=paths(T_B')$. Since the grammar is non-recursive, $f^{-1}$ maps all paths in $T_B$ to paths in $T_B'$. Thus there is a path in $T_B$ that $f$ does not map paths over, which contradicts $f$ is onto. 

\item[$Case X_3$.] Same as $X_2$.

\item[$Case X_4$.] If $f$ does not map any path in $T_B'$ to paths in $T_A$, since $T_A\nprec T_A'$, then
$paths(T_A)$ subsumes some
proper subset of $paths(T_A')$, which contradicts the hypothesis of Lemma~\ref{lem:rpSub}. Thus $f$ maps some path in $T_B'$ to a path in $T_A$, i.e. $B$
(indirectly) derives $A$. Similarly, for $T_B$, $f$ maps some path in $T_A'$ to a path in $T_B$, i.e. $A$
(indirectly) derives $B$, which contradicts the grammar is non-recursive.
\end{enumerate}

{\em Lemma~\ref{lem:rpG}:} Consider when $height(T)=k+1$. 

Note that if $T\prec T'$ then by Lemma~\ref{lem:rpPath} $root(T)=root(T')$.  We reuse Figure~\ref{fig:partOfT}. Since
$T\prec T'$,
let $f:paths(T')\rightarrow paths(T)$ be one onto mapping s.t. $p\prec p'$ if $f(p')=p$. 

Again, consider the four cases $X_1,X_2,X_3,X_4$. We now argue only $X_1$ could be true.

\begin{enumerate}
 \item [Case $X_1$.] Since $height(T_A)\leq k$ and $height(T_B)\leq k$,  by hypothesis $paths(T_A)=paths(T_A')$ and $paths(T_B)=paths(T_B')$. 
Then
 Lemma~\ref{lem:rpG} holds for $T$. 

\item[Case $X_2$.] Now since $T_B\prec T_B'$,  $height(T_B)\leq k$,  by hypothesis
$paths(T_B)=paths(T_B')$. Since $T_A\nprec T_A'$ and $T\prec T'$, either 
  \begin{enumerate}
  \item 
  functions exist from $paths(T'_A)$ to $paths(T_A)$ where $p\prec p'$ if $p'$ is mapped to $p$,  but none is {\bf onto}, i.e. 
  $ \exists p_1\in paths(T_A), p_2'\in paths(T'_B)$ s.t. $f(p_2')=p_1$, i.e. $B$
  (indirectly) derives $A$; or 

  \item  functions do not exist from $paths(T'_A)$ to $paths(T_A)$ where $p\prec p'$ if $p'$ is mapped to $p$, i.e. $\\\exists p'_1\in paths(T'_A), p_2\in
paths(T_B)$ s.t. $f(p_1')=p_2$, i.e. $A$ (indirectly) derives  $B$.
  \end{enumerate}

Note that (a) and (b) cannot both be true since the grammar is non-recursive. We now prove that both (a) and (b) will lead to contradictions.

We first consider (a). Since
$paths(T_B)=paths(T'_B)$, let $p^*_2\in paths(T_B)$ be the corresponding path of $p_2'$. Note that $f$ is onto, then $f$ has to map some
path in $T_A'$ to $p^{*}_2$, i.e. $A$ (indirectly) derives $B$, which contradicts the grammar is non-recursive. 

Now consider (b). Since $A$ derives $B$ and the grammar is non-recursive, $f$ maps all paths in $T_B'$ to  paths in $T_B$. Since $f$ maps
some path in $T_A'$ to $T_B$, $paths(T_A)$ subsumes a proper subset of $paths(T_A')$, which contradicts the hypothesis of Lemma~\ref{lem:rpSub}.


\item[Case $X_3$.] Same as $X_2$.
\item[Case $X_4$.] Since $T\prec T'$ and $T_A\nprec T_A'$, by hypothesis of Lemma~\ref{lem:rpSub}, $f$ has to map a path in $T_B'$ to some path in $T_A$,
i.e. $B$ (indirectly) derives $A$. Similarly, $f$ has to map a path in $T_A'$ to some path in $T_B$, i.e.
 $A$ (indirectly) derives $B$. It contradicts the grammar is non-recursive.
\end{enumerate}

Now we withdraw the restriction that each production has exactly two symbols on the right-hand side. The basis still holds. Now consider the inductive
step. Let $r:M\Rightarrow \alpha_1\ldots \alpha_n$ be the root production of $T$. Denote by $T_{\alpha_i}$ ($T_{\alpha_i}'$) the child subtree of $T$ ($T'$)
that is rooted at some production of $\alpha_i$. We need to consider all the following two cases.
\begin{enumerate}\centering
\item[Case $X_1':$]  $\bigwedge\limits_{i=1}^n (T_{\alpha_i}\prec T'_{\alpha_i})$
\item[Case $X_2':$]  $(\exists i, T_{\alpha_i}\nprec T_{\alpha_i}')$
\end{enumerate}
  Again, $X'_1$ indicates
$paths(T)=paths(T')$ by hypothesis. Thus both Lemma~\ref{lem:rpSub} and Lemma~\ref{lem:rpG} hold.

We now argue $X'_2$ will lead to contradictions for both Lemma~\ref{lem:rpSub} and Lemma~\ref{lem:rpG}.  W.l.o.g., suppose $T_{\alpha_1}\nprec T_{\alpha_1}'$. 

We start with proving Lemma~\ref{lem:rpSub} by contradiction. Suppose $paths(T)$ subsumes a proper subset $\cal S$' of some tree $T'$ which is of the same symbol
as $T$. Then $T'$ is also rooted at $r$. Let $f:{\cal S}'\rightarrow paths(T)$ be one onto mapping s.t. $p\prec p'$ if $f(p')=p$. By hypothesis, w.l.o.g., $f$
has to map a path in $paths(T_{\alpha_2}')$ to some path in $T_{\alpha_1}$, i.e. $\alpha_2$ (indirectly) derives $\alpha_1$. Similarly, by hypothesis and the
grammar is non-recursive, w.l.o.g., $f$ has
to map a path in $paths(T_{\alpha_3}')$ to some path in $T_{\alpha_2}$, i.e. $\alpha_3$ (indirectly) derives $\alpha_2$ and so on. Now consider $\alpha_n$. Note
that all $\alpha_1\ldots \alpha_{n-1}$ can (indirect) derive $\alpha_n$. So there is a path in $T_{\alpha_n}$ that $f$ does not map paths over, which contradicts
$f$ is onto.

We now prove Lemma~\ref{lem:rpG}. Since $T\prec T'$, $T'$ is also rooted at $r$. Let $f:paths(T')\rightarrow paths(T)$ be one onto mapping s.t. $p\prec p'$
if $f(p')=p$. Since $T_{\alpha_1}\nprec T'_{\alpha_1}$ and by hypothesis of Lemma~\ref{lem:rpSub}, w.l.o.g., $f$
has to map a path in $paths(T_{\alpha_2}')$ to some path in $T_{\alpha_1}$, i.e. $\alpha_2$ (indirectly) derives $\alpha_1$ and so on. Finally, there is a path
in $T_{\alpha_n}$ that $f$ does not map paths over, which contradicts $f$
is onto.
\end{proof}

We now extend Lemma~\ref{lem:rpSub} and Lemma~\ref{lem:rpG} to
grammars in which a production can have multiple occurrences of a
symbol on the right-hand side. To get there we first introduce some notions.

\begin{definition}\label{def:forest}
{\bf (Forest)} Given a grammar, a {\em forest} $F$ is a bag of parse
trees. We use $paths(F)$ for the bag union of paths of trees in $F$;
$height(T)$ for the maximum height of trees in $F$; and  $|F|$  for the number of
trees in $F$. 
\end{definition}

\begin{definition}\label{def:forest_subsumption}
{\bf (Forest Subsumption)} Given a grammar and two of its forests
$F,F'$, we say $F$ {\em subsumes} $F'$, denoted by $F\prec F'$ iff
$paths(F)\prec paths(F')$.
\end{definition}

We say $F$ is a forest of symbol $M$ iff all trees in $F$ are rooted
at productions of $M$. We say $F$ is a forest of production $r$ iff all trees in $F$ are rooted
at $r$.

\begin{lemma}\label{lem:frtSub}
 Given a non-recursive grammar  and any two of its forests $F,F'$, if
 $F,F'$ are of the same symbol and the same size, then  
 $\not\exists{\cal S}'\subset paths(F')$
s.t. $paths(F) \prec {\cal S}'$.
\end{lemma}
\begin{proof}
Proof can be found in Lemma~\ref{lem:frtG}. 
Note that Lemma~\ref{lem:rpSub} is a special case, i.e. $|F|=1$. 
\end{proof}

\begin{lemma}\label{lem:frtSize}
Given a non-recursive grammar and any two of its forests $F,F'$, if $F,F'$
are of the same symbol and $\exists {\cal S}'\subseteq paths(T')$
s.t. $paths(F)\prec {\cal S}'$, then $|F|\leq |F'|$.
\end{lemma}
\begin{proof}
We defer the proof to Lemma~\ref{lem:frtG}.
\end{proof}

\begin{lemma}\label{lem:frtG}
Given a non-recursive grammar and any two forests $F,F'$, if $F,F'$
are of the same symbol and the same size and $F\prec F'$, then $paths(F)=paths(F')$.
\end{lemma}
\begin{proof}
We prove Lemma~\ref{lem:frtSub}, Lemma~\ref{lem:frtSize} and this
lemma by induction over $height(F)$. W.l.o.g., we
assume each production has exactly two symbols on the right-hand side.

{\bf Basis: } Consider when $height(F)=1$, i.e. all trees in $F$ are
rooted at productions the right-hand side of which are all
terminals. Let $F'$ be a forest of the same symbol as $F$.

{\em Lemma~\ref{lem:frtSub}} If $F'$ is of the same size as $F$ and some proper subset ${\cal S}'$ of $paths(F')$ can be subsumed
by $paths(F)$, since the grammar is non-recursive,
$paths(F)=paths({\cal S}')$. Since ${\cal
  S}'\subset paths(F')$, $|F|<|F'|$. It contradicts $|F|=|F'|$.

Similarly, {\em Lemma~\ref{lem:frtSize}} \eat{ If some subset ${\cal S}'$ of $paths(F')$ can be subsumed
by $paths(F)$, since the grammar is non-recursive,
$paths(F)\subseteq paths({\cal S}')$. Since ${\cal
  S}'\subseteq paths(F')$, $|F|\leq |F'|$.}
and {\em Lemma~\ref{lem:frtG}} hold.

{\bf Inductive Step: }  Suppose Lemma~\ref{lem:frtSub},
\ref{lem:frtSize} and \ref{lem:frtG} hold for
$height(F)\leq k$. We now prove they also hold for $height(T)=k+1$. Let $F_i$ ($i=1\ldots n$) be the forest consisting of trees
in $F$ which are of the same production.

{\em Lemma~\ref{lem:frtSub}}. We prove by contradiction. Let $F'$ be a
forest of the same symbol and the same size as $F$ and some proper
subset ${\cal S}'$ of $paths(F')$ can be subsumed by
$paths(F)$. Since $paths(F)\prec {\cal
  S}'$ and the grammar is non-recursive, we can find $F_i'\subseteq
F'$ ($i=1\ldots n$) s.t.  $F_i'$ is of the same
production as $F_i$. Specifically, let $F_i'$ consist of trees 
some path of which appears in ${\cal S}'$ and rooted at the same
production as $F_i$.  Note that no path in $F_i$ can subsume paths in
$F_j'$ if $i\not=j$, since the grammar is non-recursive. Then $\forall
i$, $\exists {\cal S}_i'\subseteq paths(F_i')$ s.t. $paths(F_i)\prec
{\cal S}_i'$. And  $\exists
i$, $\exists {\cal S}_i'\subset paths(F_i')$ s.t. $paths(F_i)\prec
{\cal S}_i'$. We first prove $\forall i$, $|F_i|=|F_i'|$.

Consider $F_i$ (for any $i$). Let $r:M\Rightarrow AB$ be the root
production of $F_i$.  Let $F_A$ ($F_B$) consist of child subtrees of trees in
$F_i$ that are rooted at productions of $A$ ($B$).  Let $F_A'$ ($F_B'$) consist of child subtrees of trees in
$F_i'$ that are rooted at productions of $A$ ($B$). 
Note that $height(F_A)\leq k$, $height(F_B)\leq
k$. Consider the following cases:

\begin{enumerate}\centering
\item[Case $X_0:$] $ A=B$
\item[Case $X_1:$] $ A\not=B \wedge F_A\prec F_A'  \wedge F_B\prec F_B'$
\item[Case $X_2:$] $A\not=B \wedge F_A\nprec F_A'  \wedge F_B\prec F_B'$
\item[Case $X_3:$] $A\not=B \wedge F_A\prec F_A'  \wedge F_B\nprec F_B'$
\item[Case $X_4:$] $A\not=B \wedge F_A\nprec F_A'  \wedge F_B\nprec F_B'$
\end{enumerate}

Consider case $X_0$. By assumption, $\exists {\cal S}''\subseteq
paths(F'_A\cup F'_B)$ s.t. $paths(F_A\cup F_B)\prec {\cal S}''$. By
hypothesis of Lemma~\ref{lem:frtSize}, $|F_A\cup F_B|\leq |F_A'\cup
F_B'|$. Thus $|F_i|\leq |F_i'|$. Consider case $X_1$. Since $F_B\prec F_B'$, by hypothesis of
  Lemma~\ref{lem:frtSize}, $|F_B|\leq |F_B'|$. Thus $|F_i|\leq
  |F_i'|$. Case $X_3$ is similar to $X_2$. Consider case $X_4$. Since $paths(F_i)\prec {\cal S}_i'$, $B$ derives $A$ and $A$
  derives $B$, which contradicts the grammar is non-recursive.

Thus, $\forall i$, $|F_i|\leq |F_i'|$.  Note that by assumption,
$\sum\limits_{i}|F_i|=|F|=|F'|\geq \sum\limits_{i}|F_i'|$. Thus
$\forall i, |F_i|=|F_i'|$.

Recall that  $\exists
i$, $\exists {\cal S}_i'\subset paths(F_i')$ s.t. $paths(F_i)\prec
{\cal S}_i'$. Again, we consider the five cases
$X_0,X_1,X_2,X_3,X_4$ and show all these cases lead to contradictions. For $X_0$, by assumption, $\exists {\cal
  S}''\subset paths(F_A'\cup F_B')$ s.t. $paths(F_A\cup F_B) \prec
{\cal S}''$. Note that $|F_A\cup F_B|=|F_A'\cup F_B'|$ since
$|F_i|=|F_i'|$. It contradicts the hypothesis of
Lemma~\ref{lem:frtSub}. For $X_1$, by hypothesis of
Lemma~\ref{lem:frtG}, $paths(F_A)=paths(F_A')$ and
$paths(F_B)=paths(F_B')$. It contradicts the assumption. For
$X_2,X_3,X_4$, they all
contradict the grammar is non-recursive as Lemma~\ref{lem:rpG}.

{\em Lemma~\ref{lem:frtSize}.}  Again, let $F_i'$ be the forest of
trees some path of which appears in ${\cal S}'$ and rooted at the same
production as $F_i$. Consider $F_i$ (for any $i$). Let $r:M\Rightarrow AB$ be the root production of
$F_i$. $F_A,F_B,F_A',F_B'$ are defined as above. Consider the five cases listed above. For $X_1$, by assumption,
$\exists {\cal S}''\subseteq paths(F_A'\cup F_B')$
s.t. $paths(F_A\cup F_B) \prec {\cal S}''$. By hypothesis of
Lemma~\ref{lem:frtSize}, $|F_A\cup F_B|\leq |F_A'\cup F_B'|$. Thus
$|F_i|\leq |F_i'|$. For $X_1$, by hypothesis of Lemma~\ref{lem:frtSize},
$|F_A|\leq |F_A'|$. Thus $|F_i|\leq |F_i'|$. $X_2,X_3,X_4$ lead
to contradictions as Lemma~\ref{lem:rpG}. Thus $|F|=\sum\limits_{i}|F_i|\leq \sum\limits_i
|F_i'|\leq |F'|$. 

{\em Lemma~\ref{lem:frtG}.}  Again, let $F_i'$ be the forest of
trees in $F'$ that are rooted at the same
production as $F_i$. Let $F$ be rooted at $r:M\Rightarrow AB$. $F_A,F_B,F_A',F_B'$ are defined as above.  By assumption, $F_i\prec F_i'$. By hypothesis of
Lemma~\ref{lem:frtSize}, $|F_i|\leq |F_i'|$. Note that $\sum\limits_i
|F_i|=\sum\limits_i |F_i'|$ by assumption. Thus $\forall i,
|F_i|=|F_i'|$.  Consider the five cases again. For $X_0$, by assumption
$F_A\cup
F_B\prec F_A' \cup F_B' $. By hypothesis of Lemma~\ref{lem:frtG},
$paths(F_A\cup F_B)=paths(F_A'\cup F_B')$. Thus
$paths(F)=paths(F')$. Similar for $X_1$. $X_2,X_3,X_4$ will lead
to contradictions as Lemma~\ref{lem:rpG}.

We now withdraw the restriction to allow productions have arbitrary
number of symbols on the right-hand side. The argument is similar to Lemma~\ref{lem:rpG}.
\end{proof}

\begin{theorem}\label{the:non-recur}
 All parse trees of non-recursive grammars are  \rpts.
\end{theorem}
\begin{proof}
Given two trees $T,T'$, if $T\prec T'$ then forest $F=\{T\}$ subsumes
forest $F'=\{T'\}$. By Lemma~\ref{lem:frtG}, $paths(F)=paths(F')$.
\end{proof}

\subsubsection{Theorem~\ref{the:linear-recur}}

We now wish to prove that a representative parse tree is as at least as probable as any parse tree that it subsumes.  Although for general grammars this is not always the case
(see the example below), for an important subclass of grammars called {\em linear recursive} it is true.   Since  most real world workflows are linear-recursive  \cite{DBLP:conf/sigmod/BaoDM11},  we will base our top-k \rpts algorithm (presented in the next subsection) on finding representative parse trees.

\begin{example}\label{ex:gen_pro}
Consider the following grammar:

\noindent
\begin{tabular}{l l}
$r_1: S\Rightarrow \{S,S,S\}~(0.01)$&$r_2:S\Rightarrow \{A,B\}~(0.09)$\\
$r_3: S\Rightarrow \{s_1\}~(0.9)$&\\
$r_4:A\Rightarrow \{A,A\}~(0.5)$&$r_5:A\Rightarrow \{a\}~(0.5)$\\
$r_6: B\Rightarrow \{B,B\}~(0.5)$&$r_7:B\Rightarrow \{b\}~(0.5)$
\end{tabular}

$T,T'$ are two of its parse trees, shown below. Note that $T\prec T'$ while $\rho(T)<\rho(T')$.
\begin{figure}[h]
\begin{subfigure}{0.2\textwidth}
\begin{tikzpicture}
\tikzset{level distance=20}
\tikzset{level 1/.style={sibling distance=30}}
\tikzset{level 2/.style={sibling distance=15}}
\tikzset{level 3/.style={sibling distance=5}}
\node (T) {$r_1$}
child { node {$r_2$} 
          child { node {$r_5$} child {node {$a$}}}
          child {node {$r_7$} child {node {$b$}}}
    }
child { node {$r_2$} 
          child { node {$r_5$} child {node {$a$}}}
          child {node {$r_7$} child {node {$b$}}}
    }
child {node {$r_3$} child {node {$s_1$}}};
\node[below of=T, node distance=80] {$T$};
\end{tikzpicture}
\end{subfigure}
\begin{subfigure}{0.2\textwidth}
\begin{tikzpicture}
\tikzset{level distance=20}
\tikzset{level 1/.style={sibling distance=30}}
\tikzset{level 2/.style={sibling distance=20}}
\tikzset{level 3/.style={sibling distance=10}}
\tikzset{level 4/.style={sibling distance=5}}
\node (T) {$r_1$}
child { node {$r_2$}
          child {node {$r_4$}
                   child { node {$r_5$} child {node {$a$}}}
                   child {node {$r_5$} child {node {$a$}}}}
          child {node {$r_6$}
                   child { node {$r_7$} child {node {$b$}}}
                   child {node {$r_7$} child  {node {$b$}}}}
     }
child { node {$r_3$} child {node {$s_1$} }}
child {node {$r_3$} child {node {$s_1$}}
};
\node[below of=T, node distance=90] {$T'$};
\end{tikzpicture}
\end{subfigure}
\caption{$\rho(T)=0.000005<\rho(T')=0.000011$}
\end{figure}
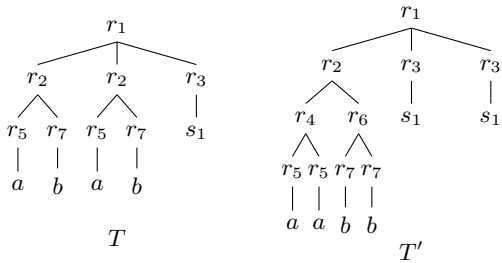
\end{example}

\begin{definition}
{\bf (Linear-Recursive Grammars)} A grammar \G is {\em linear-recursive} iff for each production $r: M\Rightarrow \alpha_1\ldots \alpha_n \in R$, at most one symbol on the right-hand side of $r$ can derive $M$.  
\end{definition}

The grammar in Example~\ref{ex:gen_pro} is not linear-recursive,
because production $r_1$ contains $3$ symbols on the right-hand side
which can derive $S$.   Productions $r_4$ and $r_6$ are also problematic.
Note that linear-recursive grammars are s.t. for each symbol $M$,
every (partial) execution (not derivation) derived by $M$ contains at most one $M$.

\begin{theorem}\label{the:linear-recur}
Given a linear-recursive grammar, and any two of its parse trees, $T,T'$, if $T\prec T'$ then $\rho(T)\geq \rho(T')$.
\end{theorem}
\begin{proof}
  Let $V(T,r)$ be the set of nodes in $T$ which represent production
$r$. Note that the probability of a tree $T$ is $\rho(T)=\prod\limits_{r}\rho(r)^{|V(T,r)|}$.
We now prove that $\forall r$, $|V(T,r)|\leq |V(T',r)|$.

 Given
a set of nodes ${\cal S}$ of a tree $T$, let $top({\cal S}, i)\subseteq {\cal S}$
($i=0\ldots height(T)$) be
the set of nodes in which every node has exactly $i$ ancestors  in ${\cal
  S}$. Note that $top({\cal S}, 0)$ is 
the set of nodes in $\cal S$ none of the ancestors of which is in
$\cal S$. Let $F(top(V(T,r),i))$ be the forest
consisting of subtrees rooted at nodes in $top(V(T,r),i)$. Since $T\prec T'$, $\forall r,i$, $\exists S'\subseteq paths(F(top(V(T',r),i))$ s.t.
$paths(F(top(V(T,r),i)))\prec S'$. 
By
Corollary~\ref{cor:linear-recur}, $|top(V(T,r),i)|\leq
|top(V(T',r),i)|$. Note that since the nodes are from a tree,  $\forall i\not=j$, $top(V(T,r),i)\cap
top(V(T,r),j)= \emptyset$ and $top(V(T',r),i)\cap
top(V(T',r),j)= \emptyset$. Thus $\forall r, |V(T,r)|\leq
|V(T',r)|$.
\end{proof}

\begin{lemma}\label{lem:linear_size}
Given a linear-recursive grammar, and any two of its forests  $F,F'$, if $F,F'$
are of the same production and $F\prec F'$ then $|F|<|F'|$.
\end{lemma}
\begin{proof}
Proof is by induction over $height(F)$. To illustrate, we assume $F$
consists of two trees $T_1, T_2$, $F'$ consists of one tree $T'$.

{\bf Basis: } Consider $height(F)=1$, i.e. $T_1,T_2$ are rooted at the same production the right-hand side of which consists of terminals. Suppose such $T'$ exists. Then by assumption, $T'$ is rooted at the same production, $path(T_1)=paths(T')$, which contradicts $paths(T_1)\cup paths(T_2)\prec paths(T')$. 

{\bf Inductive Step: } Suppose the lemma holds for
$height(F)\leq k$. Now consider when $height(F)=k+1$. We
prove by contradiction. W.l.o.g, we assume each production has exactly
two symbols on the right-hand side. Let $T_1,T_2,T'$ be rooted at
production $r:M\Rightarrow AB$ and $F\prec F'$. Let $T_{A1}$ ($T_{A2}$,
$T_A'$), $T_{B1}$ ($T_{B2}$, $T_B'$) be the child subtrees of $T_1$
($T_2$, $T'$), which are rooted at productions of $A$, $B$,
respectively (see below).

\begin{figure}[h]
\centering
\begin{tikzpicture}[level distance=20, sibling distance=25]
\node (T1){$r$} 
 child {node {$T_{A1}$}} child {node {$T_{B1}$}};
\node [below of = T1, node distance =40] {$T_1$};
\node (T2)[right of = T1, node distance=60]{$r$} 
 child {node {$T_{A2}$}} child {node {$T_{B2}$}};
\node [below of = T2, node distance =40] {$T_2$};
\node (T')[right of=T1, node distance = 150]{$r$} 
 child {node {$T_{A}'$}} child {node {$T_{B}'$}};
\node [below of = T', node distance =40] {$T'$};
\end{tikzpicture}
\end{figure}

Since the grammar is linear-recursive, consider the following cases:

\begin{enumerate}
\item[$X_1$:] $A$ can derive $M$, $B$ cannot derive $A$ or $M$
\item[$X_2$:] $A,B$ cannot derive $M$
\end{enumerate}

Consider termination productions of $A$. We say a production $r_A$ is
a termination production of $A$, if (1) $head(r_A)=A$ and $\not
\exists A'\in body(r_A)$ s.t. $A'$ can derive $A$; or (2)
$head(r_A)$ and $A$ are mutually derivable, and $\not\exists A' \in
body(r_A)$ s.t. $A'$ can derive $A$. e.g. for grammar $r_1:
S\Rightarrow \{C,s\}~r_2: S\Rightarrow \{s\}~r_3: C\Rightarrow \{S,a\}~ r_4:C\Rightarrow
\{a\}$, $r_2,r_4$ are termination productions of both $S$ and $C$. Intuitively, in linear-recursive
grammars, a recursion is fired many times and then ends with a
termination production to get out of that recursion. Note that if two
symbols are mutually derivable, their termination productions are the same.
Since the grammar is proper,
termination productions exist for each symbol. 

Let $r_{A1}$ be the termination production in $T_{A1}$. Note that
since (the grammar is linear-recursive) every (partial) execution (not
derivation) of $A$ contains at most one $A$, $T_{A1}$ must have exactly one
termination production of $A$. 
Similarly, let $r_{A2}$, $r_A'$,
$r_{B1}$, $r_{B2}$, $r_{B}'$ be the termination production of
$T_{A2}$, $T_A'$, $T_{B1}$, $T_{B2}$ and $T_{B}'$, respectively.

Case $X_1$. Since $F\prec F'$, and $B$ can not
derive $A$, $r_{A1},r_{A2}$ must appear in $T_A'$. As has argued,
$r_A'$ is the only termination production of 
$T_A'$, so $r_{A1}=r_{A2}=r_A'$. Moreover, if let $T^{t}_{A1}$ be the
subtree of $T_{A1}$ that is rooted at the termination production of
$A$, then $F_1=\{T_{A1}^t, T_{A2}^t\}\prec F_2=\{T^{t'}_A\}$ which contradicts the
hypothesis.

Case $X_2$. By assumption, $r_{A1},r_{A2},r_{B1},r_{B2}$ must appear
in $T'$. Similar to Case $X_1$, $r_{A1}$, $r_{A2}$ can not both appear
in $T_A'$. W.l.o.g. let $r_{A1}$, $r_{A2}$ appear in $T_A'$, $T_B'$,
respectively. Similarly, $r_{B1}$, $r_{B2}$ appears in $T_A'$, $T_B'$,
respectively. Thus $A$, $B$ are mutually derivable.  Thus $A$, $B$
have the same termination productions. Since $T_A'$ contains one
termination production of $A$, $r_{A1}=r_{B1}=r_{A}'$. Similarly
$r_{A2}=r_{B2}=r_{B}'$. If $r_A'\not= r_B'$, then $F_1=\{T_{A1}^t,
T_{B1}^t\}\prec F_2=\{T^{t'}_A\}$, which contradicts the
hypothesis. If $r_A'=r_B'$, then $F_1=\{T_{A1}^t,T_{A2}^t,
T_{B1}^t,T_{B2}^t\}\prec F_2=\{T^{t'}_A,T^{t'}_B\}$, which contradicts
the hypothesis.
\end{proof}

\begin{corollary}\label{cor:linear-recur}
Given a linear-recursive grammar, and any forests  $F,F'$, if $F,F'$
are of the same production and $\exists S'\subseteq paths(F')$
s.t. $paths(F)\prec S'$,  then $|F|<|F'|$.
\end{corollary}
\begin{proof}
Note that the proof of Lemma~\ref{lem:linear_size} uses only ``all
paths in $F$ must subsume paths in $F'$''.
\end{proof}

\subsection{Identifying top-$k$ representative parse trees}
\label{sec:pres:algo}

We now describe an algorithm for identifying top-$k$ \rpts of grammar
$G$ that generate $Q$.  This is a bottom-up algorithm that
progressively builds \rpts of height at most $i$ by combining \rpts of
lower height.  Correctness of such a procedure is based on
Lemma~\ref{lm:bottomup}.

Consider first a naive bottom-up algorithm, which first builds all
possible parse trees of height $i$ that generate $Q$, and then removes
subsumed trees using pair-wise subsumption checks.  The algorithm
stops when no new \rpts are found.  This algorithm, while correct by
Lemmas~\ref{lm:bottomup} and~\ref{lm:height}, will be very expensive,
as there may be exponentially many \rpts for a given grammar, which
would all have to be retained until fixpoint, and against which all
newly generated trees would need to be checked for subsumption.  Yet,
since our goal is to find only the top-$k$ highest-scoring \rpts, most
of these would be discarded at the end of the run.

Thus, to control the running time and the space overhead, we designed
an algorithm that keeps a bounded number of \rpts in memory.  As
another naive approach, consider an algorithm that keeps up to a fixed
number of highest-scoring \rpts found so far in a buffer.  When a new
tree $T$ is constructed, the algorithm checks whether any \rpt in the
buffer subsumes it and, if not, assumes that the new tree is an
\rpt.  This algorithm is straight-forward, but it may return 
trees that are not representative.  This will happen if there exists a
tree $T' \prec T$, yet $T'$ was not retained in the buffer from the
previous round.  Fortunately, we can use Theorem~\ref{the:linear-recur},
stating that a tree can only be subsumed by a tree with a higher
probability score, to devise an algorithm that is both correct and
uses bounded buffers.  This is Algorithm~\ref{algo:repTrees}, which we
now describe.

The algorithm uses the following data structures.  Denote by
$T=\langle r, T_1\ldots T_n\rangle$ a parse tree rooted at production
$r$ with $T_1\ldots T_n$ as subtrees.  Also denote by ${\cal
  T}_i(M,X)$ the \rpts of height $\leq i$ rooted at $M$ and generating
$X \subseteq Q$.  ${\cal T}_i(M,X)$ is an ordered list of \rpts sorted
by decreasing score. The size of ${\cal T}_i(M,X)$ is bounded by some
constant $c$, and we refer to this data structure as the bounded
buffer.

We associate with ${\cal T}_i(M,X)$ a boolean function \\$isTrunc_i(M,X)$,
indicating whether the list ${\cal T}_i(M,X)$ was truncated to accommodate
bounded size.  Importantly, we also associate with ${\cal T}_i(M,X)$ a score
lower-bound, denoted $LB_i(M,X)$, set as follows:

\begin{equation*}
\small
LB_i(M,X) =
\begin{cases} 
0 & \neg isTrunc_i(M,X) \\
MIN_{T \in {\cal T}_i(M,X)} \rho(T) & isTrunc_i(M,X)
\end{cases} 
\end{equation*}

$LB_i(M,X)$ represents the lowest score of any parse tree rooted at
$M$ and generating $X$ for which we can confidently state whether it
is subsumed by any \rpt currently in ${\cal T}_i(M,X)$.  Intuitively,
if no truncation took place, then we can check all trees for
subsumption ($LB_i(M,X)=0$).  If some \rpts were not retained, then we
can only check for subsumption of trees that have a higher score than
the lowest-scoring \rpt in the buffer ($LB_i(M,X)=MIN_{T \in {\cal
    T}_i(M,X)} \rho(T)$).

\begin{algorithm}[t]
\small
\DontPrintSemicolon
\KwIn{a grammar $G=(\Sigma,\Delta,S,R)$, a query $Q\subseteq \Delta$\\ $k$, buffer size $c$\\}
\KwOut{a set of top-$k$ \rpts}
\Begin{
\nl\ForEach{$M\in\Delta,X\subseteq Q$}{
\nl  \lIf{$X=\{M\}\cap Q$}{${\cal T}_0(M,X)\gets\{\langle M\rangle\}$}\\
\nl  \lElse{${\cal T}_0(M,X)\gets\emptyset$}\\
}
\nl \lForEach{$M\in\Sigma/\Delta,X\subseteq Q$}{${\cal T}_0(M,X)\gets\emptyset$}\\
\nl$i\gets 0$\\
\nl {\bf L}: $i\gets i+1$\\
\nl\ForEach{$M\in\Sigma,X\subseteq Q$}{\nl${\cal T}_i(M,X)\gets {\cal T}_{i-1}(M,X)$}
\nl\ForEach{$M\in\Sigma/\Delta$}{\nl$findNewTrees(M,i,c)$}
\nl\If{$\exists M\in\Sigma,X\subseteq Q,{\cal T}_i(M,X)\not={\cal T}_{i-1}(M,X)$}{\nl goto L ~~\tcp{fixpoint}}
\nl \lIf{$|{\cal T}_i(S,Q)|\geq k$}{\Return{${\cal T}_i(S,Q).subList(0,k)$}}\\
\nl  \lElseIf{ $!isTrunc_i(S,Q)$}{\Return{${\cal T}_i(S,Q)$}}\\
\nl  \lElse{$topKRep(G,Q,k,2*c)$}
}
\caption{$topKRep$}
\label{algo:repTrees}
\vspace{-0.2cm}
\end{algorithm}

Algorithm~\ref{algo:repTrees} ($topKRep$) finds the top-$k$ \rpts for
$G$ matching $Q$.  The most interesting part of the algorithm is in
line 10, in the call to procedure $findNewTrees(M,i,C)$.  We omit
algorithmic details of this procedure due to lack of space, and
describe it in text.

Procedure $findNewTrees(M,i,C)$ identifies new \rpts of height up to
$i$ rooted at $M$ generating $X \subseteq Q$.  For a given $X$, we
first construct candidate trees by considering all productions $r \in
productions(M)$.  A production can generate $X$ in multiple ways, by
combining different sets $X_1 \cup \ldots \cup X_n = X$.  Each
combination yields several parse trees, and we can find the
top-scoring trees among them.  

Consider for example production $r: M \rightarrow \{A, B\}$, and
suppose that $A$ can generate $\{a\}$, while $B$ can generate $\{b\}$
or $\{a, b\}$.  Then this production can generate $\{a, b\}$ in two
ways: as  $\{a \} \cup \{b\}$ or $\{a \} \cup \{a, b\}$.

Suppose now that, to generate $ \{a, b\} = \{a \} \cup \{b\}$ we
combine $T^A_1 \in {\cal T}_{i-1}(A,\{a \})$ with $T^B_1 \in {\cal
  T}_{i-1}(B,\{b \})$, deriving a tree $T^{AB}_1 = \langle r,
  T^A_1, T^B_1 \rangle$ with probability $\rho_1$.  Alternatively,
  we may use the combination $ \{a, b\} = \{a \} \cup \{a, b\}$,
  combining $T^A_1 \in {\cal T}_{i-1}(A,\{a \})$ with $T^B_2 \in
  {\cal T}_{i-1}(B,\{a, b \})$, deriving a tree $T^{AB}_2\langle r,
    T^A_1, T^B_2 \rangle$ with probability $\rho_2$.

    It is not guaranteed that $T_1^{AB}$ and $T_2^{AB}$ are the
    top-2 trees for $M$ generating $\{a, b\}$.  This is because ${\cal
      T}_{i-1}(A,\{a \})$ may have been truncated.  For example, we
    may have removed $T_2^A \in {\cal T}_{i-1}(A,\{a \})$, which, if
    used to construct $T_3^{AB} = \langle r, T_2^A, T_1^B
    \rangle$, would have probability $\rho_3 > \rho_2$.  We could be
    sure that no such tree $T_3^{AB}$ exists if either ${\cal
      T}_{i-1}(A, \{a \})$, ${\cal T}_{i-1}(B, \{b \})$, ${\cal
      T}_{i-1}(B, \{a,b \})$ were not truncated, or if the new tree
    had a higher score than $\rho(r) * MAX(LB_{i-1}(A, \{a \}) *
    LB_{i-1}(B, \{b \}), (LB_{i-1}(A, \{a \}) * LB_{i-1}(B, \{a, b
    \}))$.

    Similar reasoning is used when multiple productions are combined
    to generate ${\cal T}_i(M,X)$.

\eat{
We use the following performance optimizations in our implementation.
To reduce space overhead, we use a compact representation of parse
trees, which are represented by bags of paths.  We observed that many
paths will be common to multiple trees.  We thus maintain a pool of
paths in a hash table, and refer to them by id when constructing a
parse tree.  To reduce the running time of subsumption checks, we keep
results of pairwise path subsumption checks.
}

We note that that we implemented a more efficient version of $topKRep$
for non-recursive grammars.  Recall from Theorem~\ref{the:non-recur} that
all parse trees of a non-recursive grammar are representative.  We can
directly construct the highest-scoring parse trees by combining
highest-scoring subtrees.  No subsumption checks are required in the
process.  This algorithm is straight-forward, and its details are
omitted.

An interesting point to note is that the top-$k$ \rpts in
non-recursive grammars are made up of subtrees that are themselves
top-$k$ \rpts.  This is not necessarily the case for recursive
grammars, which makes the $topKRep$ algorithm less efficient in the
recursive case.

\eat{ We first define three notions $x$, $topRep$, $y$.  $x$ is a
  function $x:\Sigma\times2^Q\times\mathbb{N}_{\geq
  0}\rightarrow \mathbb{N}_{\geq 0}$. For $M\in\Sigma,X\subseteq
  Q,i\in\mathbb{N}_{\geq 0}$, we denote by $topRep(M,X,i)$ the set of
  top $x(M,X,i)$ representative parse trees of $M$ that generate $X$
  and that are of height $\leq i$.  $ y$ is a function
  $y: \Sigma\times2^Q\times\mathbb{N}_{\geq 0}\rightarrow \{true,
  false\}$ where $y(M,X,i)=true$ iff $topRep(M,X,i)$ consists
  of \emph{all} representative parse trees of $M$ that generate $X$
  and that are of height $\leq i$.

We override $x$, $topRep$, $y$ for
productions. $x:R\times2^Q\times\mathbb{N}_{\geq 0}\rightarrow
\mathbb{N}_{\geq 0}$. For $r\in R,X\subseteq Q,i\in\mathbb{N}_{\geq
  0}$, denote by $topRep(r,X,i)$ the set of top $x(r,X,i)$
representative parse trees rooted at $r$ that generate $X$ and that
are of height $= i$. $y: R\times2^Q\times\mathbb{N}_{\geq
  0}\rightarrow \{true, false\}$ where $y(r,X,i)=true$ iff
$topRep(r,X,i)$ consists of \emph{all} representative parse trees
rooted at $r$ that generate $X$ and that are of height $= i$.  }



\eat{
\begin{procedure*}[ht]
\nl\ForEach{$X\subseteq Q$,}{
\nl  \ForEach{combination ${\cal C}: [X_1\ldots X_n]$ of $X$~~\tcp{i.e.  $\bigcup\limits_{j=1}^n X_j=X$}}{
\nl  $y(r,X,i)\gets true$\\
\nl  $newTrees({\cal C})\gets\{T=\langle r,T_1\ldots,T_n\rangle| T_j\in topRep(\alpha_j,X_j,i-1) (j=1\ldots n), height(T)=i\}$\\
\nl  \lIf{$\forall j=1 \to n,~y(\alpha_j,X_j,i-1)=true$ }{$y({\cal C})=true$}
\nl  \lElse{$y({\cal C})=false;~y(r,X,i)=false$}
  }
\nl  $allNewTrees\gets \bigcup\limits_{{\cal C}\text{ of } X}newTrees({\cal C})$\\
  \tcp{step 1:pick up trees that we are able to check parse subsumptions}
\nl  \lIf{$y(r,X,i)=true$}{$lowerbound\gets 0.0$}\\
\nl  \lElse{ $lowerbound\gets \max\limits_{{\cal C}\text{ of }X, y({\cal C})=false}~\min\limits_{T\in newTrees({\cal C})}\rho(T)$}\\
\tcc{now, $\{T|T\in allNewTrees,\rho(T)> lowerbound\}$ consists of trees rooted at $r$ that generate $X$ and that are of height$=i$ and probabilities of which are $>lowerbound$}
\nl  \lIf{$y(M,X,i-1)=false$}{$lowerbound\gets \max(lowerbound,\min\limits_{T\in topRep(M,X,i-1)}\rho(T))$;~ $y(r,X,i)\gets false$}\\
\tcc{if $y(M,X,i-1)=false$ then for trees of probabilities $<lowerbound$, we are not able to check if they are representative }
  \tcp{step 2:find representative trees}
 \nl \ForEach{$T\in allNewTrees$}{
   
\nl	\lIf{$\rho(T)< lowerbound$}{continue}\\
\nl	\If{$\not\exists T'\in topRep(M,X,i-1)\cup allNewTrees,~ T'\prec T$} {\nl$topRep(M,X,i).add(T)$;~$x(M,X,i)\gets topRep(M,X,i).size$}
  }
}
 \caption{findTopRepTrees($r:M\rightarrow\alpha_1\ldots\alpha_n$,$i$) //$\forall X\subseteq Q$, compute $x(r,X,i)$, $topRep(r,X,i)$ and $y(r,X,i)$}
\end{procedure*}
}

\eat{
\begin{procedure}[t]
\tcp{compute $x(\cdot)$, $topRep(\cdot)$ and $y(\cdot)$ for the given $M,i$, $\forall X\subseteq Q$}
\tcp{step 1:We fist generate a set of rep trees}
\nl\lForEach{$r\in productions(M)$}{\findTopRepTrees($r,i$)}\\
\nl\ForEach{$X\subseteq Q$}{
\nl$newRepTrees\gets\bigcup\limits_{r\in productions(M)}topRep(r,X,i) \cup topRep(M,X,i-1)$\\
\nl$sortByProbabilities(newRepTrees)$\tcp{descending}
\tcp{step 2:We then find top $x$ ones}
\nl$y(M,X,i)\gets \bigwedge\limits_{r\in productions(M)}y(r,X,i)$\\
\nl\If{$y(M,X,i)=true$} {
\nl\If{$newRepTrees.size> c$}{\nl$x(M,X,i)\gets c;~ y(M,X,i)\gets false$}\nl\lElse{$x(M,X,i)\gets newRepTrees.size$}\\
  goto L
}
\nl  $lowberbound\gets\max\limits_{r\in productions(M)}~\min\limits_{T\in topRep(r,X,i)}\rho(T)$\\
\nl  $x(M,X,i)\gets 0$\\
\nl  \lWhile{$\rho(newRepTrees[x(M,X,i)])\geq lowerbound$}{\\\nl $x(M,X,i)\gets x(M,X,i)+1$}\\
\nl{\bf L}: $topRep(M,X,i)\gets newRepTrees.subList(0,x(M,X,i))$
}
\caption{findTopRepTrees($M$,$i$,$c$) //For symbols }
\end{procedure}
}

\section{Experimental Evaluation}
\label{sec:experiment}
\subsection{Experimental setup}
All experiments were implemented in
Java 6 and performed on a local PC with Intel Core i7 3.4GHz CPU and
4G memory running Linux.  Experiments were executed against
memory-resident data structures.  All reported running times are
averages of 5 executions per setting.

\begin{figure}[t!]
  \centering
\includegraphics[width=0.35\textwidth]{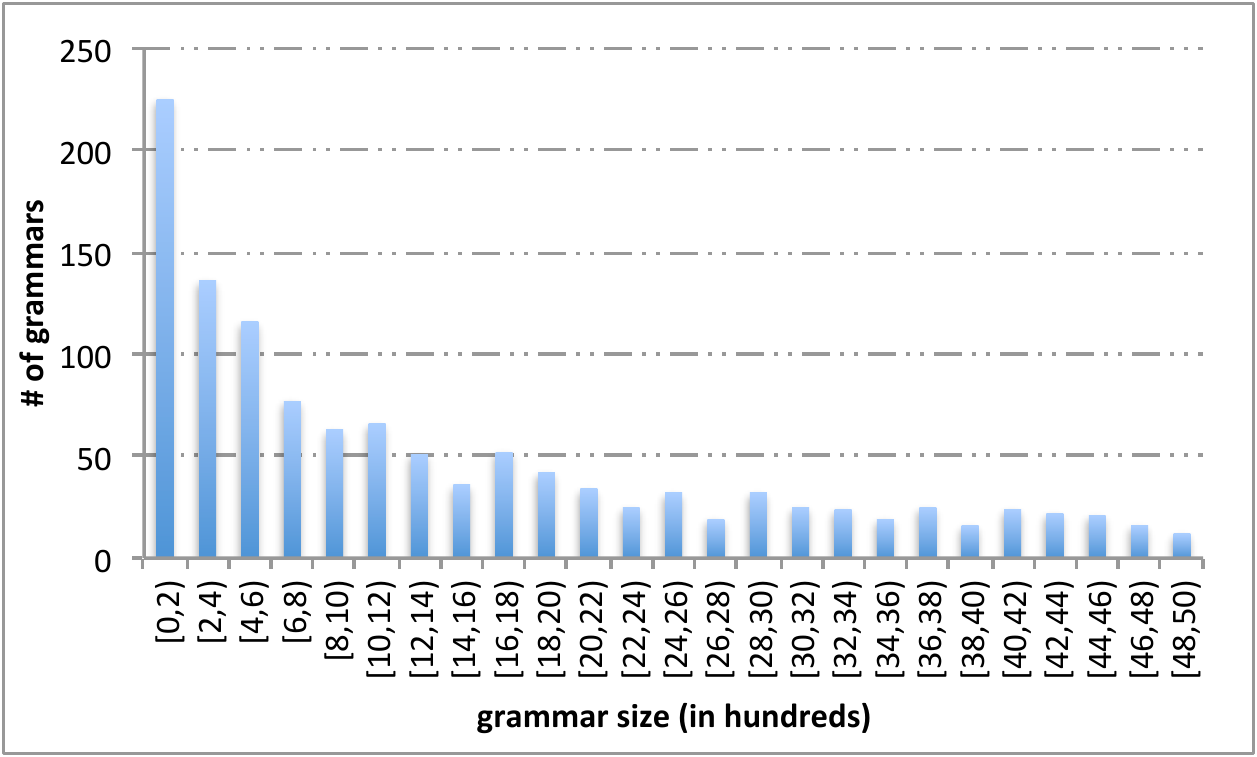}
   \caption{Distribution of grammar sizes in the repository.}
   \label{fig:size}
\vspace{-0.4cm}
\end{figure}

\textbf{Dataset.} We implemented a workflow generator that creates a
repository containing a mix of recursive and non-recursive grammars,
of which some are stand-alone, while others reuse existing workflows
as modules.  Repository size, workflow characteristics, and the amount
of reuse are specified as generator parameters.  All experiments in
this section were executed with the following parameter settings.  A
simple workflow has at most 5 modules; a given module has a
probability of 0.6 to be composite, and a probability of 0.4 to occur
multiple times within the workflow.  Each composite module has at most
3 productions, each simple workflow has a probability of $0.5$ to be
recursive, and each grammar reuses at most $5$ other grammars.

\eat{
Using these parameter settings, we generated a repository consisting
of about 1,700 grammars, and containing a total of about 400,000
distinct composite modules, 1,400,000 distinct atomic modules and
900,000 distinct simple workflows.  The distribution of grammar sizes
in the repository is shown in Figure~\ref{fig:size}. The size of a
grammar ranges from $10$ to $25,000$, and most grammars have size
smaller than $10,000$.  (Note that grammar size is defined as the
total number of symbols in its productions.)}

Using these parameter settings, we generated a repository consisting
of about 1,200 grammars.  The distribution of grammar sizes
in the repository is shown in Figure~\ref{fig:size}. The size of a
grammar ranges from $10$ to $5,000$, and most grammars have size
smaller than $1,000$.  (Note that grammar size is defined as the
total number of symbols in its productions.)

Our choice of parameters is based on our analysis of \myExp, the
largest public repository of scientific workflows, and
on~\cite{DBLP:conf/ssdbm/StarlingerBL12}, where it was observed that
most current workflows are small. \eat{, and that reuse of modules is
not very common.  However, reuse is likely limited by capabilities of
current workflow management systems and repositories, and so our
dataset has a mix or small and large workflows, with varying degree of
reuse.}

\eat{$1,694$ grammars, containing a total of $390,585$ distinct
  composite modules, $1,407,964$ distinct atomic modules and $879,091$
  distinct simple workflows.  The size of a grammar is defined as the
  sum of the sizes of its productions, where the size of a production
  is one plus the length of its right-hand side.}


Next, we generated keyword annotations for the workflows in the
repository using results of keyword co-occurrence analysis
of~\cite{Stoyanovich:2010:ERS:1833398.1833405}.  This analysis was
based on \myExp, where users tag workflows in support of keyword
search. In~\cite{Stoyanovich:2010:ERS:1833398.1833405} we used topic
mining to extract 20 topics from the repository, with each topic
defining a probability distribution over the tags.  Here, we take 20
most frequent keywords per topic, and use their probabilities to
achieve a realistic keyword assignment to workflow modules.  Given a
workflow, the repository generator first randomly chooses a topic, and
then assigns at most $3$ keywords to each module in accordance with
the topic's probability distribution.\eat{ Note that keyword
  consistency is retained, i.e., if a module is shared by more than
  one grammar, all occurrences of the module are assigned the same
  keywords.}

\begin{figure*}[t]
\begin{subfigure}[t]{0.32\textwidth}
 \includegraphics[width=0.99\textwidth]{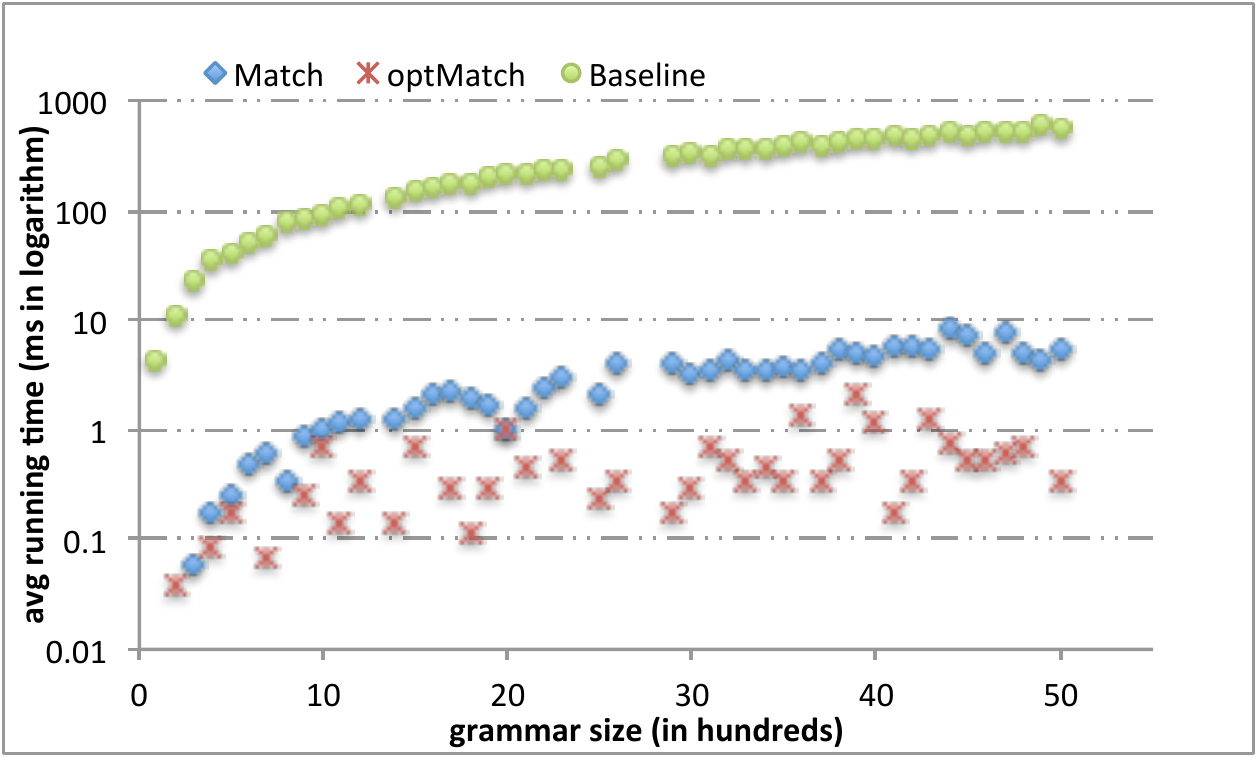}  \caption{{\em Average} running time of $Match$ for $Q_2$. 
    Point $(x,y)$ represents running time over
    grammars of size in $[x-100,x)$.} \label{fig:algo1avg1000}
\end{subfigure}
\begin{subfigure}[t]{0.32\textwidth}
    \includegraphics[width=0.99\textwidth]{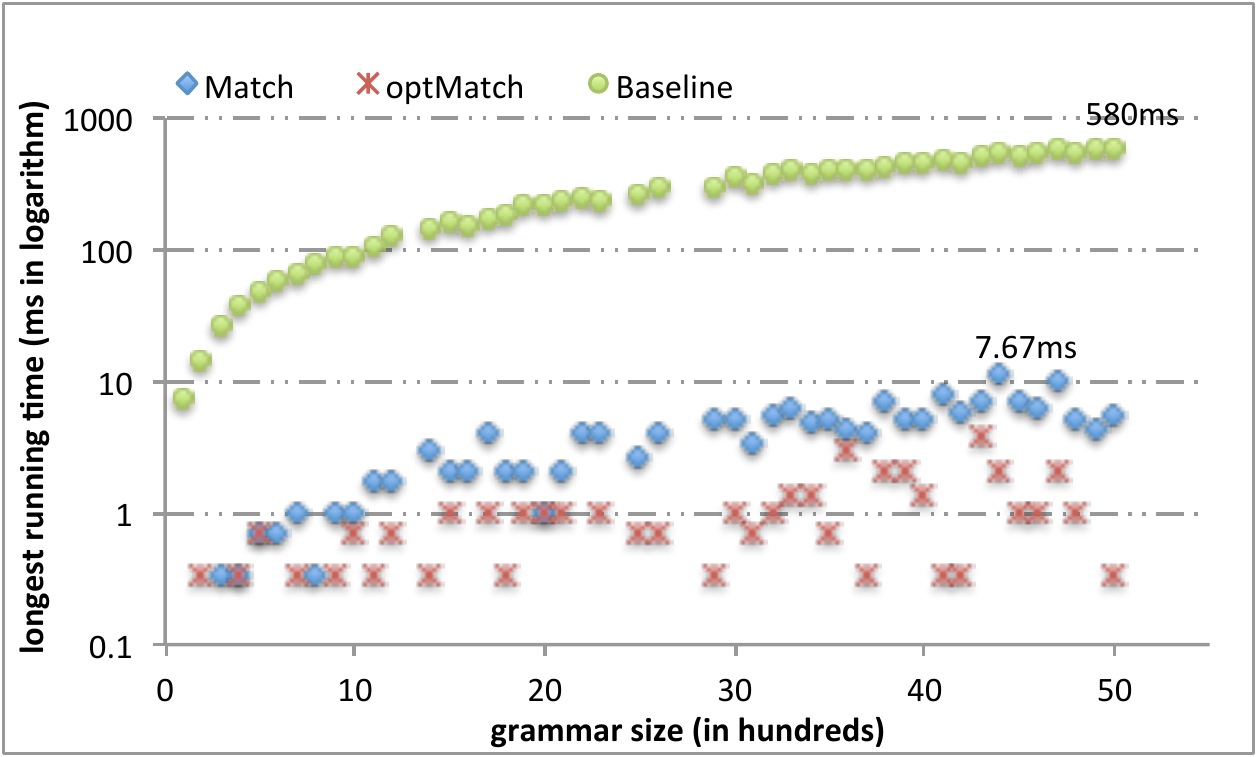}\caption{{\em Longest} running time $Match$ for $Q_2$.  Point $(x,y)$ represents running
time over
    grammars of size in $[x-100,x)$.}\label{fig:algo1longest}
   \end{subfigure}
\begin{subfigure}[t]{0.32\textwidth}
    \includegraphics[width=0.99\textwidth]{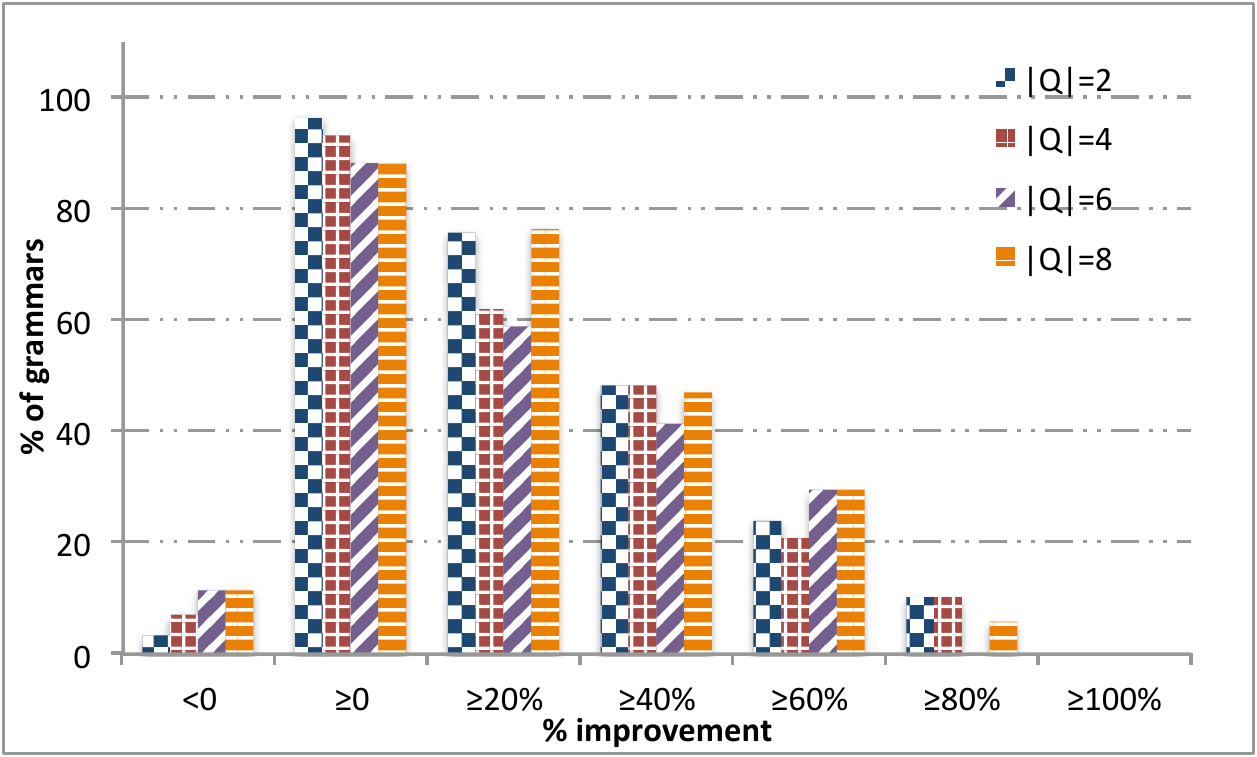}\caption{Improvement of $OptMatch$ over $Match$.}\label{fig:OAFvsAF}
\end{subfigure}  
 \vspace{-0.3cm}
  \caption{Performance of $Match$ and $OptMatch$.}
  \label{fig:algo1}
\end{figure*}

\eat{
\begin{figure*}[t]
  \centering
  \subfloat[{\em Average } running time of $Score$ for $Q_2$. Point $(x,y)$ represents running time over grammars of size in $[x-100,
x)$.]{\label{fig:avgalgo4}\includegraphics[width=0.33\textwidth]{Score_avg.pdf}}
  ~~~
  \subfloat[{\em Longest } running time of $Score$ for $Q_2$. Point $(x,y)$ represents running time over grammars of size in $[x-100,
x)$.]{\label{fig:worstalgo4}\includegraphics[width=0.33\textwidth]{Score_worst.pdf}}
  ~~~
  \subfloat[Improvement of $OptScore$ over $Score$.]{\label{fig:algo5vsalgo4}\includegraphics[width=0.33\textwidth]{Score-OptScore.pdf}}
  \caption{Performance of $Score$  and $OptScore$. \xiaocheng{We can remove Figure~\ref{fig:avgalgo4} and Figure~\ref{fig:algo5vsalgo4}}}
\vspace{-0.5cm}
  \label{fig:score}
\end{figure*}}

\eat{
\begin{figure}[ht]
\centering
\includegraphics[width=0.4\textwidth]{searchREPO.pdf}
\caption{Performance of matching a repository}
\label{fig:repo}
\end{figure}
}

\eat{
Finally, we transformed keyword-annotated grammars into
bag grammars with keywords as terminals. This increased grammar size
by a factor of 2.33 on average.  In the following sections, we report
results in terms of original grammar size.}

\textbf{Queries.} We experimented with many different queries,
generated by first randomly choosing a topic, and then drawing between
2 and 8 keywords according to the topic's probability distribution.
Due to space constraints, we show only representative results.  Unless
otherwise noted, all experiments use three queries described
below.\eat{the queries shown in Table~\ref{tab:queries}.}

$Q_1 = \{text\; mining, e\mbox{-}lico, workflow \;component\}$
consists of 3 most frequent keywords from its topic, and retrieves
workflows with text mining components, contributed by members of
$e\mbox{-}lico$ (an e-laboratory for collaborative research). $Q_2 =
\{TerMine, text\; encoding, input\}$ looks for possible inputs to
$TerMine$ (a web demonstration service). $Q_3 = \{input, xml \;
invalid, read\; file\}$ is a more technical query.  In experiments
that focus on scalability, we work with 7 additional queries that
contain up to 8 keywords, and where larger queries are supersets of
smaller queries.

\subsection{Keyword query match}

We now evaluate the performance of algorithms described in
Section~\ref{sec:algo}. We show that $Match$
(Algorithm~\ref{algo:AF}) runs in time polynomial in the size of the
grammar, and that the $OptMatch$ optimization
(Algorithm~\ref{algo:OAF}) is effective for the vast majority of
queries. We then discuss the effectiveness of the sharing
optimization (Section~\ref{subsec:searchrepo}).

We first evaluate the performance of $Match$ compared with the general
method that intersects a given grammar $G$ with a query $Q$
represented by a finite state automaton~\cite{1964} or a graph chain
pattern~\cite{DBLP:journals/jcss/DeutchM12}.  An adaption of the
method to our scenario (called $Baseline$ in Figure~\ref{fig:algo1})
works as follows.  First, we transform grammar $G$ to grammar $G'$,
where each production has at most two symbols on the right-hand side.
Having the grammar in this form guarantees quadratic data complexity
of the algorithm, and is done off-line.  Next, we intersect $G'$ with
$Q$ to construct a new grammar $G''$, where 1) for each production
$M\rightarrow \alpha_1\alpha_2$ in $G'$,  add a production
$(M,X)\rightarrow (\alpha_1,X_1)(\alpha_2,X_2)$ to $G''$, for each
$X\subseteq Q$, $X_1\cup X_2=X$, making $(M,X), ~(\alpha_1,
X_1),~(\alpha_2,X_2)$ symbols in $G''$; and 2) for each terminal
$M$ in $G'$, mark the symbol $(M,\{M\}\cap Q)$ in $G''$ as a terminal.
Having constructed $G''$, the algorithm checks whether its language is
empty in time linear in grammar size \cite{DBLP:books/daglib/0011126}.  $G$ matches $Q$ iff the language of $G''$ is not empty.

\begin{figure*}[t]
  \centering
\begin{subfigure}[t]{0.32\textwidth}
   \includegraphics[width=0.99\textwidth]{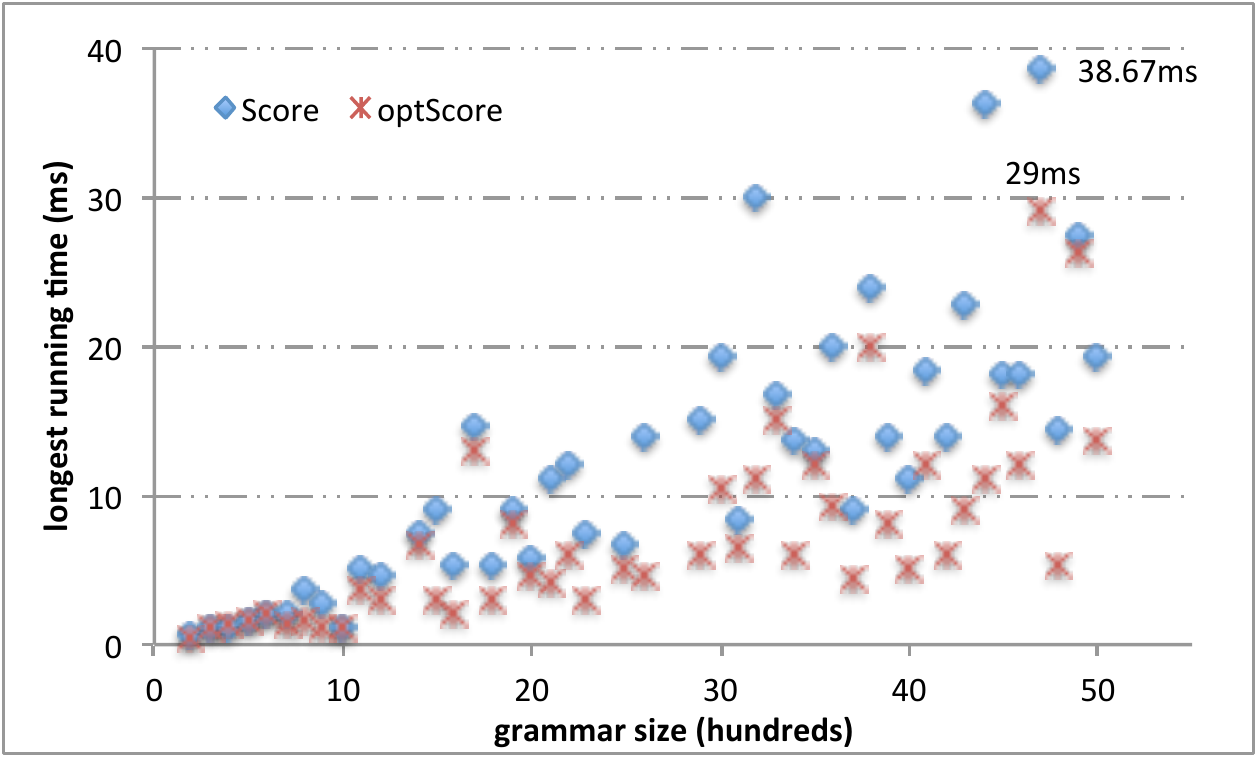}\caption{Longest running time of $Score$ and $optScore$ for
$Q_2$}\label{fig:ScoreVSoptScore_Q2}
\end{subfigure}
\begin{subfigure}[t]{0.32\textwidth}
  \includegraphics[width=0.99\textwidth]{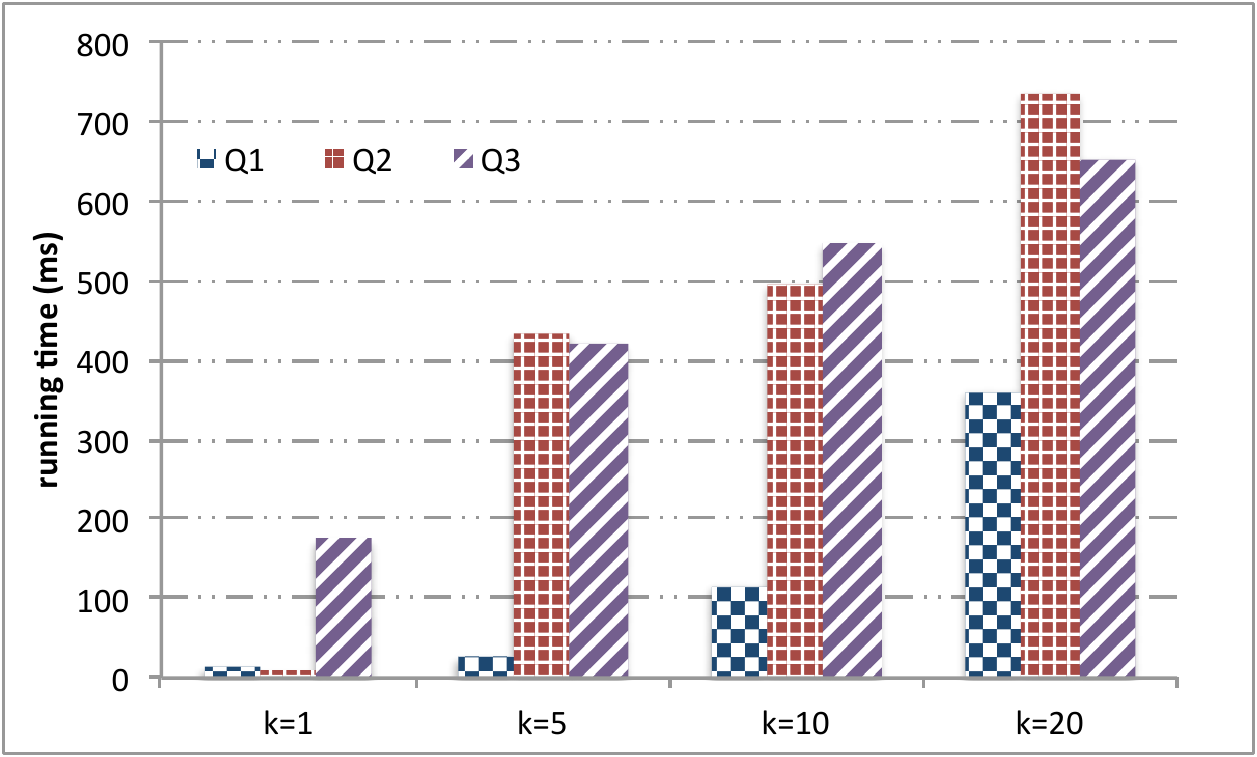}\caption{Running time of TA (ms).}\label{fig:topk-k}
\end{subfigure}
\begin{subfigure}[t]{0.32\textwidth}
  \includegraphics[width=0.99\textwidth]{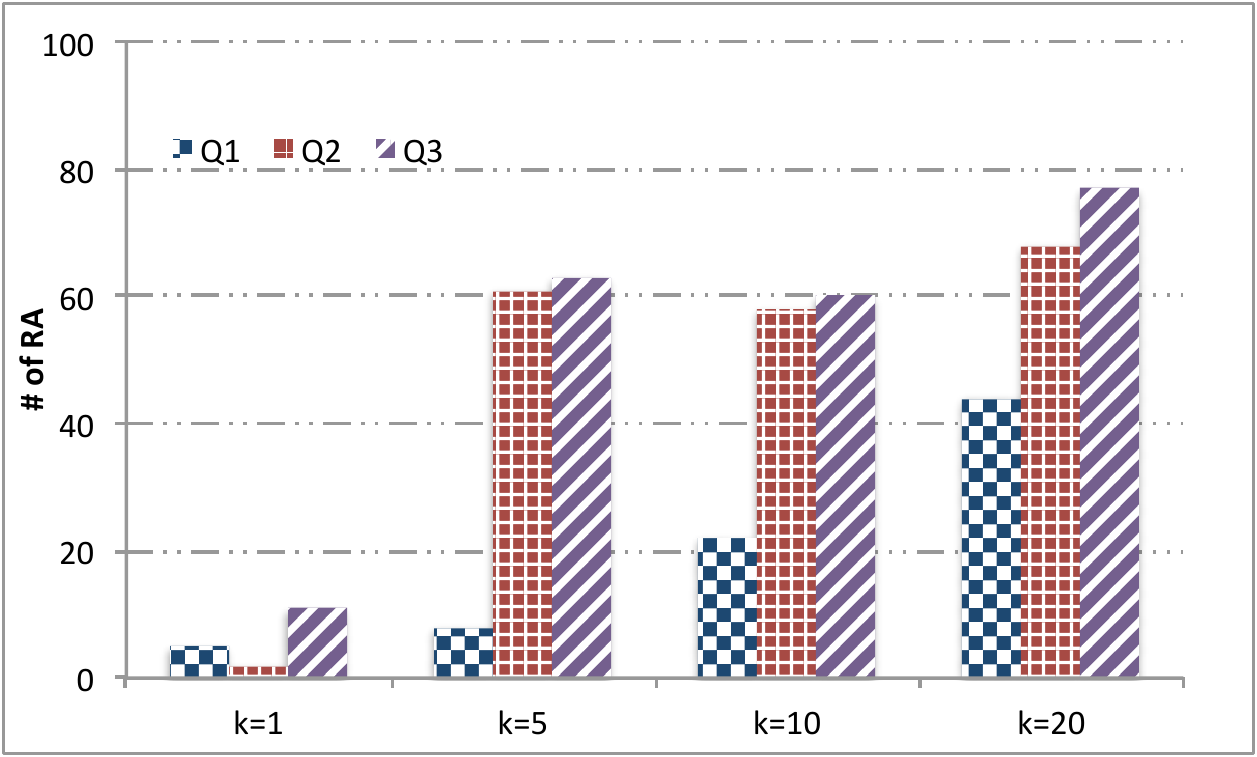}\caption{Running time of TA (RA).}\label{fig:topk-k-cost}
\end{subfigure} 
 \vspace{-0.3cm}
  \caption{
  Performance of the ranking solution: $Score$, $OptScore$ and threshold algorithm (TA).}
\label{fig:algo6}
\end{figure*}

\begin{figure*}[t]
  \centering
\begin{subfigure}{0.32\textwidth}
  \includegraphics[width=0.99\textwidth]{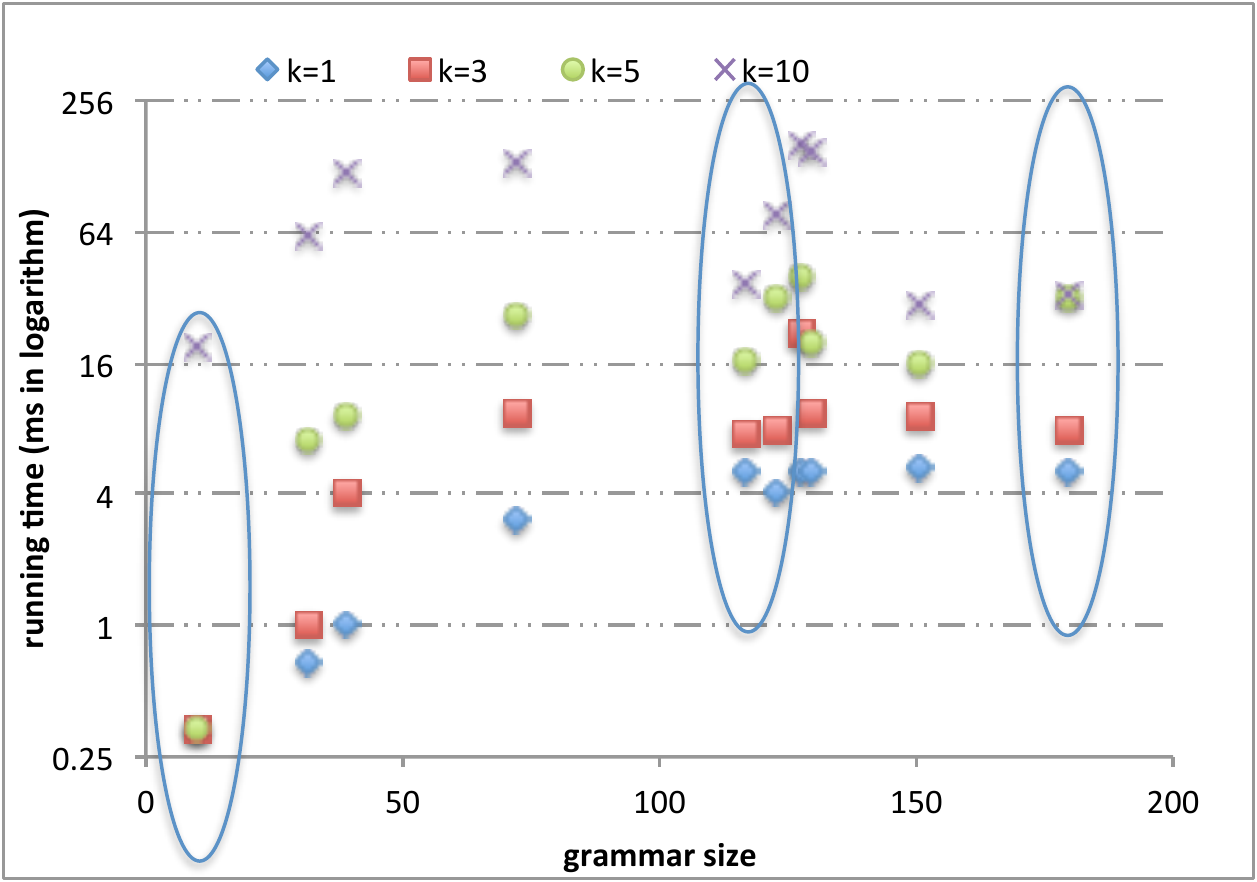}\caption{$Q_1$}\label{fig:Q1topkRep}
 \end{subfigure}%
\begin{subfigure}{0.32\textwidth}
  \includegraphics[width=0.99\textwidth]{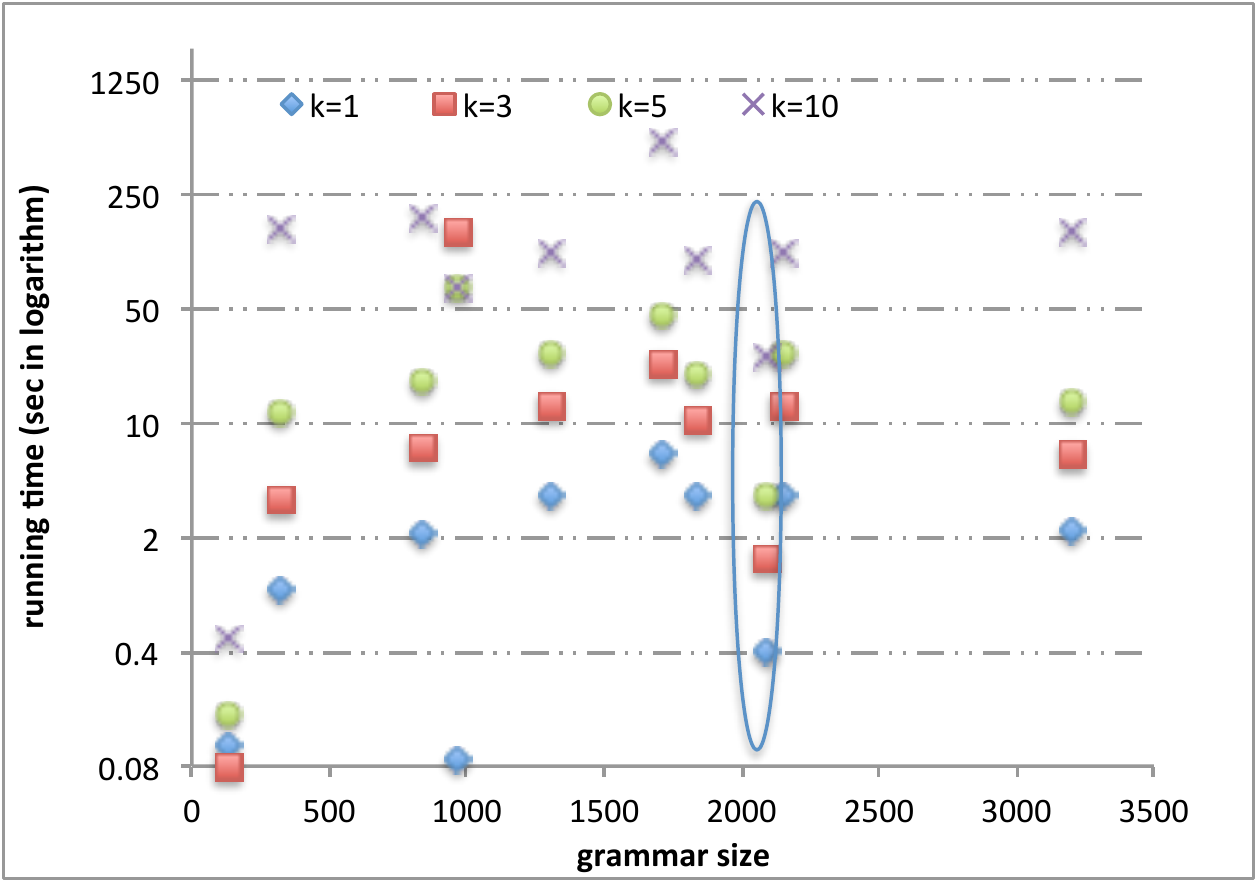}\caption{$Q_3$}\label{fig:Q3topkRep}
\end{subfigure}
\begin{subfigure}{0.32\textwidth}
 \includegraphics[width=0.99\textwidth]{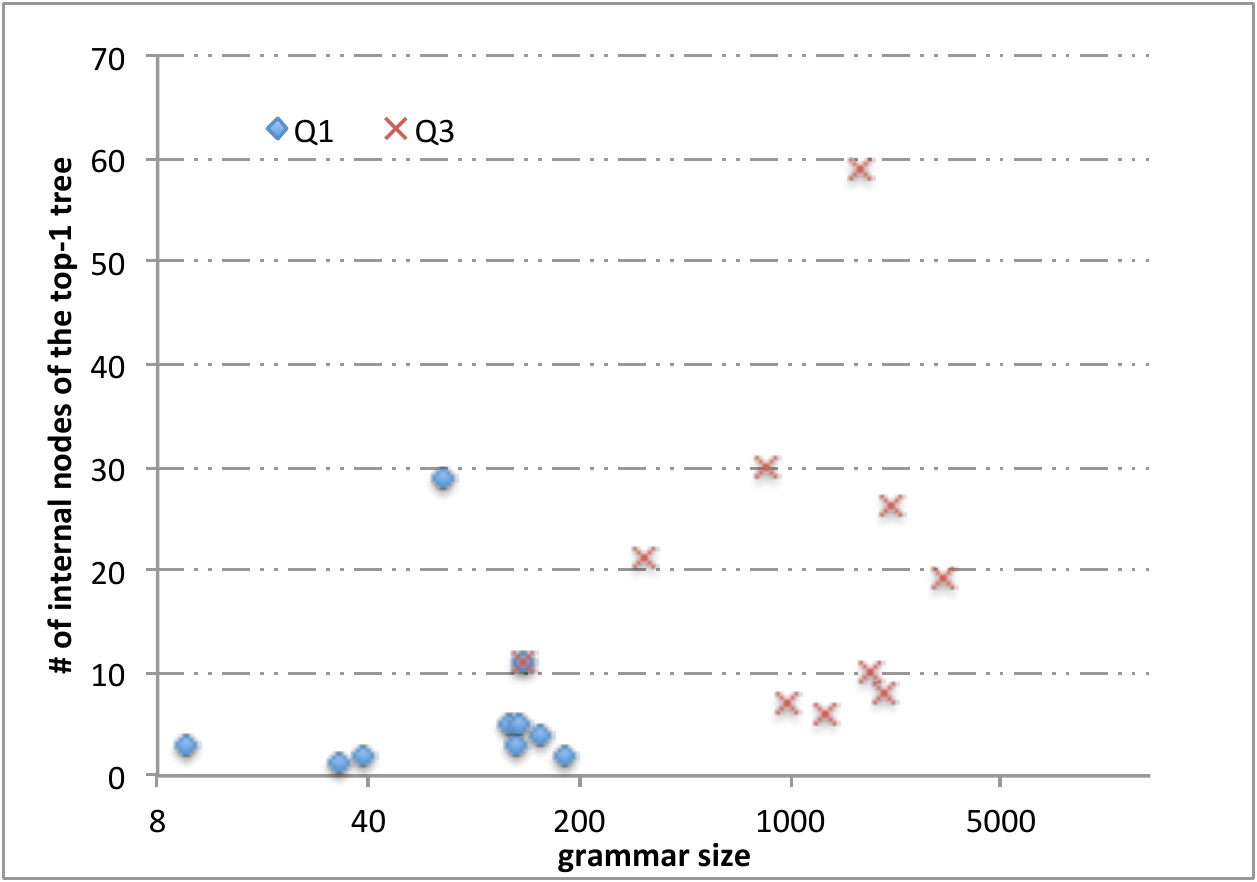} \caption{Size of the top-1 tree}\label{fig:treeSize}
\end{subfigure}
  \vspace{-0.3cm}
  \caption{Performance of finding top-$k$ \rpts for the top-$10$ grammars (ellipses indicate non-recursive grammars).}
  \vspace{-0.5cm}
\label{fig:topkRep}
\end{figure*}

Figures~\ref{fig:algo1avg1000} and~\ref{fig:algo1longest} demonstrate
the average and longest running time of $Match$ and $Baseline$ for
query $Q_2$ for grammars of different sizes.  Some 154 grammars in our
repository contain all keywords of $Q_2$, and we run the algorithms on
these grammars.  According to Figure~\ref{fig:algo1avg1000} and
\ref{fig:algo1longest}, $Match$ runs in time polynomial in grammar
size, as expected.  The running time of the algorithm is reasonable,
and is below 10ms for all grammars. \eat{We also observe that there is
  an irregular point in Figure~\ref{fig:algo1longest}, with a running
  time of $69.33 ms$ for a grammar of size $11,160$.  Recall that
  $Match$ terminates early if a variable other than the start symbol
  matches the query.  Early termination is a common condition,
  particularly for large grammars.  The irregular point corresponds to
  a case where no early termination was possible.}We observed similar
trends for other queries. Although $Baseline$ and $Match$ both run in
time polynomial in grammar size, $Match$ significantly outperforms
$Baseline$ in all cases, because it terminates early if a variable
other than the start symbol matches the query.

Figure~\ref{fig:OAFvsAF} shows that $OptMatch$, which is an optimization of $Match$, is
effective at reducing the running time for most queries.  For example,
$OptMatch$ outperforms $Match$ by at least 20\% for 80\% of 2-keyword
queries.  $OptMatch$ slightly increases running times for some
queries. We also measured the total running time of $Match$ and
$OptMatch$ for a variety of queries, and for all workflows in the
repository.  We found that $OptMatch$ brings an over-all gain of at
least a factor of 2 for queries of size between 2 and 8.  For example,
the total running time of $Match$ for queries of size 2 is 77.67ms,
compared to 43ms for $OptMatch$. 

\eat{ Figure~\ref{fig:repo} shows the total running time of $Match$
  and $OptMatch$, for all workflows in the repository.}  

Finally, we measured the effectiveness of leveraging grammar reuse,
when executing $Match$ and $OptMatch$ on all workflows in the
repository (see Section~\ref{subsec:searchrepo}).  We found that this
optimization, to which we refer as {\em sharing}, is \eat{ effective,
  but we omit results due to lack of space.}extremely effective,
bringing the total running time of $Match$ to between 130ms and 150ms
for queries of size 2 to 8.  {\em Match} and {\em OptMatch} have
comparable performance with this optimization.  Details are omitted
due to lack of space.

{\bf In summary,} $Match$ and $OptMatch$ are efficient algorithms.
$OptMatch$ outperforms $Match$ for most grammars, and should be used
when individual grammars are tested.  Either {\em Match} or {\em
  OptMatch} with the sharing optimization may be used when all
workflows in the repository are tested.


\subsection{Ranking of grammars}
\label{sec:exp:rank}

We now demonstrate that the techniques of
Section~\ref{sec:ranking} can be implemented efficiently.  We first
show that the running time of $Score$ (Algorithm~\ref{algo:AF_score})
is polynomial in grammar size\eat{ (Figure~\ref{fig:avgalgo4})}, and
that $OptScore$ outperforms $Score$ in most
cases\eat{(Figure~\ref{fig:algo5vsalgo4})}. We then show that top-$k$
workflows can be identified efficiently when TA is used to find the
promising grammars\eat{ (Section~\ref{sec:ranking:topK})}.

\eat{ \xiaocheng{eat on Feb. 18}
Figure~\ref{fig:avgalgo4} reports the average running time of $Score$
for query $Q_2$.  Comparing Figure~\ref{fig:avgalgo4} with
Figure~\ref{fig:algo1avg1000}, we note that $Score$ is almost three
times slower than $Match$.  Figure~\ref{fig:worstalgo4} shows the
longest running times of $Score$ for the grammars in our repository,
and demonstrates that the running time of this algorithm is
reasonable, and is below 40ms for all grammars. We observed similar
trends for other queries.}

We observed similar trends for average and longest running time of
$Score (optScore)$ to that of $Match
(optMatch)$. Figure~\ref{fig:ScoreVSoptScore_Q2} reports the longest
running times of $Score$ and $optScore$ for grammars in our repository
and demonstrates that the running time of this algorithm is
reasonable, and is below 40ms for all grammars.  Comparing
Figure~\ref{fig:ScoreVSoptScore_Q2} with
Figure~\ref{fig:algo1longest}, we note that $Score$ is almost three
times slower than $Match$. We also observe that our {\em Score}
algorithm outperforms {\em
  Baseline}~\cite{DBLP:journals/jcss/DeutchM12} (which is used for
matching).  Since a scoring algorithm is necessarily slower than a
matching algorithm, we conclude that a scoring algorithm that uses a
similar framework as {\em Baseline} will be less efficient than {\em
  Score}, and do not run a direct experimental comparison.

We also measured the improvement of $optScore$ over $Score$ and got a trend very similar to those observed for $OptMatch$ (Figure~\ref{fig:OAFvsAF}). $OptScore$ results
in an improvement for the vast majority of grammars, for queries of
varying lengths.  

Figure~\ref{fig:topk-k} reports the running time of Threshold Algorithm
(TA), followed by an execution of $OptScore$ for the promising
grammars, for queries $Q_1, Q_2, Q_3$. We can see from
Figure~\ref{fig:topk-k} that it takes under $100ms$ to find the
top\mbox{-}5 grammars for $Q_1$, and around $400ms$ for $Q_2$ and
$Q_3$.  These queries all match between $150$ and $180$ grammars in our
repository, and, as is usually the case for TA, the difference in
performance is due to the distribution of
scores. Figure~\ref{fig:topk-k-cost} gives the running time of TA in
terms of the number of random accesses (RA), demonstrating that the
stopping condition for TA is reached after only a fraction of all
matching grammars have been considered. 

\eat{ \xiaocheng{this paragraph becomes invalid when we change the repository  Finally,
Figure~\ref{fig:TAvsBase} compares the running time of our top-$k$
solution to baseline search, with $k=10$, for queries of length
between 2 and 8.  Figure~\ref{fig:TAvsBase} illustrates that our
top-$k$ solution, in which TA is used to identify the promising
grammars, and $OptScore$ is executed only for such grammars, strongly
outperforms a baseline solution, in which scores of all matching
grammars in the repository are computed (note the logarithmic scale).}}

{\bf In summary,} $Score$ and $OptScore$ are efficient algorithms, and
$OptScore$ outperforms $Score$ for most grammars.  Using TA to
identify promising grammars, and then invoking $OptScore$ for these
grammars, allows us to achieve interactive response times when
retrieving the top-$k$ grammars from the repository.

\eat{We also note that it takes more than $100ms$ to find the
top-1 related grammar for $Q_3$. There are two reasons. Firstly, TA
does more random accesses for $Q_3$ (leftmost bar of
Figure~\ref{fig:topk-k-cost}) and secondly TA calls $OptScore$ once
for $Q_1,Q_2,Q_3$ in practice while the top-1 related grammar of $Q_3$
is larger than that of $Q_1$ and $Q_2$. } 

\eat{
For the frequent query $Q_1$, the number of random access is much less
than $Q_2, Q_3$. Figure~\ref{fig:TAvsBase} reports a comparison
between baseline search i.e. to compute the score of grammars that
contain all keywords of the query one by one using $OptScore$ and TA
($Top\mbox{-}10$). We can see that TA speeds up the searching time
significantly. Figure~\ref{fig:algo6} validates Section~\ref{};}


\eat{\begin{figure}
\centering
\includegraphics[width=0.4\textwidth]{figures/top-10-querysize.pdf}
\caption{Improvement of top-$k$ over baseline search.}
\label{fig:TAvsBase}
\end{figure}}

\subsection{Result presentation}

Finally, we evaluate the running time and quality of $topKRep$
(Algorithm~\ref{algo:repTrees}), and show that it can be used to find
the highest-scoring representative parse trees in interactive time.

Recall from Section~\ref{sec:resultPresentation} that $topKRep$ is
invoked on a particular grammar, typically one that is among the
highest-scoring grammars for a given query, and computes a fixed
number of \rpts for that grammar that match the query.

\eat{
We first find top-10 grammars for each query using TA algorithm
(Section.\ref{???}) and then for each grammar find top-k
representative matching parse trees w.r.t. the query
(Section.\ref{???}) where $k$ varies from $1$ to $10$. Due to space
constraint, we show only results of $Q_1$, $Q_3$ with
$k\in\{1,3,5,10\}$ (Figure~\ref{fig:topkRep}).}

Figure~\ref{fig:topkRep} reports the total running time of $topKRep$
over top-$10$ grammars for queries $Q_1$ and $Q_3$, as a function of
grammar size.  The number of \rpts, denoted $k$, varies from 1 to 10.
Observe from Figure~\ref{fig:Q1topkRep} that the top-$10$ grammars for
$Q_1$ are small $(<200)$, and that the total running time of $topKRep$
is reasonable, under $161.33ms$ for $k=10$.

We use ellipses to indicate running times for non-recursive grammars.
We noted in Section~\ref{sec:resultPresentation} that, because all
parse trees of non-recursive grammars are representative, we can
design an efficient version of $topKRep$ for this case.  The
difference in running time is not significant for $Q_1$
(Figure~\ref{fig:Q1topkRep}), but becomes more pronounced for $Q_3$
(Figure~\ref{fig:Q3topkRep}).

Figure~\ref{fig:Q3topkRep} shows that $topKRep$ is significantly
slower for $Q_3$ than for $Q_1$, for three reasons.  First, the top-10
grammars for $Q_3$ are much larger, and have larger parse trees.
Figure~\ref{fig:treeSize} shows that sizes of top-1 trees for large
grammars are usually larger than those for small grammars. Second,
there are more parse trees to be constructed for large
grammars. Third, it is more common for large grammars to require a
larger buffer size when computing the top-$k$ \rpts.  Recall that
buffer size is an argument in Algorithm~\ref{algo:repTrees}, and that
the algorithm is re-executed with a larger buffer if the original
setting does not yield enough \rpts.  For the $4^{th}$ grammar in
Figure~\ref{fig:Q3topkRep}, the required buffer size for $k=3$ was
$48$ (meaning that $topKRep$ was executed 5 times).  This was higher
than for $k=5$ and $k=10$, where buffer size of 40 (for 4 and 3
executions of $topKRep$, respectively) was sufficient.

We now present an example that illustrates the effectiveness of
\rpts. For query $Q_1$, the top-1 matching grammar in our repository
is one where the start module is annotated with all the keywords in
$Q_1$. Thus all parse trees of this grammar match $Q_1$. To simplify
presentation, we eliminate the keywords of modules and show only the
grammar below:

\begin{tabular}{ll}
$r_1: S\Rightarrow \{S,B\}~(\frac{1}{3})$&$r_2:S\Rightarrow \{A,A,S,s_1\}~(\frac{1}{3})$\\
$r_3:S\Rightarrow \{s_2\}~(\frac{1}{3})$&$r_4: A\Rightarrow \{B,a_1\}~(\frac{1}{2})$\\
$r_5: A\Rightarrow \{a_2\}~(\frac{1}{2})$&$r_6: B\Rightarrow \{S,b_1\}~(\frac{1}{3})$\\
$r_7:B\Rightarrow \{S,b_2\}~(\frac{1}{3})$&$r_8:B\Rightarrow \{b_3\}~(\frac{1}{3})$
\end{tabular}

Note that the top-6 trees of this grammar are all \rpts. The $7^{th}$
most probable tree $T_7$ (with a probability $0.004$) is shown in
Figure~\ref{fig:top-7tree}. Observe that $T_7$ is subsumed by the top
$2^{nd}$ tree $T_2$ (Figure~\ref{fig:top-2tree}), and so is not an
\rpt.  On the other hand, $T_8$ (with probability $0.003$) is an \rpt,
and is more interesting to show to the user than $T_7$, although its
probability score is lower.  

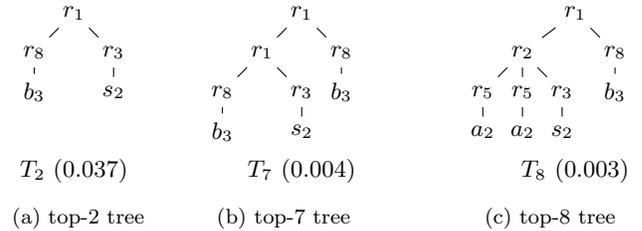
\begin{figure}[h]
\begin{subfigure}{0.1\textwidth}
\begin{tikzpicture}[node distance=15]
\node(r11)  {$r_1$};
\node(o11)[below of=r11, node distance=15]{};
 \node(r8) [left of=o11]{$r_8$};
\node(b3)[below of=r8]{$b_3$};
\node(r3) [right of=o11]{$r_3$};
\node(s2)[below of=r3]{$s_2$};
\draw[-](r11)--(r3);\draw[-](r11)--(r8);\draw[-](r8)--(b3);\draw[-](r3)--(s2);
\node(T)[below of=r11,node distance=60] {$T_2$ $(0.037)$};
\end{tikzpicture}
\caption{top-$2$ tree}
\label{fig:top-2tree}
\end{subfigure}
\begin{subfigure}{0.2\textwidth}
\centering
\begin{tikzpicture}[node distance=15]
\node(r1) {$r_1$};
\node(o)[below of=r1,node distance =15] {};
\node(r11) [left of=o, node distance=15] {$r_1$};
\node(o11)[below of=r11, node distance=15]{};
\node(r8) [left of=o11]{$r_8$};
\node(b3)[below of=r8]{$b_3$};
\node(r3) [right of=o11]{$r_3$};
\node(s2)[below of=r3]{$s_2$};

\node(r81) [right of=o, node distance=15] {$r_8$};
\node(b31)[below of=r81]{$b_3$};
\draw[-](r1)--(r11);\draw[-](r11)--(r3);\draw[-](r11)--(r8);\draw[-](r8)--(b3);\draw[-](r3)--(s2);\draw[-](r1)--(r81);\draw[-](r81)--(b31);
\node(T)[below of=r1,node distance=60] {$T_7$ $(0.004)$};
\end{tikzpicture}
\caption{top-$7$ tree}
\label{fig:top-7tree}
\end{subfigure}%
\begin{subfigure}{0.2\textwidth}
\centering
\begin{tikzpicture}[node distance=15]
\node(r1) {$r_1$};
\node(o)[below of=r1] {};
\node(r2)[left of=o, node distance=20] {$r_2$};
\node(r5)[below of=r2]{$r_5$};
\node(r51)[left of=r5]{$r_5$};
\node(r3)[right of=r5]{$r_3$};
\node(a2)[below of=r5]{$a_2$};
\node(a21)[below of=r51]{$a_2$};
\node(r8)[right of=o]{$r_8$};
\node(b3)[below of=r8]{$b_3$};
\node(s2)[below of=r3]{$s_2$};
\node(T)[below of=r1, node distance=60]{$T_8$ (0.003)};
\draw[-](r1)--(r2);\draw[-](r1)--(r8);\draw[-](r8)--(b3);\draw[-](r2)--(r5);\draw[-](r2)--(r51);\draw[-](r2)--(r3);\draw[-](r5)--(a2);\draw[-](r51)--(a21);\draw[-](r3)--(s2);
\end{tikzpicture}
\caption{top-$8$ tree}
\label{fig:top-7rptree}
\end{subfigure}
\caption{Trees for the top-1 grammar that matches $Q_1$}
  \vspace{-0.5cm}
\end{figure}

{\bf In summary,} $topKRep$ can be used to efficiently compute the
highest-scoring representative parse trees for many grammars.  For
certain large grammars, $topKRep$ does not terminate in interactive
time, due to parse tree size, and to conservative buffer size
requirements.  Performance can be improved using alternative
strategies for setting buffer size.

\section{Conclusions}
\label{sec: concl}
In this paper we addressed the problem of searching a repository of
workflow specifications in which modules, both atomic and composite,
are annotated with keywords.  Since search does not interact with the
graph structure of workflows, we reduced the problem to one of
searching a repository of {\em bag grammars}.  \eat{We showed that the
problem of deciding whether a grammar matches a query is NP-complete
with respect to {\em combined} (data and query) complexity, but has }
We gave an efficient polynomial-time matching algorithm with respect
to {\em data} complexity, and extended this to search over a
repository of bag grammars.  We developed efficient algorithms for
calculating the relevance score of a grammar to a given query, and for
finding the top-$k$ grammars for a given query. Finally, we proposed a
novel result presentation method.\eat{ based on a notion of
representative parse trees.}

This work introduces a novel use of bag grammars, and shows the
importance of probabilistic bag grammars.  Our approach has been based
on efficiency considerations; in the future we would like to gain a
deeper understanding of how to use probabilistic bag grammars and
continue to explore ways of presenting concise search results.  Moving
beyond keyword search, we would like to add structural features into
queries.  We also plan on testing the usability of these ideas on real
datasets.

\section{Acknowledgements}
\label{sec:ack}
We thank Tova Milo for extensive discussions and feedback on this draft.


%
\bibliographystyle{abbrv}
\small
\bibliography{sigproc,paper}  
%
%

\end{document}